\documentclass[11pt]{article}
\usepackage[margin=1in]{geometry}
\usepackage[T1]{fontenc}
\usepackage{hyperref}
\usepackage{svg}
\usepackage{xspace}
\usepackage[colorinlistoftodos]{todonotes}
\usepackage{caption}
\usepackage{algorithm,algpseudocode}
\usepackage{subcaption}
\usepackage{amsmath,amsthm,amssymb,mathtools}
\usepackage{thm-restate}
\usepackage[capitalize]{cleveref}
\usepackage{url}
\usepackage{bookmark}

\newcommand{\kvnnew}[1]{\textcolor{black}{#1}}

\newcommand{\kvn}[1]{#1}
\newcommand{\ari}[1]{#1}
\newcommand{\deb}[1]{#1}
\newcommand{\mt}[1]{#1}

\newcommand{\kvnr}[1]{}
\newcommand{\akr}[1]{}
\newcommand{\djr}[1]{}
\newcommand{\mtr}[1]{}
\newcommand{\kjr}[1]{}

\newcommand{\abs}[1]{\left|#1\right|}               
\newcommand{\ceil}[1]{\left\lceil#1\right\rceil}    

\newcommand{\eps}{\varepsilon}
\renewcommand{\epsilon}{\varepsilon} 

\renewcommand{\tilde}{\widetilde}

\newcommand{\AM}{\mathcal{A}}

\newcommand{\CC}{\mathcal{C}}

\newcommand{\SM}{\mathcal{S}}

\newcommand{\tdk}{\textsc{3DK}}

\newcommand{\tdkr}{\textsc{3DKR}}

\newcommand{\boxes}{containers\xspace}
\newcommand{\Container}{Container~}
\newcommand{\container}{container\xspace}

\newcommand{\NFDH}{\texttt{NFDH}\xspace}
\newcommand{\tNFDH}{\texttt{3D-NFDH}\xspace}

\newtheorem{theorem}{Theorem}[section]
\newtheorem{lemma}[theorem]{Lemma}
\newtheorem{claim}[theorem]{Claim}
\newtheorem{remark}[theorem]{Remark}
\newtheorem{corollary}[theorem]{Corollary}

\theoremstyle{definition}
\newtheorem{definition}{Definition}[section]

\newcommand{\hmax}{\mu}

\newcommand{\alg}{\mathtt{opt}_{\mathrm{gs}}}

\newcommand{\opt}{\mathtt{opt}}

\newcommand{\optgs}{\mathtt{opt}_{\mathrm{gs}}}
\usepackage[colorinlistoftodos]{todonotes}

\newcommand{\lca}{LC\textsuperscript{a}}
\newcommand{\lcb}{LC\textsuperscript{b}}

\newcommand{\HH}{\mathrm{H}}

\newcommand{\LContainer}{$\mathsf{L}$-Container}
\newcommand{\LContainers}{$\mathsf{L}$-Containers}

\newcommand{\optonel}{\mathtt{opt}_{1\ell}}
\newcommand{\optones}{\mathtt{opt}_{1s}}
\newcommand{\opttwol}{\mathtt{opt}_{2\ell}}
\newcommand{\opttwos}{\mathtt{opt}_{2s}}
\newcommand{\optthreel}{\mathtt{opt}_{3\ell}}
\newcommand{\optthrees}{\mathtt{opt}_{3s}}
\newcommand{\honel}{h_{1\ell}}
\newcommand{\vones}{v_{1s}}

\newcommand{\vtwos}{v_{2s}}

\newcommand{\vthrees}{v_{3s}}

\newcommand{\opttwot}{\mathtt{opt}_{2t}}
\newcommand{\opttwoh}{\mathtt{opt}_{2h}}
\newcommand{\optthreet}{\mathtt{opt}_{3t}}
\newcommand{\optthreeh}{\mathtt{opt}_{3h}}
\newcommand{\optlt}{\mathtt{opt}_{Lt}}
\newcommand{\optlh}{\mathtt{opt}_{Lh}}

\newcommand{\OPT}{\mathrm{OPT}}

\newcommand{\OPTonel}{\text{OPT}_{1\ell}}
\newcommand{\OPTL}{\text{OPT}_L}
\title{Improved Approximation Algorithms for\\ Three-Dimensional Knapsack}
\author{Klaus Jansen\footnote{Kiel University, Kiel, Germany, \texttt{kj@informatik.uni-kiel.de}} $\qquad$ Debajyoti Kar\footnote{Indian Institute of Science, Bengaluru, India, \texttt{debajyotikar@iisc.ac.in}} $\qquad$ Arindam Khan\footnote{Indian Institute of Science, Bengaluru, India, \texttt{arindamkhan@iisc.ac.in}}\\[1em]$\qquad$K. V. N. Sreenivas\footnote{Indian Institute of Science, Bengaluru, India, \texttt{venkatanaga@iisc.ac.in}}$\qquad$Malte Tutas\footnote{Kiel University, Kiel, Germany, \texttt{mtu@informatik.uni-kiel.de}}}
\date{\empty}

\begin{document}
\maketitle
\begin{abstract}
We study the three-dimensional Knapsack (3DK) problem, in which we are given a set of axis-aligned cuboids with associated profits and an axis-aligned cube knapsack. The objective is to find a non-overlapping axis-aligned packing (by translation) of the maximum profit subset of cuboids into the cube. 
The previous best approximation algorithm is due to Diedrich,  Harren, Jansen, Th{\"{o}}le, and Thomas (2008), who gave a $(7+\varepsilon)$-approximation algorithm for 3DK and a $(5+\varepsilon)$-approximation algorithm for the variant when the items can be rotated by 90 degrees around any axis, for any constant $\varepsilon>0$. Chleb{\'{\i}}k and Chleb{\'{\i}}kov{\'{a}} (2009) showed that the problem does not admit an asymptotic polynomial-time approximation scheme. 

We provide an improved polynomial-time $(139/29+\varepsilon) \approx 4.794$-approximation algorithm for 3DK and $(30/7+\varepsilon) \approx 4.286$-approximation when rotations by 90 degrees are allowed. We also provide improved approximation algorithms for several variants such as the cardinality case (when all items have the same profit) and uniform profit-density case (when the profit of an item is equal to its volume).
Our key technical contribution is {\em container packing} -- a structured packing in 3D such that all items are assigned into a constant number of containers, and each container is packed using a specific strategy based on its type.
We first show the existence of highly profitable container packings. Thereafter, we show that one can find near-optimal container packing efficiently using a variant of the Generalized Assignment Problem (GAP). 
\end{abstract}

\usetikzlibrary{patterns}
\usetikzlibrary{decorations,arrows}
\usetikzlibrary{decorations.pathmorphing}
\usepgflibrary{decorations.pathreplacing} 
\usetikzlibrary{decorations.text}

\definecolor{orangeborder}{rgb}{0.33,0.33,0.33}
\definecolor{greenItem}{rgb}{0.4,0.4,0.4}
\definecolor{yellowItem}{rgb}{0.2,0.2,0.2}
\definecolor{redItem}{rgb}{0.86,0.86,0.86}
\definecolor{formerLB}{rgb}{0.9,0.9,0.9}

\definecolor{tbBgOdd}{rgb}{0.82,0.82,0.82}

\newcommand{\drawGreenItem}[5][$ $]{\draw[color=greenItem,  thick, fill =greenItem, opacity=0.75] (#2,#3) rectangle node[midway, opacity = 1]{#1} (#4,#5)}
\newcommand{\drawGreenItemDashed}[5][$ $]{\draw[color=greenItem,very thick,dashed, fill =greenItem, opacity=0.5] (#2,#3) rectangle node[midway, opacity = 1]{#1} (#4,#5)}
\newcommand{\drawYellowItem}[5][$ $]{\draw[color=yellowItem,  thick,fill =yellowItem, opacity=0.7] (#2,#3) rectangle node [midway, opacity = 1]{#1} (#4,#5)}
\newcommand{\drawRedItem}[5][$ $]{\draw[color=redItem,  thick, fill=redItem, opacity=0.7] (#2,#3) rectangle node[midway, opacity = 1]{#1} (#4,#5)}
\newcommand{\drawBlueFilledItem}[5][$ $]{\draw[color=orangeborder,  thick,fill=tbBgOdd, fill opacity =0.7] (#2,#3) rectangle node[midway, opacity = 1]{#1} (#4,#5)}
\newcommand{\drawLightBlueFilledItem}[5][$ $]{\draw[color=orangeborder,  thick,fill=tbBgOdd, fill opacity =0.4] (#2,#3) rectangle node[midway, opacity = 1]{#1} (#4,#5)}
\newcommand{\drawGreenFilledArea}[5][$ $]{\draw[color=black!45!white,  thick,fill=greenItem!70!black, fill opacity =0.3] (#2,#3) rectangle node[midway, opacity = 1]{#1} (#4,#5)}

\newcommand{\drawBlackCuboid}[7][$ $]{\draw[color=black] (#2,#3,#4) -- (#2+#5,#3,#4);
\draw[color=black] (#2,#3,#4)-- (#2,#3+#6,#4);
\draw[color=black] (#2,#3,#4) -- (#2,#3,#4+#7);
\draw[color=black] (#2+#5,#3,#4) -- (#2+#5,#3+#6,#4);
\draw[color=black] (#2+#5,#3,#4) -- (#2+#5,#3,#4+#7);
\draw[color=black] (#2,#3+#6,#4) -- (#2,#3+#6,#4+#7);
\draw[color=black] (#2,#3+#6,#4) -- (#2+#5,#3+#6,#4);
\draw[color=black] (#2+#5,#3+#6,#4) -- (#2+#5,#3+#6,#4+#7);
\draw[color=black] (#2,#3+#6,#4+#7) -- (#2+#5,#3+#6,#4+#7);
\draw[color=black] (#2+#5,#3,#4+#7) -- (#2+#5,#3+#6,#4+#7);
\draw[color=black] (#2,#3,#4+#7) -- (#2+#5,#3,#4+#7);
\draw[color=black] (#2,#3,#4+#7) -- (#2,#3+#6,#4+#7);
\node at (#2+0.5*#5, #3+0.5*#6, #4+0.5*#7) {#1};
}
\newcommand{\drawRedCuboid}[7][$ $]{\draw[color=redItem!80!black] (#2,#3,#4) -- (#2+#5,#3,#4);
	\draw[color=redItem!80!black] (#2,#3,#4)-- (#2,#3+#6,#4);
	\draw[color=redItem!80!black] (#2,#3,#4) -- (#2,#3,#4+#7);
	\draw[color=redItem!80!black] (#2+#5,#3,#4) -- (#2+#5,#3+#6,#4);
	\draw[color=redItem!80!black] (#2+#5,#3,#4) -- (#2+#5,#3,#4+#7);
	\draw[color=redItem!80!black] (#2,#3+#6,#4) -- (#2,#3+#6,#4+#7);
	\draw[color=redItem!80!black] (#2,#3+#6,#4) -- (#2+#5,#3+#6,#4);
	\draw[color=redItem!80!black] (#2+#5,#3+#6,#4) -- (#2+#5,#3+#6,#4+#7);
	\draw[color=redItem!80!black] (#2,#3+#6,#4+#7) -- (#2+#5,#3+#6,#4+#7);
	\draw[color=redItem!80!black] (#2+#5,#3,#4+#7) -- (#2+#5,#3+#6,#4+#7);
	\draw[color=redItem!80!black] (#2,#3,#4+#7) -- (#2+#5,#3,#4+#7);
	\draw[color=redItem!80!black] (#2,#3,#4+#7) -- (#2,#3+#6,#4+#7);
	\node at (#2+0.5*#5, #3+0.5*#6, #4+0.5*#7) {#1};
}
\newcommand{\drawLightBlueCuboid}[7][$ $]{\draw[color=formerLB] (#2,#3,#4) -- (#2+#5,#3,#4);
	\draw[color=formerLB] (#2,#3,#4)-- (#2,#3+#6,#4);
	\draw[color=formerLB] (#2,#3,#4) -- (#2,#3,#4+#7);
	\draw[color=formerLB] (#2+#5,#3,#4) -- (#2+#5,#3+#6,#4);
	\draw[color=formerLB] (#2+#5,#3,#4) -- (#2+#5,#3,#4+#7);
	\draw[color=formerLB] (#2,#3+#6,#4) -- (#2,#3+#6,#4+#7);
	\draw[color=formerLB] (#2,#3+#6,#4) -- (#2+#5,#3+#6,#4);
	\draw[color=formerLB] (#2+#5,#3+#6,#4) -- (#2+#5,#3+#6,#4+#7);
	\draw[color=formerLB] (#2,#3+#6,#4+#7) -- (#2+#5,#3+#6,#4+#7);
	\draw[color=formerLB] (#2+#5,#3,#4+#7) -- (#2+#5,#3+#6,#4+#7);
	\draw[color=formerLB] (#2,#3,#4+#7) -- (#2+#5,#3,#4+#7);
	\draw[color=formerLB] (#2,#3,#4+#7) -- (#2,#3+#6,#4+#7);
	\node at (#2+0.5*#5, #3+0.5*#6, #4+0.5*#7) {#1};
}
\newcommand{\drawGreenCuboid}[7][$ $]{\draw[color=greenItem!75!black] (#2,#3,#4) -- (#2+#5,#3,#4);
	\draw[color=greenItem!75!black] (#2,#3,#4)-- (#2,#3+#6,#4);
	\draw[color=greenItem!75!black] (#2,#3,#4) -- (#2,#3,#4+#7);
	\draw[color=greenItem!75!black] (#2+#5,#3,#4) -- (#2+#5,#3+#6,#4);
	\draw[color=greenItem!75!black] (#2+#5,#3,#4) -- (#2+#5,#3,#4+#7);
	\draw[color=greenItem!75!black] (#2,#3+#6,#4) -- (#2,#3+#6,#4+#7);
	\draw[color=greenItem!75!black] (#2,#3+#6,#4) -- (#2+#5,#3+#6,#4);
	\draw[color=greenItem!75!black] (#2+#5,#3+#6,#4) -- (#2+#5,#3+#6,#4+#7);
	\draw[color=greenItem!75!black] (#2,#3+#6,#4+#7) -- (#2+#5,#3+#6,#4+#7);
	\draw[color=greenItem!75!black] (#2+#5,#3,#4+#7) -- (#2+#5,#3+#6,#4+#7);
	\draw[color=greenItem!75!black] (#2,#3,#4+#7) -- (#2+#5,#3,#4+#7);
	\draw[color=greenItem!75!black] (#2,#3,#4+#7) -- (#2,#3+#6,#4+#7);
	\node at (#2+0.5*#5, #3+0.5*#6, #4+0.5*#7) {#1};
}
\newcommand{\drawRedCuboidTC}[7][$ $]{\draw[color=redItem!80!black] (#2,#3,#4) -- (#2+#5,#3,#4);
	\draw[pattern=north east lines, pattern color=redItem!80!black,opacity = 0.3,]  (#2,#3+#6,#4) -- (#2+#5,#3+#6,#4) -- (#2+#5,#3+#6,#4+#7) -- (#2,#3+#6,#4+#7) -- cycle;
	\draw[color=redItem!80!black] (#2,#3,#4)-- (#2,#3+#6,#4);
	\draw[color=redItem!80!black] (#2,#3,#4) -- (#2,#3,#4+#7);
	\draw[color=redItem!80!black] (#2+#5,#3,#4) -- (#2+#5,#3+#6,#4);
	\draw[color=redItem!80!black] (#2+#5,#3,#4) -- (#2+#5,#3,#4+#7);
	\draw[color=redItem!80!black] (#2,#3+#6,#4) -- (#2,#3+#6,#4+#7);
	\draw[color=redItem!80!black] (#2,#3+#6,#4) -- (#2+#5,#3+#6,#4);
	\draw[color=redItem!80!black] (#2+#5,#3+#6,#4) -- (#2+#5,#3+#6,#4+#7);
	\draw[color=redItem!80!black] (#2,#3+#6,#4+#7) -- (#2+#5,#3+#6,#4+#7);
	\draw[color=redItem!80!black] (#2+#5,#3,#4+#7) -- (#2+#5,#3+#6,#4+#7);
	\draw[color=redItem!80!black] (#2,#3,#4+#7) -- (#2+#5,#3,#4+#7);
	\draw[color=redItem!80!black] (#2,#3,#4+#7) -- (#2,#3+#6,#4+#7);
	\node at (#2+0.5*#5, #3+0.5*#6, #4+0.5*#7) {#1};
}
\newcommand{\drawGreenCuboidTC}[7][$ $]{\draw[color=greenItem!75!black] (#2,#3,#4) -- (#2+#5,#3,#4);
	\draw[pattern=north east lines, pattern color=greenItem!75!black,opacity = 0.3,]  (#2,#3+#6,#4) -- (#2+#5,#3+#6,#4) -- (#2+#5,#3+#6,#4+#7) -- (#2,#3+#6,#4+#7) -- cycle;
	\draw[color=greenItem!75!black] (#2,#3,#4)-- (#2,#3+#6,#4);
	\draw[color=greenItem!75!black] (#2,#3,#4) -- (#2,#3,#4+#7);
	\draw[color=greenItem!75!black] (#2+#5,#3,#4) -- (#2+#5,#3+#6,#4);
	\draw[color=greenItem!75!black] (#2+#5,#3,#4) -- (#2+#5,#3,#4+#7);
	\draw[color=greenItem!75!black] (#2,#3+#6,#4) -- (#2,#3+#6,#4+#7);
	\draw[color=greenItem!75!black] (#2,#3+#6,#4) -- (#2+#5,#3+#6,#4);
	\draw[color=greenItem!75!black] (#2+#5,#3+#6,#4) -- (#2+#5,#3+#6,#4+#7);
	\draw[color=greenItem!75!black] (#2,#3+#6,#4+#7) -- (#2+#5,#3+#6,#4+#7);
	\draw[color=greenItem!75!black] (#2+#5,#3,#4+#7) -- (#2+#5,#3+#6,#4+#7);
	\draw[color=greenItem!75!black] (#2,#3,#4+#7) -- (#2+#5,#3,#4+#7);
	\draw[color=greenItem!75!black] (#2,#3,#4+#7) -- (#2,#3+#6,#4+#7);
	\node at (#2+0.5*#5, #3+0.5*#6, #4+0.5*#7) {#1};
}
\newcommand{\drawLightBlueCuboidTC}[7][$ $]{\draw[color=formerLB] (#2,#3,#4) -- (#2+#5,#3,#4);
	\draw[pattern=north east lines, pattern color=formerLB,opacity = 0.3]  (#2,#3+#6,#4) -- (#2+#5,#3+#6,#4) -- (#2+#5,#3+#6,#4+#7) -- (#2,#3+#6,#4+#7) -- cycle;
	\draw[color=formerLB] (#2,#3,#4)-- (#2,#3+#6,#4);
	\draw[color=formerLB] (#2,#3,#4) -- (#2,#3,#4+#7);
	\draw[color=formerLB] (#2+#5,#3,#4) -- (#2+#5,#3+#6,#4);
	\draw[color=formerLB] (#2+#5,#3,#4) -- (#2+#5,#3,#4+#7);
	\draw[color=formerLB] (#2,#3+#6,#4) -- (#2,#3+#6,#4+#7);
	\draw[color=formerLB] (#2,#3+#6,#4) -- (#2+#5,#3+#6,#4);
	\draw[color=formerLB] (#2+#5,#3+#6,#4) -- (#2+#5,#3+#6,#4+#7);
	\draw[color=formerLB] (#2,#3+#6,#4+#7) -- (#2+#5,#3+#6,#4+#7);
	\draw[color=formerLB] (#2+#5,#3,#4+#7) -- (#2+#5,#3+#6,#4+#7);
	\draw[color=formerLB] (#2,#3,#4+#7) -- (#2+#5,#3,#4+#7);
	\draw[color=formerLB] (#2,#3,#4+#7) -- (#2,#3+#6,#4+#7);
	\node at (#2+0.5*#5, #3+0.5*#6, #4+0.5*#7) {#1};
}
\newcommand{\drawLightBlueCuboidFC}[7][$ $]{\draw[color=formerLB] (#2,#3,#4) -- (#2+#5,#3,#4);
	\draw[pattern=north east lines, pattern color=formerLB,opacity = 0.3]  (#2,#3+#6,#4+#7) -- (#2+#5,#3+#6,#4+#7) -- (#2+#5,#3,#4+#7) -- (#2,#3,#4+#7) -- cycle;
	\draw[color=formerLB] (#2,#3,#4)-- (#2,#3+#6,#4);
	\draw[color=formerLB] (#2,#3,#4) -- (#2,#3,#4+#7);
	\draw[color=formerLB] (#2+#5,#3,#4) -- (#2+#5,#3+#6,#4);
	\draw[color=formerLB] (#2+#5,#3,#4) -- (#2+#5,#3,#4+#7);
	\draw[color=formerLB] (#2,#3+#6,#4) -- (#2,#3+#6,#4+#7);
	\draw[color=formerLB] (#2,#3+#6,#4) -- (#2+#5,#3+#6,#4);
	\draw[color=formerLB] (#2+#5,#3+#6,#4) -- (#2+#5,#3+#6,#4+#7);
	\draw[color=formerLB] (#2,#3+#6,#4+#7) -- (#2+#5,#3+#6,#4+#7);
	\draw[color=formerLB] (#2+#5,#3,#4+#7) -- (#2+#5,#3+#6,#4+#7);
	\draw[color=formerLB] (#2,#3,#4+#7) -- (#2+#5,#3,#4+#7);
	\draw[color=formerLB] (#2,#3,#4+#7) -- (#2,#3+#6,#4+#7);
	\node at (#2+0.5*#5, #3+0.5*#6, #4+0.5*#7) {#1};
}
\newcommand{\drawGreenCuboidFC}[7][$ $]{\draw[color=greenItem!75!black] (#2,#3,#4) -- (#2+#5,#3,#4);
	\draw[pattern=north east lines, pattern color=greenItem!75!black,opacity = 0.3]  (#2,#3+#6,#4+#7) -- (#2+#5,#3+#6,#4+#7) -- (#2+#5,#3,#4+#7) -- (#2,#3,#4+#7) -- cycle;
	\draw[color=greenItem!75!black] (#2,#3,#4)-- (#2,#3+#6,#4);
	\draw[color=greenItem!75!black] (#2,#3,#4) -- (#2,#3,#4+#7);
	\draw[color=greenItem!75!black] (#2+#5,#3,#4) -- (#2+#5,#3+#6,#4);
	\draw[color=greenItem!75!black] (#2+#5,#3,#4) -- (#2+#5,#3,#4+#7);
	\draw[color=greenItem!75!black] (#2,#3+#6,#4) -- (#2,#3+#6,#4+#7);
	\draw[color=greenItem!75!black] (#2,#3+#6,#4) -- (#2+#5,#3+#6,#4);
	\draw[color=greenItem!75!black] (#2+#5,#3+#6,#4) -- (#2+#5,#3+#6,#4+#7);
	\draw[color=greenItem!75!black] (#2,#3+#6,#4+#7) -- (#2+#5,#3+#6,#4+#7);
	\draw[color=greenItem!75!black] (#2+#5,#3,#4+#7) -- (#2+#5,#3+#6,#4+#7);
	\draw[color=greenItem!75!black] (#2,#3,#4+#7) -- (#2+#5,#3,#4+#7);
	\draw[color=greenItem!75!black] (#2,#3,#4+#7) -- (#2,#3+#6,#4+#7);
	\node at (#2+0.5*#5, #3+0.5*#6, #4+0.5*#7) {#1};
}
\newcommand{\drawRedCuboidFC}[7][$ $]{\draw[color=redItem!80!black] (#2,#3,#4) -- (#2+#5,#3,#4);
	\draw[pattern=north east lines, pattern color=redItem!80!black,opacity = 0.3]  (#2,#3+#6,#4+#7) -- (#2+#5,#3+#6,#4+#7) -- (#2+#5,#3,#4+#7) -- (#2,#3,#4+#7) -- cycle;
	\draw[color=redItem!80!black] (#2,#3,#4)-- (#2,#3+#6,#4);
	\draw[color=redItem!80!black] (#2,#3,#4) -- (#2,#3,#4+#7);
	\draw[color=redItem!80!black] (#2+#5,#3,#4) -- (#2+#5,#3+#6,#4);
	\draw[color=redItem!80!black] (#2+#5,#3,#4) -- (#2+#5,#3,#4+#7);
	\draw[color=redItem!80!black] (#2,#3+#6,#4) -- (#2,#3+#6,#4+#7);
	\draw[color=redItem!80!black] (#2,#3+#6,#4) -- (#2+#5,#3+#6,#4);
	\draw[color=redItem!80!black] (#2+#5,#3+#6,#4) -- (#2+#5,#3+#6,#4+#7);
	\draw[color=redItem!80!black] (#2,#3+#6,#4+#7) -- (#2+#5,#3+#6,#4+#7);
	\draw[color=redItem!80!black] (#2+#5,#3,#4+#7) -- (#2+#5,#3+#6,#4+#7);
	\draw[color=redItem!80!black] (#2,#3,#4+#7) -- (#2+#5,#3,#4+#7);
	\draw[color=redItem!80!black] (#2,#3,#4+#7) -- (#2,#3+#6,#4+#7);
	\node at (#2+0.5*#5, #3+0.5*#6, #4+0.5*#7) {#1};
}
\newcommand{\drawBlackCuboidDashed}[7][$ $]{\draw[color=black,dashed] (#2,#3,#4) -- (#2+#5,#3,#4);
	\draw[color=black,dashed] (#2,#3,#4)-- (#2,#3+#6,#4);
	\draw[color=black,dashed] (#2,#3,#4) -- (#2,#3,#4+#7);
	\draw[color=black] (#2+#5,#3,#4) -- (#2+#5,#3+#6,#4);
	\draw[color=black] (#2+#5,#3,#4) -- (#2+#5,#3,#4+#7);
	\draw[color=black] (#2,#3+#6,#4) -- (#2,#3+#6,#4+#7);
	\draw[color=black] (#2,#3+#6,#4) -- (#2+#5,#3+#6,#4);
	\draw[color=black] (#2+#5,#3+#6,#4) -- (#2+#5,#3+#6,#4+#7);
	\draw[color=black] (#2,#3+#6,#4+#7) -- (#2+#5,#3+#6,#4+#7);
	\draw[color=black] (#2+#5,#3,#4+#7) -- (#2+#5,#3+#6,#4+#7);
	\draw[color=black] (#2,#3,#4+#7) -- (#2+#5,#3,#4+#7);
	\draw[color=black] (#2,#3,#4+#7) -- (#2,#3+#6,#4+#7);
	\node at (#2+0.5*#5, #3+0.5*#6, #4+0.5*#7) {#1};
}
\newcommand{\drawRedCuboidDashed}[7][$ $]{\draw[color=redItem!60!black,dashed] (#2,#3,#4) -- (#2+#5,#3,#4);
	\draw[color=redItem!60!black,dashed] (#2,#3,#4)-- (#2,#3+#6,#4);
	\draw[color=redItem!60!black,dashed] (#2,#3,#4) -- (#2,#3,#4+#7);
	\draw[color=redItem!60!black] (#2+#5,#3,#4) -- (#2+#5,#3+#6,#4);
	\draw[color=redItem!60!black] (#2+#5,#3,#4) -- (#2+#5,#3,#4+#7);
	\draw[color=redItem!60!black] (#2,#3+#6,#4) -- (#2,#3+#6,#4+#7);
	\draw[color=redItem!60!black] (#2,#3+#6,#4) -- (#2+#5,#3+#6,#4);
	\draw[color=redItem!60!black] (#2+#5,#3+#6,#4) -- (#2+#5,#3+#6,#4+#7);
	\draw[color=redItem!60!black] (#2,#3+#6,#4+#7) -- (#2+#5,#3+#6,#4+#7);
	\draw[color=redItem!60!black] (#2+#5,#3,#4+#7) -- (#2+#5,#3+#6,#4+#7);
	\draw[color=redItem!60!black] (#2,#3,#4+#7) -- (#2+#5,#3,#4+#7);
	\draw[color=redItem!60!black] (#2,#3,#4+#7) -- (#2,#3+#6,#4+#7);
	\node at (#2+0.5*#5, #3+0.5*#6, #4+0.5*#7) {#1};
}
\newcommand{\drawLightBlueCuboidDashed}[7][$ $]{\draw[color=formerLB,dashed] (#2,#3,#4) -- (#2+#5,#3,#4);
	\draw[color=formerLB,dashed] (#2,#3,#4)-- (#2,#3+#6,#4);
	\draw[color=formerLB,dashed] (#2,#3,#4) -- (#2,#3,#4+#7);
	\draw[color=formerLB] (#2+#5,#3,#4) -- (#2+#5,#3+#6,#4);
	\draw[color=formerLB] (#2+#5,#3,#4) -- (#2+#5,#3,#4+#7);
	\draw[color=formerLB] (#2,#3+#6,#4) -- (#2,#3+#6,#4+#7);
	\draw[color=formerLB] (#2,#3+#6,#4) -- (#2+#5,#3+#6,#4);
	\draw[color=formerLB] (#2+#5,#3+#6,#4) -- (#2+#5,#3+#6,#4+#7);
	\draw[color=formerLB] (#2,#3+#6,#4+#7) -- (#2+#5,#3+#6,#4+#7);
	\draw[color=formerLB] (#2+#5,#3,#4+#7) -- (#2+#5,#3+#6,#4+#7);
	\draw[color=formerLB] (#2,#3,#4+#7) -- (#2+#5,#3,#4+#7);
	\draw[color=formerLB] (#2,#3,#4+#7) -- (#2,#3+#6,#4+#7);
	\node at (#2+0.5*#5, #3+0.5*#6, #4+0.5*#7) {#1};
}
\newcommand{\drawGreenCuboidDashed}[7][$ $]{\draw[color=greenItem!75!black,dashed] (#2,#3,#4) -- (#2+#5,#3,#4);
	\draw[color=greenItem!75!black,dashed] (#2,#3,#4)-- (#2,#3+#6,#4);
	\draw[color=greenItem!75!black,dashed] (#2,#3,#4) -- (#2,#3,#4+#7);
	\draw[color=greenItem!75!black] (#2+#5,#3,#4) -- (#2+#5,#3+#6,#4);
	\draw[color=greenItem!75!black] (#2+#5,#3,#4) -- (#2+#5,#3,#4+#7);
	\draw[color=greenItem!75!black] (#2,#3+#6,#4) -- (#2,#3+#6,#4+#7);
	\draw[color=greenItem!75!black] (#2,#3+#6,#4) -- (#2+#5,#3+#6,#4);
	\draw[color=greenItem!75!black] (#2+#5,#3+#6,#4) -- (#2+#5,#3+#6,#4+#7);
	\draw[color=greenItem!75!black] (#2,#3+#6,#4+#7) -- (#2+#5,#3+#6,#4+#7);
	\draw[color=greenItem!75!black] (#2+#5,#3,#4+#7) -- (#2+#5,#3+#6,#4+#7);
	\draw[color=greenItem!75!black] (#2,#3,#4+#7) -- (#2+#5,#3,#4+#7);
	\draw[color=greenItem!75!black] (#2,#3,#4+#7) -- (#2,#3+#6,#4+#7);
	\node at (#2+0.5*#5, #3+0.5*#6, #4+0.5*#7) {#1};
}

\section{Introduction}
Three-dimensional Knapsack (\tdk) is a fundamental problem in computational geometry, operations research, and approximation algorithms. 
In 3DK, we are given a set of items which are axis-aligned cuboids (defined by their width, depth, and height) with associated profits, and a three-dimensional (3D) knapsack which is an axis-aligned unit cube. The goal is to find a nonoverlapping axis-aligned packing (by translation) of a subset of items into the knapsack such that the profit of the packed items is maximized. 

In 1965, Gilmore and Gomory \cite{gilmore1965multistage} introduced the 3D Knapsack problem -- motivated by cutting stock problems (e.g., cutting up of graphite blocks for anodes). They provided heuristics for the special case when the items need to be cut along axis-aligned planes.
Since then, the problem and its several variants (under practical constraints) have been well-studied in operations research -- typically from the viewpoints of meta-heuristics, MILP, tree search, machine learning methods, and greedy algorithms \cite{cacchiani2022knapsack, silva2019exact}.
In the last decade, \tdk~has become a central problem in logistics to increase operational efficiency by efficiently utilizing the space in various domains, such as airline cargo management, warehouse management, and robotic container loading \cite{ali2022line}.
The prominence of this problem in the industry is evidenced by the comprehensive survey on 3D packing by Ali, Ramos, Carravilla, and Oliveira \cite{ali2022line}, who have mentioned hundreds of papers in recent years.  
There have been multiple programming challenges and competitions focusing on variants of \tdk, e.g., OPTIL.io \cite{optil}, ICRA palletization competition \cite{vmac}, etc.   
Several companies like Boxify \cite{boxify}, Symbotic \cite{symbotic}, etc. are specializing in 3D packing for their operations.  

Packing problems are fundamental and popular among mathematicians (e.g., see Kepler's sphere packing conjecture \cite{lagarias2011kepler} and {\em Soma cube} puzzle \cite{goodman2019soma}). 
Surprisingly, \tdk~is relatively less studied in computational geometry and approximation algorithms. 
Only in 2008, Diedrich, Harren, Jansen, Th{\"{o}}le, and Thomas \cite{3d-knapsack} provided the first provable approximation guarantees for the problem. They first gave a simple polynomial-time $(9+\eps)$-approximation algorithms for any constant $\eps>0$. Thereafter, they gave a more sophisticated $(7+\eps)$-approximation algorithm for \tdk. For the case when the items can be rotated by 90 degrees around any axis, they provided a $(5+\eps)$-approximation algorithm. 
Afterward, Chleb{\'{\i}}k and Chleb{\'{\i}}kov{\'{a}} \cite{ChlebikC09} showed that the problem admits no asymptotic polynomial-time approximation schemes (APTAS) even for the cardinality case (when each item has the same profit).\footnote{See Appendix \ref{sec:approx} for definitions related to approximation algorithms and approximation schemes.} Lu, Chen, and Cha \cite{lu2015packing1} showed that even packing cubes into a cube is strongly NP-hard.

There has been progress in some special cases.
When all items are $d$-dimensional hypercubes, Harren \cite{harren-journal} gave a $(1+2^{-d}+\eps)$-approximation algorithm. 
Recently, Jansen, Khan, Lira, and Sreenivas \cite{Jansencubepacking22} gave a PTAS for this special case. 
Afterwards, Buchem, Dueker, and Wiese \cite{BuchemDW24}  obtained an EPTAS for the problem. 
Furthermore, they gave dynamic algorithms with polylogarithmic query and update times, matching the same approximation guarantees.

For the 2D variant when all items are rectangles (2DK), Jansen and Zhang \cite{jansen-zhang} gave a $(2+\eps)$-approximation algorithm. 
G{\'{a}}lvez, Grandoni, Ingala, Heydrich, Khan, and Wiese \cite{2dks-Lpacking} 
gave an improved 1.89-approximation algorithm. Later, G{\'{a}}lvez, Grandoni, Khan, Ram{\'{\i}}rez{-}Romero, and Wiese \cite{Galvez00RW21} gave a pseudopolynomial-time 4/3-approximation algorithm. 
For the variant of {\em guillotine packing} (where items need to be packed such that they can be separated by a series of end-to-end cuts), Khan, Maiti, Sharma, and Wiese \cite{KMSW21} gave a pseudopolynomial-time approximation scheme (PPTAS) for the problem. 
 

Compared to the classical (1D) Knapsack, multidimensional Knapsack is much harder as there are different types of items and they can interact in a complicated way. 
For example, 2D items can be big (large in both dimensions), small (tiny in both dimensions), vertical (large in height, small in width), or horizontal (large in width, small in height).
Most approximation algorithms for 2DK show the existence of a highly profitable structured packing where the knapsack is partitioned into a constant number of boxes such that each box contains only items of one type and can be packed easily by simple greedy algorithms. 
Large boxes contain a single large item, vertical (resp. horizontal) boxes contain only vertical (resp. horizontal) items (thus one can essentially ignore the long dimension of the item and the box becomes a 1D knapsack), and small boxes contain small items (one can consider the area of an item as a proxy and use 1D knapsack solution). 
However, this is difficult for 3D packing.  
As pointed out by Bansal, Correa, Kenyon, and Sviridenko \cite{bansal2006bin}, one of the obstacles to generalize methods from rectangle packing to cuboid packing is that ``{\em the interaction of say 3-dimensional rectangles which are long in one direction but short in the other two seems much more complicated than in the two-dimensional case.}'' 

In 2008, the authors in \cite{3d-knapsack} mentioned that ``{\em it is of interest whether here (for \tdk) an algorithm with ratio $(6+\eps)$ or less exists.}''
However, despite significant progress in the 2D case, there has been no improved approximation algorithm for \tdk~in almost two decades. 
\subsection{Our results}
In this paper, we introduce {\em container packing} for 3D Knapsack. We show that we can divide the 3D knapsack into $O(1)$ number of regions  (called {\em containers}) such that the dimensions of these containers come from a polynomial-sized set. 
For each container $C$ we define its capacity  $\mathtt{cap}(C) \in \mathbb{R}_{\ge 0}$ and for each item $i$ we define its size in $C$ to be $f_C(i) \in \mathbb{R}_{\ge 0}$. 
We also associate an algorithm $\AM_C$ for the container.
We require $\AM_C, f_C$  such that for any $C$ and any set of items $T$ satisfying the packing constraint $\sum_{i\in T} f_C(i) \le \mathtt{cap}(C)$, almost all items in $T$ (barring a small profit subset) can be packed into $C$ by algorithm $\AM_C$. 
The assignment of items into containers thus becomes a variant of the Multiple Knapsack problem (however, the knapsacks can have different capacities and items can have sizes specific to the knapsacks) --- also a special case of Generalized Assignment Problem (GAP). 
For $O(1)$ knapsacks, there exists a PTAS for GAP \cite{2dks-Lpacking}.
After the assignment of items to the container $C$, the corresponding algorithm $\AM_C$ ensures that the items admit a feasible profitable packing.

Container packing is a powerful and general technique. For example, by choosing appropriate functions $f_C$, our containers can handle complicated packings, beyond simple greedy algorithms. 
For example, we use six different types of containers in this work (see Figures  \ref{fig:large-stack-containers} and \ref{fig:l-container-2}).
Apart from the Stack containers and Volume containers (similar to 2D packing), we introduce Area containers (packed using a variant of Next-Fit-Decreasing-Height (\NFDH) algorithm \cite{coffman1980performance}), Steinberg containers (where different types of items are packed together using our new volume-based packing algorithm \texttt{3D-Vol-Pack}
-- it  packs items in layers where each layer is packed using either a greedy algorithm or a 2D packing algorithm called Steinberg's algorithm \cite{steinberg}), and \LContainer{}s (where two different types of stacks are packed together when rotations are allowed). 
Note that Steinberg container or \LContainer~goes beyond simple stack or shelf-based algorithms and equips us to handle more complicated packings of different types of items. 
\ari{Finally, we show that given an optimal packing, we can modify it to some \emph{container packing} having significant profit.} The functions $f_C$ can be considered analogous to the {\em weighting functions} in bin packing \cite{coffman1980performance}, and we expect \kvn{that} more complicated packing algorithms can be incorporated into this framework by suitably defining $f_C$.
We also believe that our techniques might find usage in other multidimensional packing and scheduling problems.

Using this, we obtain improved approximation guarantees for 3DK (see Table \ref{tab:summary}).
In particular, we obtain $(139/29+\eps)$-approximation for \tdk~and $(30/7+\eps)$-approximation for \tdkr~(when rotations by 90 degrees are allowed around all axes). 
For the cardinality case (when all items have the same profit), we obtain $17/4 +\eps$ and $24/7+\eps$, respectively, for the cases without and with rotations. Finally, in the case when the profit of an item is equal to its volume (also called uniform profit-density), we obtain approximation guarantees of $4+\eps$ and $3+\eps$, without and with rotations, respectively. We also show that one cannot obtain a polynomial-time $(3-\eps)$-approximation \ari{for \tdk~even in the cardinality case} through container packing, for any constant $\eps>0$,  using the containers mentioned in this paper (see \cref{sec:hardnesscontainer}).

\begin{table}
\centering{}%
\bgroup
\def\arraystretch{3}
\begin{tabular}{|c|c|c|c|}
\hline 
  & \textbf{General case} & \textbf{Cardinality case}  & \textbf{Profit = Volume} \tabularnewline
\hline 
Without rotations & $\displaystyle\frac{139}{29}+\epsilon < 4.794$ & $\displaystyle\frac{17}{4}+\epsilon = 4.25 + \epsilon$ & $4+\epsilon$  
\tabularnewline\hline
With rotations & $\displaystyle\frac{30}{7}+\epsilon < 4.286$  & $\displaystyle\frac{24}{7}+\epsilon < 3.429$  & $3+\epsilon$ 
\tabularnewline
\hline
\end{tabular}
\egroup
\caption{Summary of our results}
\label{tab:summary}
\end{table}



\subsection{Related work}
Two other related problems are 3D Bin Packing (BP) and Strip Packing (SP). 
In 3D BP, given a set of cuboids and unit cubes as bins, the goal is to find a nonoverlapping axis-aligned packing of all items into the minimum number of bins.
In 3D SP,  given a set of  cuboids, the goal is to pack them into a single 3D rectangular box of unlimited height such that the height of the packing is minimized.
For 3D BP and 3D SP, the best-known (asymptotic) approximation ratios are $T_\infty^{2}\approx 2.86$ \cite{caprara2008packing} and  $3/2+\eps$ \cite{3d-strip-packing}, resp.
There has also been a series of work in 2D BP \cite{BansalCJPS09, bansal2014improved, jansen-pradel-bin-packing, Khan16c} and 
2D SP \cite{harren20145, KenyonR00}.
Another related problem is the maximum independent set of rrectangles, where given a set of possibly overlapping hyperrectangles the goal is to find the maximum number of disjoint hyperrectangles
\cite{AHW19, misr-arxiv, Mitchell21, GalvezKMMPW22, GargK024}.
In 2D, many variants of the Knapsack problem have been studied where the knapsack is a unit square and objects are polygons \cite{MerinoW20} or circles \cite{AcharyaBG0MW24, ChagasNew}. 
For more details on related problems, we refer the readers to surveys \cite{alt2016computational, CKPT17} on approximation and online algorithms for multidimensional packing.

\section{Notation and Preliminaries}
\label{sec:prelim}
We are given a set of $n$ items $I=\{1, 2, \dots, n\}$.
Each item $i \in I$ is a cuboid with width $w_i \in (0,1]$, depth $d_i \in (0,1]$, height $h_i \in (0,1]$, and  an associated profit $p_i \in \mathbb{Q}$.
The volume of an item $i$ is $v(i)=w_i \times d_i \times h_i$,
and its base area is defined as $w_id_i$.
We are also given a knapsack $K:=[0,1]\times [0,1] \times [0,1]$.
If an item $i$ is placed (by translation) in an axis-aligned way at $(i_x, i_y, i_z)$ then it occupies the region: $[i_x, i_x+w_i]\times [i_y, i_y+d_i] \times [i_z, i_z+h_i]$.
For a feasible packing we need $i_x \in [0, 1-w_i], i_y \in [0, 1-d_i], i_z \in [0, 1-h_i]$.
In a packing, two items are nonoverlapping if their interiors are disjoint. 
Our goal is to find a nonoverlapping packing of a subset of these items that maximizes the total profit.
In \tdk, we allow only translation, whereas \tdkr~allows translation and rotation by 90 degrees around all axes.
For any set of items $T$, we shall let $p(T), h(T)$ and $v(T)$ denote the total profit, sum of heights, and volume of $T$, respectively. Let OPT denote an optimal packing and $\opt := p(\text{OPT})$.
\subsection{\kvnnew{Packing} subroutines}
\label{sec:subroutines}
We now describe a few procedures that we repeatedly use throughout the paper to pack items into \boxes. For a cuboidal box $B$, we let $w_B, d_B, h_B$ denote the width, depth, and height of $B$, respectively. The volume of the box is then $v(B) = w_B \times d_B \times h_B$.


\noindent \textbf{Next-Fit-Decreasing-Height (\NFDH) Algorithm \cite{coffman1980performance}.}
\NFDH~is a shelf-based packing algorithm for packing rectangles. \tNFDH~is its generalization in three dimensions.
See \Cref{subsec:nfdh} for details on these algorithms and proof of the following lemmas.

\begin{lemma}[\cite{2dks-Lpacking}]
\label{lem:2d-hfdh}
    Given a rectangular box of length $\ell$ (horizontal dimension) and breadth $b$ (vertical dimension), 
    and a set $S$ \kvn{of $n$} (2D) \kvn{rectangles}, where each $i\in S$ has length $\ell_i \le \epsilon\ell$ 
    and breadth $b_i \le \epsilon b$. If $\sum_{i\in S} \ell_i b_i \le (1-2\epsilon)\ell b$, then the whole set $S$ can be packed into the box using \NFDH in $O(n \log n)$ time.
\end{lemma}

\begin{restatable}{lemma}{threednfdh}
\label{lem:3d-hfdh}
    Given a cuboidal box $B$, and a set $T$ \kvn{of $n$} items where each $i\in T$ satisfies $w_i\le \epsilon w_B$, $d_i \le \epsilon d_B$ and $h_i \le \epsilon h_B$. If $v(T) \le (1-3\epsilon)v(B)$, then the whole set $T$ can be packed into $B$ using \tNFDH in $O(n \log^2 n)$ time.
\end{restatable}

\noindent \textbf{Steinberg's Algorithm \cite{steinberg}.} 
Steinberg's Algorithm is a commonly used 2D packing algorithm with the following performance guarantee.
\begin{restatable}[\cite{steinberg}]{lemma}{steinb}
\label{lem:stein}
We are given a set $S$ of $n$ \kvn{rectangles} and a \kvn{rectangular box} of size $\ell \times b$. 
Let $\ell_{\max} \le \ell$ and $b_{\max} \le b$ be the maximum
\kvn{length} and maximum \kvn{breadth} among the rectangles in $S$, respectively, and let
$a(S)$ be the total area of the rectangles in $S$. Also, we denote $x_+:=
\max(x, 0)$. 
If
$2a(S) \le \ell b - (2\ell_{\max} - \ell)_+(2b_{\max}-b)_+$,
then there is an algorithm that can pack \kvn{the whole set $S$ into the rectangular box} in $O\left(\frac{n \log^2 n}{\log \log n}\right)$ time.
\end{restatable}

We utilize the above algorithm to pack items in layers inside a box. Below, we present two algorithms \texttt{3D-Vol-Pack} and \texttt{3DR-Vol-Pack} for packing 3D items for the cases without and with rotations, respectively. As Steinberg's algorithm is used as a black box in many 2D packing problems, we believe our volume-based bounds and algorithms should find usage in other 3D packing problems. In both cases, we are given a box $B$ and a set of items $T$.

\noindent \textbf{\texttt{3D-Vol-Pack.}} We assume that each $i\in T$ satisfies either $w_i \le w_B/2$ or $d_i \le d_B/2$. 
Let $T_w$ be the items of $T$ having width at most $w_B/2$, and let $T_d = T\setminus T_w$. Then, similar to \cite{3d-knapsack}, we further classify $T_w$ as follows: let $T_{w\ell}\subseteq T_w$ be the items whose base area (area of the bottom face) exceeds $\frac{1}{6}w_Bd_B$, and let $T_{ws} = T_w \setminus T_{w\ell}$. Next, we sort the sets $T_{w\ell}$ and $T_{ws}$ in non-increasing order of heights. We group the items of $T_{w\ell}$ in pairs (except possibly the last item), and pack each pair in a single layer. See \cref{fig:3d-vol-pack}. For packing items in $T_{ws}$, we group them into maximal collections of total base area not exceeding $\frac{1}{2}w_Bd_B$, and pack each collection in a single layer using Steinberg's algorithm (\Cref{lem:stein}). Analogously, we obtain a packing of the items in $T_d$ in layers. Finally, we stack the layers one above the other inside the box $B$ as long as the height of the box is not exceeded. The following lemma provides a guarantee on the packed volume 
when the item heights are very small compared to the box $B$. 

\begin{figure}
    \centering
    \includesvg[width=0.75\linewidth]{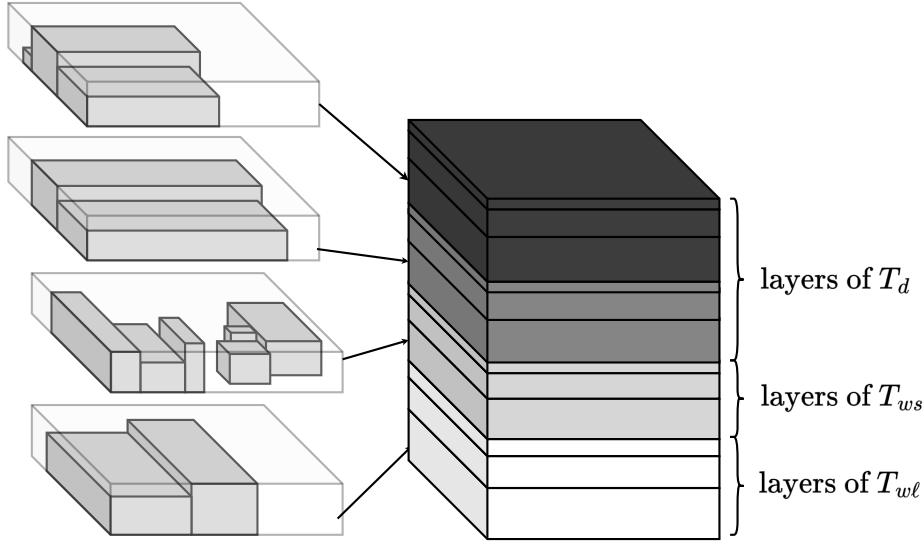}
    \caption{An example packing by Algorithm \texttt{3D-Vol-Pack}.}
    \label{fig:3d-vol-pack}
\end{figure}


\begin{lemma}
\label{lem:steinberg}
    Given a box $B$, and a set $T$ \kvn{of $n$} items where each $i\in T$ satisfies $w_i \le w_B$, $d_i \le d_B$ and $h_i \le \epsilon h_B$. Further, either $w_i \le w_B/2$ 
    or $d_i \le d_B/2$ holds for every $i\in T$. If $v(T) \le \left(\frac{1}{3}-2\epsilon\right)v(B)$, then \texttt{3D-Vol-Pack} packs the whole set $T$ into $B$ in $O\left(\frac{n \log^2 n}{\log \log n}\right)$ time.
\end{lemma}
\begin{proof}
We first establish an intermediate claim.
\begin{claim}
    \label{lem:steinberg-layers-height}
    Let $R$ be a rectangular base of dimensions $w\times d$, and let $S$ be a set of cuboids such that for each $i\in S$, either $w_i \le w/2$ or $d_i \le d/2$. 
    Let $h_{\max}$ denote the maximum height of an item in $S$. Then, $S$ can be packed on the base $R$ within a height of $4 h_{\max}+\frac{3}{wd}v(S)$. 
    \kvn{The runtime is upper bounded by $O(\frac{m \log^2 m}{\log \log m})$ time, where $m=\abs{S}$.}
\end{claim}
\begin{proof}
    We divide the set $S$ into four sets as follows.
    \begin{itemize}
        \item $S_{w1}$ denotes all the items whose width is at most $w/2$ but whose individual base area is {\em at most} $\frac16wd$.
        \item $S_{w2}$ denotes all the items whose width is at most $w/2$ but whose individual base area is {\em more than} $\frac16wd$.
        \item $S_{d1}$ denotes all the items whose depth is at most $d/2$ but whose individual base area is {\em at most} $\frac16wd$.
        \item $S_{d2}$ denotes all the items whose depth is at most $d/2$ but whose individual base area is {\em more than} $\frac16wd$.
    \end{itemize}
    We show how to pack the sets $S_{w1},S_{w2}$. The packing of $S_{d1},S_{d2}$ is similar. First, let us consider $S_{w1}$. 
    Assume that the items in $S_{w1}$ are arranged in non-increasing order of heights. Now, pick the largest prefix $P_1$
    of this order whose total base area does not exceed $\frac12wd$.
    \kvn{If $P_1=S$, then we can use Steinberg's algorithm to pack $S$ directly.}
    \kvn{Otherwise, }the maximality of $P_1$, together with the fact that each item in $S_{w1}$ has a base area at most $\frac16wd$, implies that
    the total base area of $P_1$ is at least $\frac13wd$. Using Steinberg's algorithm, we can pack the entire prefix $P_1$
    in a layer of dimensions $w\times d\times h'_1$, where $h'_1$ is the height of the tallest item in $P_1$. Now, we look at $S_{w1}\setminus P_1$ and repeat the process of choosing a maximal prefix $P_2$ of base area at most $\frac12wd$. In this way, we obtain a partition of
    $S_{w1}$ into sets $P_1,P_2,\dots,P_s$. For $j\in[s]$, let $h_j,h'_j$ denote the respective heights of the shortest, tallest items of $P_j$. We have the following guarantees for each $j\in[s]$.
    \begin{itemize}
        \item Each $P_j$ can be packed in a layer $\mathcal L_j$ of dimensions $w\times d\times h'_j$ (by the same arguments used for $P_1$).
        \item If $j<s$, the total base area of items in $P_j$ is at least $\frac13wd$ (by the same arguments used for $P_1$).
        \item If $j<s$, then $h_j>h'_{j+1}$ (since is sorted $S_{w1}$ in non-increasing order of heights).
    \end{itemize}
    We can stack up the layers $\mathcal L_1,\mathcal L_2,\dots,\mathcal L_s$ on top of each other, resulting in a height of~$h'_1+h'_2+\dots+h'_s$. Now,
    \begin{align}
        &h'_1+h'_2+\dots+h'_s\nonumber\\
        \le\:\:& h_{\max}+h_1+\dots+h_{s-1}\nonumber\\
                        =\:\:&  h_{\max}+3\left(\frac{h_1}{3}+\frac{h_2}{3}+\dots+\frac{h_{s-1}}{3}\right)\nonumber\\ 
                        \le\:\:& h_{\max}+\frac{3}{wd}\left(v(P_1)+v(P_2)+\dots+v(P_{s-1})\right)\tag{since each $P_j$, except the last, has a base area of at least $1/3$}\nonumber\\
                        =\:\:& h_{\max}+\frac{3}{wd}v(S_{w1})\label{eq:telescope}.
    \end{align}
    Now, we move to packing $S_{w2}$. This is easy. Again, we assume that the items in $S_{w2}$ are arranged in non-increasing order of heights.
    Since each item in $S_{w2}$ has a base area at least $\frac16wd$ and width at most $w/2$, we can pack the first two of them side by side in a single layer, ensuring that the total base area of the layer is at least $\frac16wd$. \kvn{(If $S_{w2}$ contains only one item, then we utilize only one layer to pack $S_{w2}$.)} We repeat this process for the rest of the items as well.
    Using an analysis similar to \cref{eq:telescope}, we obtain that the total height of the layers is at most $h_{\max}+\frac{3}{wd}v(S_{w2})$. We can pack the sets $S_{d1},S_{d2}$ analogously. The total height of
    all four packings is thus at most
    \begin{align*}
        &4h_{\max}+\frac{3}{wd}(v(S_{w1})+v(S_{w2})+v(S_{d1})+v(S_{d2}))\\
       \le\:\:&4h_{\max}+\frac{3}{wd}(v(S)).
    \end{align*}
    This ends the proof \kvn{of the packing guarantee.
    To see the runtime, observe that it is dominated by the time to run Steinberg's algorithm. Let $m_j$ denote the number of items in prefix $P_j$ defined above. Then,
    the time to pack $P_j$ is given by $O(m_j\log^2 m/\log\log m)$. Summing these times up gives us the claimed runtime guarantee}.
\end{proof}

Using \cref{lem:steinberg-layers-height}, we obtain that the set $T$ can be pack on the base of $C$ within a height of at most $4(\eps h_C)+\frac{3}{w_Cd_C}(\frac13-2\eps)w_Cd_Ch_C$,
which on simplification can be seen to be $h_C-2\eps h_C$. The runtime also follows from \cref{lem:steinberg-layers-height}.
\end{proof}



Next, we give an algorithm when the items can be rotated by 90 degrees about any axis.

\noindent \textbf{\texttt{3DR-Vol-Pack.}}
Let $T_{\ell} = \{i\in T \mid w_i > w_B/2 \text{ and } d_i > d_B/2\}$, and let $T_s = T\setminus T_{\ell}$. First, we pack $T_s$ in layers 
using \texttt{3D-Vol-Pack}. Next, we sort the items of $T_{\ell}$ in non-increasing order of their widths and pack them in layers, 
one in each layer, touching the left face of the box, as long as the height of the box is not exceeded. Following this, 
we rotate the remaining items of $T_{\ell}$ about the depth dimension so that their widths and heights are interchanged, and as 
long as possible, pack them \kvn{next to each other} starting from the right face of the box, such that each item touches the top face of the box (see \Cref{fig:3d-r-vol-pack}). 

The next lemma states that when \kvn{$B$ is a cube and when} the height of each item is very small compared to box $B$, we can pack a significant volume.
See \Cref{sec:3dr-vol-pack} for the proof.

\begin{restatable}{lemma}{tDRVolPack}
\label{lem:3DRVolPack}
    Given a cubical box $B$ of side length $w$,
    and a set $T$ of \kvn{$n$ items} where for each $i\in T$, there exists an orientation so that $w_i \le w$, $d_i \le w$ and $h_i \le \epsilon^2 w$ holds. 
    If $v(T) \le \left(\frac{7}{24}-5\epsilon\right)v(B)$, then $T$ can be packed into $B$ by \texttt{3DR-Vol-Pack} in $O\left(\frac{n \log^2 n}{\log \log n}\right)$ time.
\end{restatable}
\begin{figure}
    \centering
    \includesvg[width=0.8\linewidth]{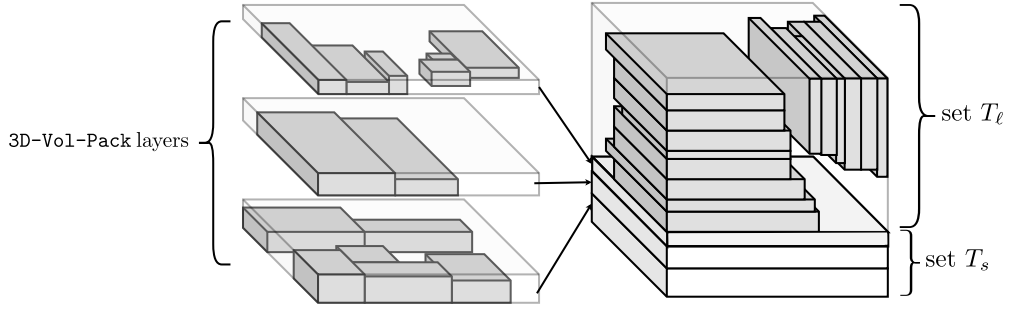}
    \caption{An example packing by Algorithm \texttt{3DR-Vol-Pack}.}
    \label{fig:3d-r-vol-pack}
\end{figure}

\section{Container Packing}
\label{subsec:cp}
\kvnnew{In this section, we first define container packing and describe the containers used in our work.
Then, we state our main structural result which shows that there exist container packings with high profit that can be searched for in polynomial time.}
\begin{definition}[\Container Packing]
    A packing of a set of items $I' \subseteq I$ is said to be a \Container Packing if \kvnr{changed $T, T', T''$}
    \begin{itemize}
        \item the knapsack can be partitioned into a collection $\CC$ of non-overlapping regions called \boxes and some empty spaces. Each container $C \in \CC$ has an associated capacity $\mathtt{cap}(C)$, a function $f_C \colon I \rightarrow \mathbb{R}_{\ge 0}$, and a polynomial-time algorithm $\AM_C$,
        \item there is a function $g \colon I' \rightarrow \CC$, such that for each $C\in \CC$, items in $g^{-1}(C)$ are packed into $C$ using algorithm $\AM_C$,
        \item for any $C\in \CC$ and $T\subseteq I$ such that $\sum_{i\in T} f_C(i) \le \mathtt{cap}(C)$, for any constant $\eps > 0$, there exists a polynomial-time computable set $T' \subseteq T$ with $p(T')\ge (1-O(\epsilon))p(T)$ such that $T'$ can be packed into $C$ by the algorithm $\AM_C$. 
    \end{itemize}
\end{definition}

A Container Packing is said to be \textit{guessable} if 
\begin{itemize}
    \item the number of \boxes inside the knapsack is bounded by a constant $O_\eps(1)$,\footnote{The notation $O_{\eps}(f(n))$ means that the implicit constant hidden in big-$O$ notation can depend on $\eps$.}
    \item the number of distinct possible types of \boxes (defined by their sizes along each dimension) is bounded by \kvn{$n^{O_\eps(1)}$, a polynomial in $n$},
    \item whether a given collection of $O_{\epsilon}(1)$ \boxes can be placed non-overlappingly inside the knapsack, can be checked in \kvn{$n^{O_{\eps}(1)}$ (polynomial) time}.
\end{itemize}

We shall call a container \emph{guessable} if it comes from some polynomially-sized set of distinct container types. Though, in this paper we only use cuboids as \boxes, the idea of container packing can be generalized to other shapes in 3D or higher dimensions. 

We show that we can first use a PTAS for GAP (see \cref{subsec:GAP} for details) to assign items into \boxes and
then use the corresponding algorithms to pack items inside the \boxes\ -- this gives us a near-optimal container packing.

\begin{restatable}{theorem}{gapptas}
\label{thm:ptasboxpacking}
    There is a PTAS for maximum profit guessable container packing.
\end{restatable}
\begin{proof}
    Since the number of distinct dimensions of containers is polynomially bounded, and a guessable container packing only has $O_{\epsilon}(1)$ containers, it is possible to enumerate all feasible choices of containers that fit into the knapsack in polynomial-time. For each such choice, we create an instance of GAP with one (1D) knapsack per container, where the (1D) knapsack corresponding to a container $C$ has capacity $\mathtt{cap}(C)$ and the size of an item $i$ for the (1D) knapsack is given by $f_C(i)$. The profit of an item for any (1D) knapsack is set to be the same as the actual profit of the item. We now use the PTAS for GAP with a constant number of knapsacks
    (\cref{lem:gap})
    to obtain an assignment of the input items into the containers that do not violate the container capacities. Finally, for each container $C$, we use the algorithm $\AM_C$ to pack a subset of the items assigned to it, by losing only an $\epsilon$-fraction of profit. Overall, this gives a PTAS to compute the maximum profitable guessable container packing. 
\end{proof}


\subsection{Classification of \boxes}
\label{sec:cont-class}

\begin{figure}
    \centering
    \includesvg[scale=0.7]{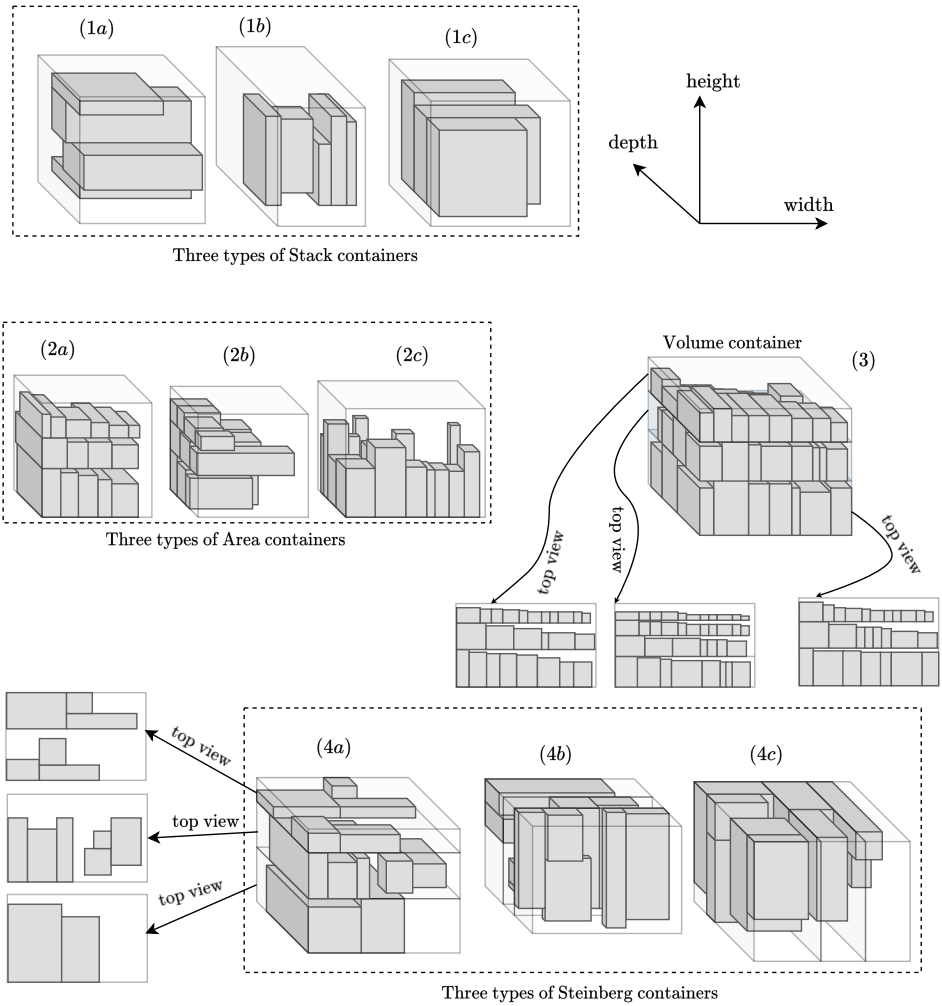}
    \caption{Different kinds of \boxes that we use. \kvnnew{We have one type of} Volume container and three types of Stack, Area, and Steinberg \boxes (one for each dimension).}
    \label{fig:large-stack-containers}
\end{figure}

We now describe the various types of \boxes\ that we use to pack items. For a container $C$, its width, depth, and height are denoted by $w_C, d_C$, and $h_C$, respectively, and its volume $v(C) := w_C \times d_C \times h_C$.
For each container $C$, we specify its capacity $\mathtt{cap}(C)$, the function $f_C,$ and the packing algorithm $\AM_C$.
Let $T$ denote the set of items assigned to $C$ such that $\sum_{i\in T}f_C(i)\le \mathtt{cap}(C)$.
We suggest the reader to follow the descriptions alongside \cref{fig:large-stack-containers}.
\begin{itemize}
    
    
    \item \textbf{Stack \boxes:} These \boxes\ pack items in layers, with one item in each layer. We consider a container $C$ that stacks items along the height.  See container ($1a$) in \cref{fig:large-stack-containers}. The capacity $\mathtt{cap}(C)$ of such a container is set to its height $h_C$, and for any item $i$, $f_C(i):=h_i$ if $w_i \le w_C$ and $d_i \le d_C$, and $f_C(i):=\infty$ otherwise.

    $\AM_C$ just stacks the items of set $T$ in the container $C$ one above the other. Analogously containers ($1b$) and ($1c$) show containers that stack along width \kvn{and} depth, respectively.
    
    \item \textbf{Area \boxes:} Inside these \boxes, the items are packed using the \NFDH~algorithm, by ignoring one of the dimensions. We describe Area \boxes\ that use \NFDH~along the front face of the container (see container ($2a$)). For such a container $C$, we set $\mathtt{cap}(C):= w_Ch_C$, and for any item $i$, we set $f_C(i):= w_i h_i$ if $w_i \le \epsilon w_C$, $h_i \le \epsilon h_C$ and $d_i \le d_C$, and $f_C(i):= \infty$ otherwise. 

     \kvn{$\AM_C$ first sorts the items} of set $T$ in non-increasing order of profit/front area ratio.
    Then, we select the largest prefix whose total front area does not exceed $(1-2\eps)w_Ch_C$. We pack this prefix using \NFDH on the front face.
    Analogously, container ($2b$) and ($2c$) show packing when \NFDH is used on the left and bottom face of $C$, respectively.

    \item \textbf{Volume \boxes:} The items are packed using \tNFDH~inside these \boxes (see container (3)). The capacity of such a container $C$ is set to be its volume $v(C)$, and $f_C(i):= v(i)$ if $w_i \le \epsilon w_C$, $d_i \le \epsilon d_C$ and $h_i \le \epsilon h_C$, and $f_C(i) := \infty$ otherwise.

    $\AM_C$ sorts the items of $T$ in non-decreasing order of profit/volume and selects the largest prefix with volume not exceeding $(1-3\eps)w_Ch_Cd_C$. We pack this prefix using \tNFDH.
    
    \item \textbf{Steinberg \boxes:} These \boxes pack items in layers using \texttt{3D-Vol-Pack}. We describe Steinberg \boxes that pack items in layers along the height  (analogous for other cases). We set $\mathtt{cap}(C) := v(C)/3$. For any item $i$, we set $f_C(i):= v(i)$ if $h_i \le \epsilon h_C$, and either $w_i \le w_C/2$ or $d_i \le d_C/2$ holds. For all other items, we set $f_C(i) := \infty$.

    \kvn{$\AM_C$ first sorts} the items in $T$ in non-increasing order of profit/volume ratio.
    Then, we select the largest prefix whose volume does not exceed $\left(\frac13-2\eps\right)v(C)$, and pack this prefix using \texttt{3D-Vol-Pack}.
    Analogously, container ($4b$) and ($4c$) show packing when layers are stacked along the depth and width, respectively.
\end{itemize}

Next, we establish the packing guarantee of algorithm $\AM_C$ for each container $C$.

\begin{lemma}
    For any $\eps >0$, any container $C$, and any subset $T\subseteq I$ that satisfies $\sum_{i\in T} f_C(i) \le \mathtt{cap}(C)$, the algorithm $\AM_C$ packs a subset $T' \subseteq T$ with $p(T') \ge (1-O(\eps))p(T)$ into the container $C$ in polynomial-time.
\end{lemma}
\begin{proof}
    The case when $C$ is \kvnnew{a Stack container is} easy because the entire set $T$ will be packed in $C$ by $\AM_C$.
For the other types of containers, if $\AM_C$ selects the entire set $T$, we are done. Hence, assume this is not the case.

Consider the case when $C$ is an Area container which uses \NFDH on its front face (the other two types are analogous).
Recall that $\AM_C$ sorts $T$ in order of non-increasing profit/front area ratio and picks the largest prefix $T'$
whose total front area does not exceed $(1-2\eps)w_Ch_C$. Otherwise, since each item in $T$ has a front area at most $\eps^2w_Ch_C$,
we obtain that the total front area of $T'$ is at least $(1-3\eps)w_Ch_C$. Further, since $T$ has a total front area of at most $w_Ch_C$ and since $T'$ was obtained
after sorting $T$ in the profit/front area ratio, it follows that $p(T')\ge (1-3\eps)p(T)$.

Now, we consider the case when $C$ is a Volume container. Again, recall that $\AM_C$ selects the largest prefix $T'$, after arranging $T$ according to non-increasing profit density,
such that $v(T')\le (1-3\eps)v(T)$. Since each item in $T$ has a volume at most $\eps^3w_Ch_Cd_C$, it follows that $v(T')\ge (1-4\eps)v(T)$,
which further implies that $p(T')\ge (1-4\eps)p(T)$.

Finally, we consider the case when $C$ is a Steinberg container that packs layers stacked along the height (other two types are similar). 
The set $T'\subseteq T$ picked by $\AM_C$ is the largest prefix of $T$, arranged in non-increasing order of profit/base area ratio, whose
total volume does not exceed $(1/3-2\eps)v(C)$. Since each item in $T$ has a volume at most $\eps v(C)$,
we obtain that $v(T')\ge (1/3-3\eps)v(C)$, which in turn implies that $p(T')\ge (1-9\eps)p(T)$.
\end{proof}


\subsection{Our Structural Lemma}
Let $\optgs$ be the maximum profit of a guessable container packing. The main structural result of our paper is the following.

\begin{restatable}{lemma}{mainlemma}
\label{thm:mainthm}
For any $\epsilon > 0$, the optimal profit $\opt$ and the optimal profit of a guessable container packing $\optgs$ satisfies $\opt \le \left(\frac{139}{29}+\epsilon\right)\optgs$.
\end{restatable}

Together with \Cref{thm:ptasboxpacking}, this gives us the following theorem.

\begin{theorem}
    \ari{For any constant $\epsilon >0$, there is a polynomial-time $\left(\frac{139}{29}+\epsilon\right)$-approximation algorithm for \tdk.}
\end{theorem}
We prove \Cref{thm:mainthm} in \cref{sec:structurallemma} by establishing several lower bounds on $\optgs$. Specifically, we present five different ways to restructure OPT into a guessable container packing. Our first lower bound is based on a structural result of the 3D Strip Packing problem due to \cite{3d-strip-packing}.
The second lower bound is obtained by restructuring the packing of items in $\text{OPT}\cap (I_{1\ell}\cup L)$ (or similarly $\text{OPT}\cap (I_{2\ell}\cup L)$ or $\text{OPT}\cap (I_{3\ell} \cup L$)). The third and fourth lower bounds are based on volume-based arguments, \ari{using our simple packing algorithm \texttt{3D-Vol-Pack}.} Our final lower bound is based on a structural result from the 2D Knapsack problem \cite{2dks-Lpacking}.

\section{Proof of Structural Lemma}
\label{sec:structurallemma}
Now we prove the structural lemma. 
\paragraph{Classification of items.}
Let $\epsilon>0$ be an accuracy parameter. Let $\mu$ be a sufficiently small constant depending on $\epsilon$ (setting $\mu:= \epsilon^{2^{1/\epsilon^2}}$ suffices). We now classify the input items based on their dimensions. Let $L\subseteq I$ be the set of items whose width, depth and height all exceed $\mu$. Let $I_1\subseteq I$ be the set of items having height at most $\mu$. Similarly, let $I_2\subseteq I\setminus I_1$ be those items whose widths are at most $\mu$, and $I_3 := I\setminus (L\cup I_1 \cup I_2)$ be the remaining items. Thus each item in $I_3$ has depth at most $\mu$.   
Let $I_{1\ell}$ consist of the items in $I_1$ whose width and depth are both more than 1/2, and let $I_{1s} = I_1 \setminus I_{1\ell}$.
See \cref{fig:input-items}.
\begin{figure}
    \centering
    \includesvg[scale=0.4]{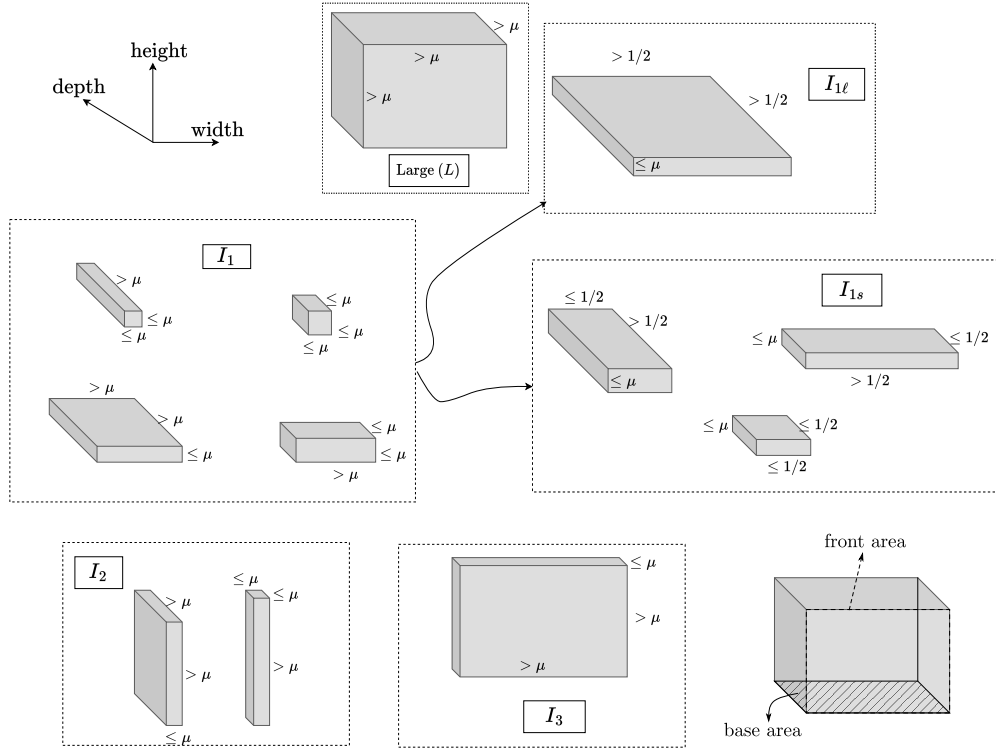}
    \caption{Classification of input items. Front area and base area are shown at the bottom right.}
    \label{fig:input-items}
\end{figure}
The sets $I_{2\ell}, I_{2s}, I_{3\ell}, I_{3s}$ are defined analogously.

Let $\opt_i := p(\text{OPT}\cap I_i)$ for $i\in [3]$, and $\opt_L := p(\text{OPT}\cap L)$, so that $\opt = \opt_1+\opt_2+\opt_3+\opt_L$.
Let $\optonel := p(\text{OPT}\cap I_{1\ell})$ and $\optones := p(\text{OPT}\cap I_{1s})$, so that $\optonel+\optones=\opt_1$. The quantities $\opttwol, \opttwos$ and $\optthreel, \optthrees$ are defined analogously for the sets $I_2$ and $I_3$, respectively.

Finally, let $v_1 := v(\text{OPT}\cap I_1)$, and define $v_2, v_3$ analogously.  Clearly, $v_1, v_2, v_3$ sum up to at most the volume of the knapsack which is 1. We define $\honel := h(\text{OPT}\cap I_{1\ell})$ and $\vones := v(\text{OPT}\cap I_{1s})$. The quantities $\vtwos,\vthrees$ are defined analogously.

\subsection{First lower bound on \texorpdfstring{$\optgs$}{optgs}}
\label{sec:firstlb}
We consider the packing of the items in $\text{OPT}\cap I_1$ (i.e., items with height at most $\mu$) whose total profit is $\opt_1$. We build on the structure derived from the 3D Strip Packing algorithm of Jansen and Pr{\"a}del \cite{3d-strip-packing}. 
However, 
we need several refinements to transform the packing of \cite{3d-strip-packing} into a guessable container packing. 




First, we delete a set of (medium) items from the packing at a small loss in profit.

\begin{lemma}
\label{lem:deletemediumitems}
    \ari{There exists a $\delta \in \left(\epsilon^{2^{O(1/\epsilon)}}, \epsilon\right]$ such that all items in $\text{OPT}\cap I_1$ whose widths or depths lie in the range $[\delta^{10},\delta)$ have a total profit of at most $\epsilon\opt_1$.}
\end{lemma}
\begin{proof}
    Set $k=2/\varepsilon$ and define constants $\delta_1,\ldots \delta_{k+1}$ such that $\delta_1=\varepsilon$ and $\delta_i=\delta_{i-1}^{10}.$ Consider the $k$ intervals $[\delta_{i+1}, \delta_i)$ and assign items to them based on their width or depth.  Each item can be in at most two of these intervals, one for its width and one for its depth. Thus, the sum of profits in all of the intervals is at most $2\opt_1.$  Therefore, by pigeonhole principle, there exists a pair $(j+1, j)$ such that the items in $OPT\bigcap I_1$ with widths or depths in the interval $[\delta_{j+1}, \delta_j)$ have a total profit of at most $\frac{1}{k}2\opt_1=2\cdot \frac{\varepsilon}{2}\opt_1=\eps \opt_1$ Setting $\delta=\delta_{j}$ yields $\delta^{10}=\delta_{j+1},$ thereby completing the proof. We can lower bound $\delta$ by $\eps^{(10)^{k}}=\eps^{10^{O(\frac{1}{\eps})}}.$
\end{proof}

We guess the value of $\delta$ guaranteed by the above lemma. Then the items are divided into four classes depending on their dimensions. An item $i$ is \emph{big} if both $w_i \ge \delta$ and $d_i \ge \delta$, \emph{wide} if $w_i \ge \delta$ and $d_i < \delta^{10}$, \emph{long} if $w_i < \delta^{10}$ and $d_i \ge \delta$, and \emph{small} if both $w_i < \delta^{10}$ and $d_i < \delta^{10}$. The following lemma is obtained by an adaptation of the result of \cite{3d-strip-packing}.

\begin{restatable}{lemma}{jansenpradelSP}
\label{lem:jansenpradelSP}
    There exist disjoint sets $T_1,T_2\subseteq \text{OPT}\cap I_1$ having a profit of at least $\left(\frac{2}{3}-O(\epsilon)\right)\opt_1$, such that
    \begin{itemize}
        \item Items of $T_1$ are packed into a collection of $O(1/\delta^6)$ configurations, where each configuration is a box having a $1\times 1$ base, that is partitioned 
    into  $O(1/\delta^2)$ (cuboidal) slots for packing big items, $O(1/\delta^6)$ slots for packing wide and long items (using \NFDH), and $O(1/\delta^5)$ slots for packing small items (using \tNFDH). The height of a slot is the same as the height of the corresponding configuration,
        \item The configurations are stacked one on top of the other inside the knapsack and their heights sum up to at most $1-2\epsilon$,
        \item The set $T_2$ contains only wide, long and small items and has a total volume of $O(\delta^4)$. 
    \end{itemize}
\end{restatable}
\begin{proof}
    First we delete a set of items by losing an $(\frac{1}{3}+O(\epsilon))$-fraction of profit so that the remaining ones fit within a height of $\frac{2}{3}-4\epsilon$.

\begin{claim}
\label{lem:optspsmall}
     By discarding items having a total profit of at most $\left(\frac{1}{3}+O(\epsilon)\right)\opt_1$, the remaining items of $\text{OPT}\cap I_1$ can be packed within a height of at most $\frac{2}{3}-4\epsilon$.
\end{claim}
\begin{proof}
    We consider two horizontal planes at heights of $1/3$ and $2/3$ from the base of the knapsack, which partition it into three regions as shown in \Cref{fig:strip-removal}. Then one of these regions must satisfy the property that the items completely lying inside the region have a total profit of at most $\opt_1/3$. We discard all such items that lie completely in the interior of the region. Thus the remaining items can be packed within a height of at most $\frac{2}{3}+2\mu\le \frac{2}{3}+2\epsilon$, where the extra $2\mu$ accounts for the heights of the items that are cut by the two planes. Next we consider horizontal strips of height $6\epsilon$ each, and discard all items intersecting the minimum profitable strip (see \Cref{fig:strip-removal}). The profit lost by this is only $O(\epsilon)\opt_1$, and the surviving items are all packed within a height of at most $\frac{2}{3}-4\epsilon$.
    \begin{figure}
    \centering
    \includesvg[width=\linewidth]{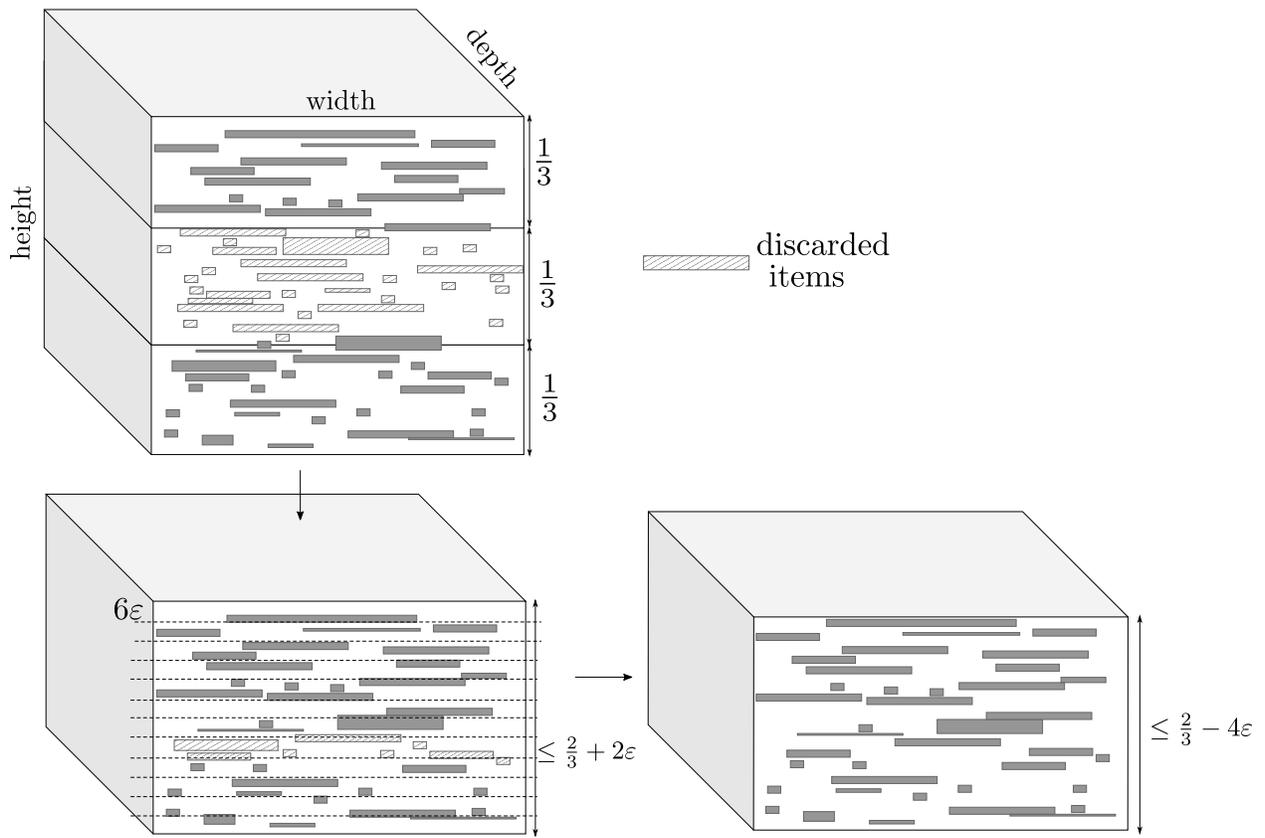}
    \caption{Removing low profitable strips to obtain a packing within a height of $2/3-4\eps$.}
    \label{fig:strip-removal}
\end{figure}
\end{proof}
\begin{figure}
    \centering
		\includesvg[width=0.9\linewidth]{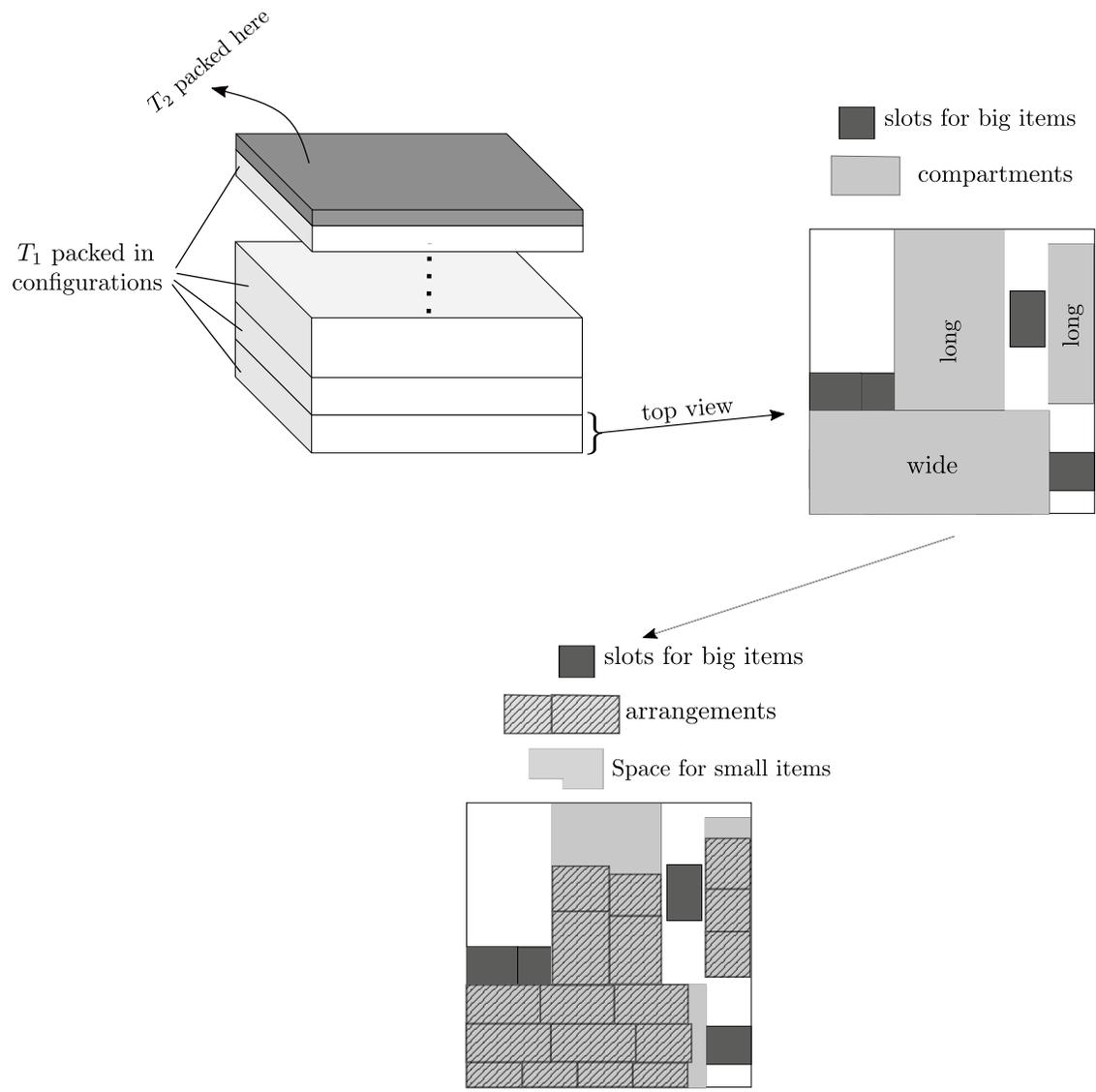}
    \caption{The division of configurations into containers which are further divided into arrangements.}
    \label{fig:jp-configurations}
\end{figure}

Let $\text{OPT}'_1$ be the items that remain after applying \Cref{lem:optspsmall}. 
The algorithm of \cite{3d-strip-packing} repacks a subset $T_1 \subseteq \text{OPT}'_1$ into a collection of $m = O(1/\delta^6)$ {\em configurations}. We first describe the packing inside the configurations. Each such configuration $C$ is a box of height $h_C$ with a base of dimensions $1\times 1$ that is partitioned into at most $1/\delta^2$ {\em compartments} for big items, and at most $O(1/\delta^3)$ {\em compartments} for packing the wide, long, and small items (see \cref{fig:jp-configurations}).\footnote{Compartment is just a rectangular region.
In \cite{3d-strip-packing}, the term container was used to refer to compartment. In this work, we use compartment to avoid ambiguity.} 
\kvn{\cite{jansen-pradel-bin-packing,3d-strip-packing} also ensure that the depth (resp. width) of a wide (resp. long) compartment is at least $\delta^4$.}
Since the items of $\text{OPT}'_1$ were packed within a height of at most $\frac{2}{3}-4\epsilon$ due to \Cref{lem:optspsmall}, the result of \cite{3d-strip-packing} implies that the sum of the heights of all the $m$ configurations is at most $\left(\frac{3}{2}+\epsilon\right)\left(\frac{2}{3}-4\epsilon\right)+\epsilon+
f(1/\eps)\mu$, where $f$ is a function that is independent of $n$.  
This is at most $1-3\epsilon$, for $\mu \le \eps / f(1/\eps)$ and sufficiently small $\epsilon$.


For packing items inside the configurations, first the height of each configuration is extended by $\mu$. The big items are packed into their respective compartments by stacking them one on top of the other, using a result of Lenstra, Shmoys and Tardos \cite{lenstra-scheduling}.
We now describe the packing inside the wide compartments. The packing in the long compartments is analogous.


\begin{figure}
    \centering
        \includesvg[width=0.5\linewidth]{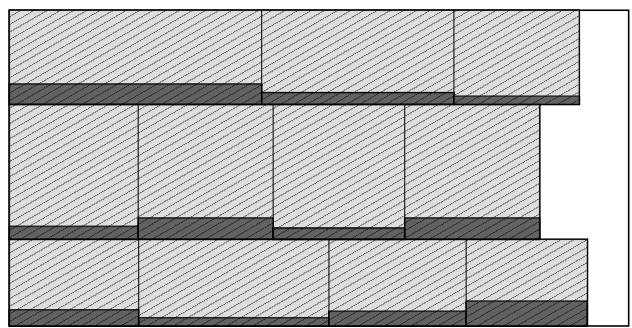}
    \caption{Showing three arrangements (hatched) inside a compartment. Notice how the width of each arrangement is exactly equal to the total width of a multiset of wide items.}
\label{fig:arrangements}
\end{figure}

\begin{figure}
    \centering
    \includesvg[width=\linewidth]{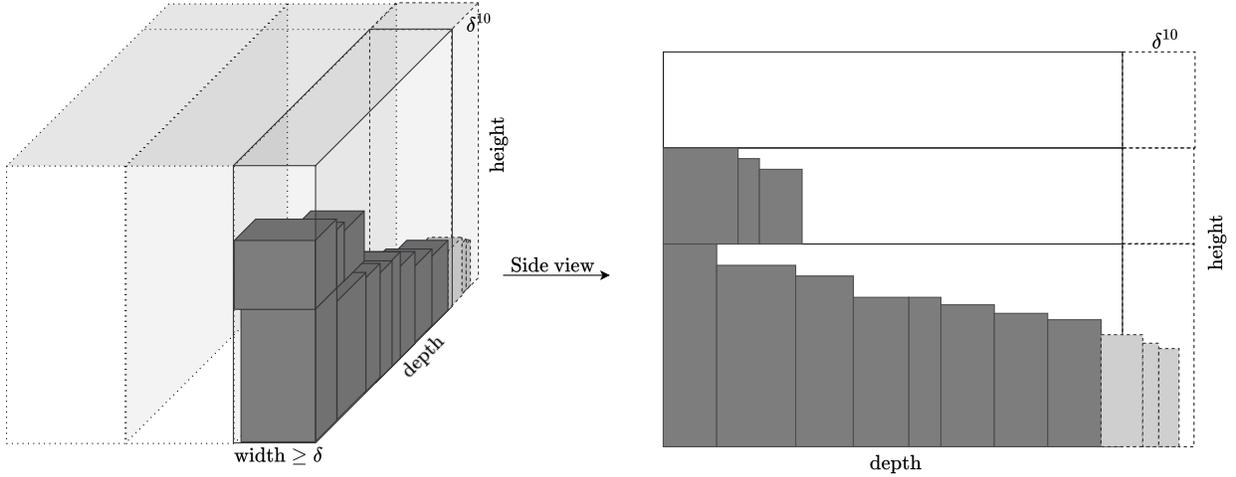}
    \caption{Packing wide items using \NFDH~inside an extended slot}
\label{fig:nfdhcrossing}
\end{figure}

Each wide compartment is divided into at most $O(1/\delta^2)$ \textit{arrangements}.\footnote{Again, in \cite{3d-strip-packing, jansen-pradel-bin-packing}, the term configuration was used instead of arrangement.} Each arrangement is a region whose width is equal to some multi-set of wide items that have a total width of at most the width of the compartment (see \cref{fig:arrangements}). Since the width of any wide item is at least $\delta$, each arrangement consists of at most $1/\delta$ slots. 
To pack the wide items inside the slots, in the algorithm of \cite{3d-strip-packing}, first the depth of the slot is extended by $\delta^{10}$, and then the items are packed using \NFDH~inside the extended slot (see \Cref{fig:nfdhcrossing}). 
For each slot width, it happens exactly once that we run out of items while packing using \NFDH. In that case, we swap the corresponding slot with the rightmost slot in the arrangement. If this happens for multiple slots inside an arrangement, then we sort the slots from left to right in order of non-increasing packing heights. Also each time this happens, the whole configuration is split into two at the corresponding height, and the height of the lower configuration is extended by $\mu$ so that the cut items still fit inside the configuration. The number of configurations only increases by $O(1/\delta^2)$ as a result of this. The sum of heights of all the resulting configurations is bounded by $(1-3\eps)+\mu\cdot O(1/\delta^6) \le 1-2\eps$, for sufficiently small $\mu$.

Next, the space on the right side of the arrangements is partitioned into \emph{slots} for packing the small items. Since each compartment contains $O(1/\delta^2)$ arrangements and a configuration has $O(1/\delta^3)$ compartments, there are $O(1/\delta^5)$ slots for the small items inside each configuration. For each such slot, the width and depth are extended by $\delta^{10}$ and the items are packed using \tNFDH~inside the extended slot. The packing inside the long compartments is done in an analogous way. We next show that the overshooting items have negligible volume.

\begin{figure}[t]
    \centering
    \begin{tikzpicture}[scale=0.5]
                \draw[color=white,dashed,fill=white!75!black] (0,5)--(0,6.5)--(10.5,6.5)--(10.5,0)--(9,0)--(9,5)--cycle;
                \draw [very thick] (0,0) rectangle (9,5);
                \draw[very thick, dashed] (0,5)--(0,6.5)--(10.5,6.5)--(10.5,0)--(9,0);
                \node at (-0.75, 6) {$\delta^{10}$};
                \node at (10, -0.75) {$\delta^{10}$};  
                \foreach \x/\y/\xx/\yy in {
                    0/0/1.5/1.5,
                    1.5/0/3.1/1.45,
                    3.1/0/4.5/1.43,
                    4.5/0/5.8/1.40,
                    5.8/0/7.2/1.38,                   
                    0/1.5/1.3/2.8,
                    1.3/1.5/2.5/2.75,
                    2.5/1.5/4/2.73,
                    4/1.5/6/2.7,
                    6/1.5/7.5/2.68,
                    7.5/1.5/9/2.67,
                    0/2.8/1.4/4.65,
                    1.4/2.8/2.7/4.63,
                    2.7/2.8/4.1/4.61,
                    4.1/2.8/5.5/4.58,
                    5.5/2.8/6.8/4.55,
                    6.8/2.8/8.2/4.5
                }
                {
					\drawGreenItem{\x}{\y}{\xx}{\yy};
				}
                    \foreach \x/\y/\xx/\yy in {
                    7.2/0/9.1/1.35,
                    9.1/0/9.8/1.33,
                    8.2/2.8/9.5/4.45,
                    0/4.65/1.4/5.8,
                    1.4/4.65/2.6/5.78,
                    2.6/4.65/4.1/5.76,
                    4.1/4.65/5.3/5.74,
                    5.3/4.65/6.6/5.72,
                    6.6/4.65/7.3/5.6,
                    7.3/4.65/8.1/5.55,
                    8.1/4.65/9.1/5.52,
                    9.1/4.65/9.8/5.4
                }
                {
                    \drawYellowItem{\x}{\y}{\xx}{\yy};
                }                           
            \end{tikzpicture}
            \caption{Items overshooting the slot are shown as darker. These items can be repacked using \tNFDH.}
\label{fig:smallovershootrepack}
\end{figure}
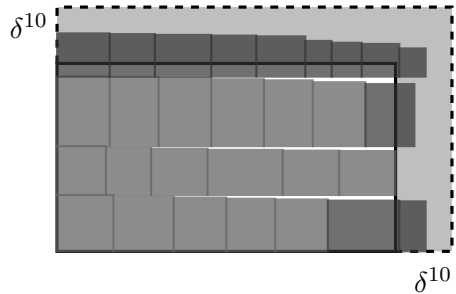

The set $T_2$ consists of the wide, long and small items that overshoot their respective slots, together with all the small items that are packed in slots of width (resp. depth) smaller than $\delta^9$ inside the wide (resp. long) compartments. Consider first a (extended) slot for the wide items. Since the depth of the extension as well as the depth of any wide item are bounded by $\delta^{10}$, the total depth of the items not completely lying inside the original slot is at most $2\delta^{10}$. Since there are $O(1/\delta^6)$ such slots inside each configuration and the configurations have a total height not exceeding 1, the volume of the overshooting wide items is bounded by $2\delta^{10}\cdot O(1/\delta^6) = O(\delta^4)$. Analogously, the volume of the overshooting long items is also at most $O(\delta^4)$. Next consider a slot for small items. The total base area of the overshooting small items of any layer is at most $4\delta^{10}$ (see \Cref{fig:smallovershootrepack}). Thus their total volume is bounded by $4\delta^{10}\cdot O(1/\delta^5)=O(\delta^5)$. Finally, consider a slot of width (resp. depth) less than $\delta^9$ in which small items are packed inside a wide (resp. long) compartment. The total volume of all such slots is again bounded by $\delta^9\cdot O(1/\delta^5)=O(\delta^4)$. Hence we have $v(T_2)\le O(\delta^4)$ and are done.
\end{proof}


Our objective is to build Area and Volume containers out of the slots for packing the wide/long and small items, respectively. The trouble is that the slot dimensions might not be large enough compared to the dimensions of the packed items, which is a necessary condition for the construction of such containers (see \Cref{sec:cont-class}). 
We obtain a guessable container packing by the following lemma.


\begin{restatable}{lemma}{roundingconfigTone}
\label{lem:roundingconfigTone}
    By losing a profit of at most $\delta \opt_1$, and rounding the height of each configuration to the next integer multiple of $\mu/\epsilon$, each configuration for $T_1$ can be further partitioned into $O(1/\delta^6)$ guessable containers. Moreover, the items of $T_2$ can be packed inside a Steinberg container of height $\eps$.
\end{restatable}
\begin{proof}
We begin by discretizing the depths of the slots inside the wide compartments.

\begin{claim}
\label{lem:slotsandregions}
     By losing a profit of at most $O(\delta)\opt_1$, the depths of all the slots for wide and small items inside a wide compartment can be assumed to be a multiple of $\delta^9$.
\end{claim}
\begin{proof}
 We partition the compartment into strips each of depth $\delta^5$ and discard the minimum profitable strip, see \cref{fig:minStrip} for an illustration. Note that the depth of the compartment is a multiple of $\delta^4$, and hence the number of strips is at least $1/\delta$. Since each item has depth at most $\delta^{10}$, any item intersects at most two strips, and hence the total profit of the deleted items is at most an $O(\delta)$-fraction of the profit of the compartment. For any slot that is intersected by the deleted strip, we push the remaining items towards the front face of the slot, and notice that
 in the resulting packing, the wide and small items are packed using \NFDH (ignoring widths) and \tNFDH, respectively.

 \begin{figure}
    \centering
    \includesvg[width=0.7\linewidth]{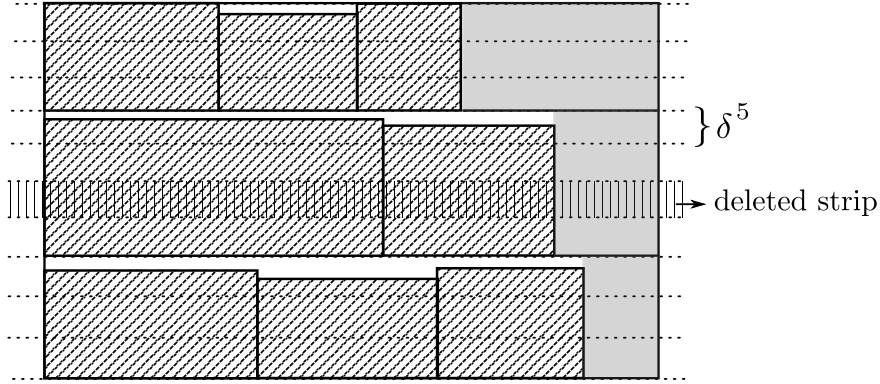}
    \caption{Rounding the depths of slots in a wide compartment. \kvnnew{The gray regions contain the small items.} All items intersecting the deleted strip are discarded.}
    \label{fig:minStrip}
\end{figure}
    
The resulting free space of depth $\delta^5$ inside the compartment will now be used to increase the depth of each arrangement. Specifically, we round up the depth of each arrangement (and hence each slot) to the next integral multiple of $\delta^9$. Since the total number of arrangements is $O(1/\delta^2)$, the total increase in depth is bounded by $O(\frac{1}{\delta^2})\cdot \delta^9 < \delta^5$.  
\end{proof}

In an analogous way, the widths of the slots inside long compartments can be made integral multiples of $\delta^9$. We next round up the heights of each of the configurations to integer multiples of $\mu/\epsilon$. Finally, we create the containers out of the configurations.

Recall that each configuration comprises $1/\delta^2$ compartments for big items, and $O(1/\delta^3)$ compartments for packing the wide, long, and small items. We create a Stack container corresponding to each of the compartments for the big items. The width (resp. depth) of this container is the same as the width (resp. depth) of some big item. Each wide compartment was further partitioned into $O(1/\delta^3)$ slots for packing the wide and small items. Note that due to \Cref{lem:slotsandregions}, the depth of a wide/small item is at most a $\delta$-fraction of the depth of the corresponding slot. Also, since we put the small items packed inside slots of width smaller than $\delta^9$ in the set $T_2$, the width of any small item in the set $T_1$ is at most a $\delta$-fraction of the width of the slot in which it is packed. Finally, the height of any configuration is at least $\mu/\epsilon$, which is at least an $1/\epsilon$-factor larger than the height of any item in $I_1$. We create an Area container corresponding to each slot for the wide items and a Volume container for each slot for the small items. The width of an Area container is the same as the width of some item in $I_1$. Also since the width of a wide compartment can have  $O(1/\delta^2)$ distinct values, and the width of a Volume container is either the same as the width of the wide compartment or equals the difference of the width of the compartment and the width of some slot (Area container), the width of a Volume container also lies in a polynomial-sized set. An analogous thing is done for the slots inside the long compartments. Finally, the heights of each of the  Stack, Area, and Volume containers are set to be the height of the configuration. Overall, we obtain $O(1/\delta^6)$ guessable containers corresponding to each configuration.

Since $T_2$ does not contain any big item and $v(T_2)\le O(\delta^4)$ by \Cref{lem:jansenpradelSP}, the conditions of \Cref{lem:steinberg} are trivially satisfied, and thus the whole set $T_2$ can be packed into a Steinberg container of height $\eps$.
\end{proof}

After applying the above lemma, the rounded heights of all the configurations for $T_1$ sum up to at most $(1-2\epsilon)+\frac{\mu}{\epsilon}\cdot O(\frac{1}{\delta^{6}}) \le 1-\eps$, for sufficiently small $\mu$. Since items of $T_2$ are packed inside a Steinberg container of height $\eps$, all the containers fit inside the knapsack and we obtain the following theorem overall.

\begin{theorem}
\label{thm:guessableconfig}
    We can repack a subset of $\text{OPT}\cap I_1$ with profit at least $\left(\frac{2}{3}-O(\epsilon)\right)\opt_1$ into \ari{$O(1/{\delta^{12}})$} guessable Stack, Area, and Volume containers. 
\end{theorem}

An equivalent restructuring of $\text{OPT}\cap I_2$ and $\text{OPT}\cap I_3$ gives the following corollary.

\begin{restatable}{corollary}{structuredpackingalgo}
\label{thm:structuredpackingalgo}
    $\optgs \ge \left(\frac{2}{3}-O(\epsilon)\right) \max\{\opt_1, \opt_2, \opt_3\}$.
\end{restatable}

\subsection{A simple \texorpdfstring{$(5+\eps)$}{5+ε}-approximation}
\label{sec:5approx}
We now present a simple $(5+\epsilon)$-approximation for \tdk, which already improves upon the $(7+\epsilon)$-approximation \cite{3d-knapsack}. For this, first, we provide the following packing guarantee that will also be useful later.

\begin{lemma}
\label{lem:stackandsteincombine}
    Given a box $B$ whose base has dimensions $w_B\times d_B$ and height 1. Let $T$ be a set of items with $v(T)\le v(B)/4$ such that each $i\in T$ has $w_i \le w_B$, $d_i \le d_B$ and  $h_i\le \epsilon^4$. Let $T_{\ell} := \{i\in T \mid w_i > w_B/2 \text{ and } d_i > d_B/2\}$ and $T_s := T\setminus T_{\ell}$. Then it is possible to find a subset $T'_{\ell}\subseteq T_{\ell}$ with $p(T'_{\ell}) \ge (1-\epsilon)p(T_{\ell})$ and partition $B$ into a Stack and a Steinberg container by a horizontal plane parallel to the base, such that $T'_{\ell}$ and $T_s$ can be fully packed into the Stack and Steinberg containers, respectively. Furthermore, the heights of the two containers are integer multiples of $\epsilon^2$.
\end{lemma}
\begin{proof}
    To begin with, we create a Steinberg container having height $3(1+7\epsilon^2)\frac{v(T_s)}{w_Bd_B}+\epsilon^2$. Since $h_i \le \epsilon^4$, we have that the height of any item is at most an $\epsilon^2$-fraction of the height of the container. Clearly, the volume of the container is at least $3(1+7\epsilon^2)v(T_s)$. This implies that the whole set $T_s$ can be packed into the container since 
    \[ \left(\frac{1}{3}-2\epsilon^2\right)\cdot 3(1+7\epsilon^2)v(T_s) \ge v(T_s)\]
    for sufficiently small $\epsilon$. The height of the Steinberg container is then rounded up to the nearest integral multiple of $\epsilon^2$ - the resulting height is bounded by $\frac{3v(T_s)}{w_Bd_B}+O(\epsilon^2)$. In the remaining height of the box, we create a Stack container. Next, we sort the items of $T_{\ell}$ in non-increasing order of their profit to height ratio and pick a maximal prefix $T'_{\ell}$ that fits inside the Stack container. If $T'_{\ell} = T_{\ell}$, we already obtain a profit of $p(T_{\ell})$ and are done. Else, $h(T_{\ell})$ must be at least $1-\frac{3v(T_s)}{w_Bd_B}-O(\epsilon^2) \ge \frac{1}{5}$, since $v(T_s) \le v(B)/4$. Also, since the height of each item is bounded by $\epsilon^4$, we have 
    \begin{equation}
    \label{eqn:heightpacked}
        h(T'_{\ell}) \ge 1-\frac{3v(T_s)}{w_Bd_B}-O(\epsilon^2) - \epsilon^4 = 1-\frac{3v(T_s)}{w_Bd_B}-O(\epsilon^2).
    \end{equation}
    Now since $v(T) \le v(B)/4$, we have $\frac{1}{4}w_Bd_Bh(T_{\ell}) + v(T_s) \le \frac{1}{4}w_Bd_B$, and thus
    \begin{equation}
    \label{eqn:heightopt}
        h(T_{\ell}) \le 1 - \frac{4v(T_s)}{w_Bd_B}\le 1-\frac{3v(T_s)}{w_Bd_B}.
    \end{equation}
    From \eqref{eqn:heightpacked} and $\eqref{eqn:heightopt}$, we get $h(T'_{\ell}) \ge h(T_{\ell})-O(\epsilon^2) \ge (1-\epsilon)h(T_{\ell})$, since $h(T_{\ell})\ge 1/5$. Hence $p(T'_{\ell})\ge (1-\epsilon)p(T_{\ell})$ and we are done.   
    \kvnr{Low priority. Display math punctuation is not correct.}
\end{proof}




Then the following lemma, together with \Cref{thm:ptasboxpacking}, gives a $(5+\epsilon)$-approximation.

\begin{lemma}
$\opt\le (5+\epsilon)\alg$.    
\end{lemma}
\begin{proof}
Trivially we have $\optgs \ge \opt_L$, since we can create a \kvnnew{Stack} container for every item in $\text{OPT}\cap L$ (at most $1/\mu^3$ such items). Hence if $\opt_L > \opt/5$, we are already done. Otherwise we have $\opt_1+\opt_2+\opt_3 \ge \frac{4}{5}\opt$. Now since $v_1 + v_2 + v_3 \le 1$, there must exist an $i \in [3]$ such that $\frac{\opt_i}{v_i} \ge \frac{4}{5}\opt$ -- w.l.o.g.~assume $i=1$. We have two cases.


\textbf{Case 1: $v_1 \le \frac{1}{4}$.} In this case, \Cref{lem:stackandsteincombine} (with the box $B$ corresponding to the whole knapsack) implies that $\optgs \ge (1-\epsilon)\opt_1$ (since each item in $I_1$ has height at most $\mu < \epsilon^4$). Hence if $\opt_1 > \opt/5$, we are done. Otherwise we have $\opt_2 + \opt_3 \ge \frac{3}{5}\opt$.  W.l.o.g.~ assume that $\opt_2 \ge \frac{3}{10}\opt$, \Cref{thm:structuredpackingalgo} gives us $\optgs \ge \left(\frac{2}{3}-O(\epsilon)\right)\cdot \frac{3}{10}\opt = \left(\frac{1}{5}-O(\epsilon)\right)\opt$.


\textbf{Case 2: $v_1 > \frac{1}{4}$.} In this case, we sort the items of $\text{OPT}\cap I_1$ in non-increasing order of their profit to volume ratio, and pick the maximal prefix $T$ whose volume does not exceed $\frac{1}{4}$. Then we have $v(T) \ge \frac{1}{4}-\mu$. Applying \Cref{lem:stackandsteincombine} to $T$, we obtain
$\optgs \ge (1-\epsilon)\left(\frac{1/4-\mu}{v_1}\right)\opt_1 \ge \left(\frac{1}{5}-O(\epsilon)\right)\opt$,
thus completing the proof.
\end{proof}


\mt{The ratio of $(5+\eps)$ is tight. One such instance is given when $v_1 =\frac{1}{8}, v_2=v_3=\frac{3}{8} $ and $\opt_1=\opt_L=\frac{1}{5}\opt$ and $\opt_2=\opt_3=\frac{3}{10}\opt.$ }

We will obtain further improvements by using the lower bounds in the following sections.

\subsection{Second lower bound on \texorpdfstring{$\optgs$}{optgs}}
\label{sec:secondlb}
In this subsection, we consider the packing of $\OPT\cap (I_{1\ell}\cup L)$, and restructure it to obtain a guessable container packing, by losing only an $O(\epsilon)$-fraction of the profit. 
Let $\OPTonel := \OPT\cap I_{1\ell}$ and $\OPTL := \OPT\cap L$. Since each item in $L$ has all side lengths at least $\hmax$, the number of such items in the packing is $|\OPTL| \le 1/\hmax^3$. 
Also, since the width and depth
of each item in $I_{1\ell}$ is at least 1/2, the items of $\OPTonel$ must be packed in layers, i.e., the top and bottom faces of any item in $\OPTonel$ when extended parallel to the bottom face of the knapsack will not intersect another item from $\OPTonel$. We now discretize the positions of the items in the packing.

\begin{restatable}{lemma}{discretepos}
\label{lem:discretepos}
    The distance of the left (resp. front) face of any item from the left (resp. front) face of the knapsack can be assumed to lie in a polynomial-sized set.
\end{restatable}
\begin{proof}
We push each item as much to the left as possible. Then for each item in $\OPTonel$, the distance of its left face (from the left face of the knapsack) can be written as the sum of the widths of at most $1/\hmax$ items from $L$. Similarly, the distance of the left face of any item in $\OPTL$ is the sum of the widths of at most $1/\hmax$ items of $L$ and at most one item from $\OPTonel$. Analogous claims hold by pushing the items as much to the front as possible.
\end{proof}

\begin{figure}
    \centering
    \includesvg[width=0.7\linewidth]{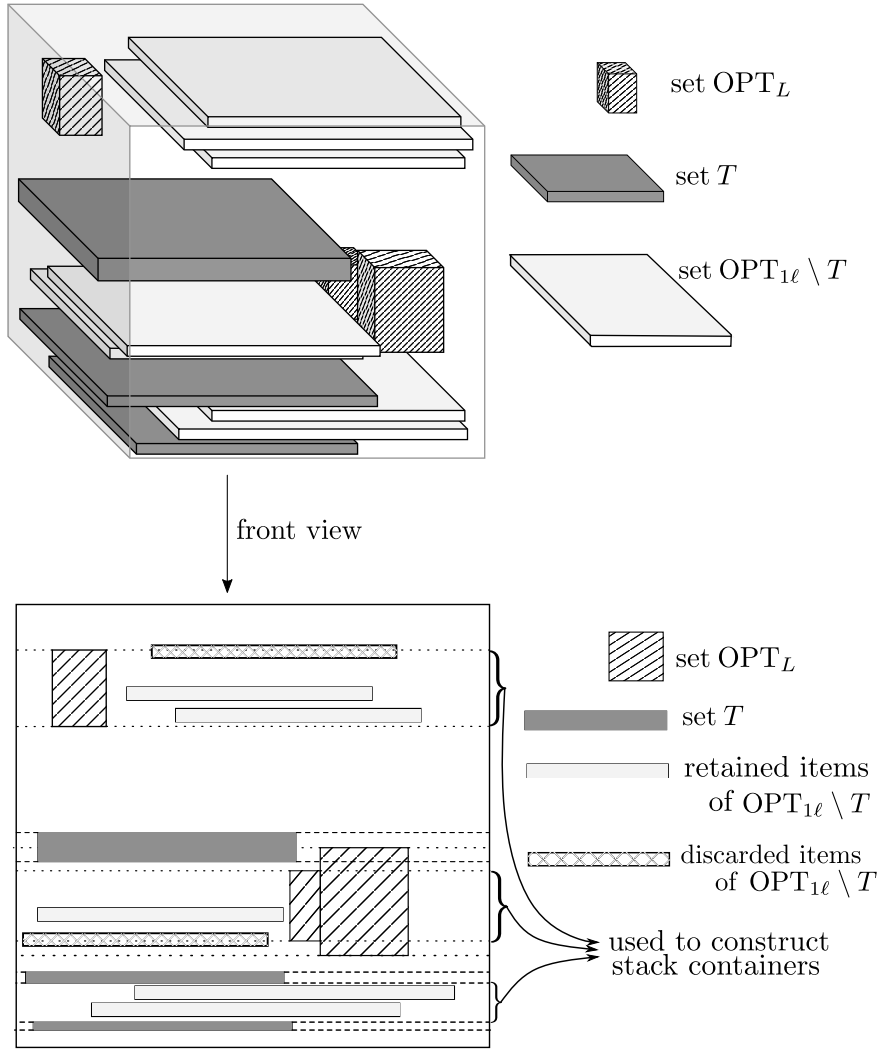}
    \caption{Packing of items in $OPT \cap (L\cup I_{1\ell})$: The dotted (resp. dashed) lines show (front view of) horizontal planes passing through the top and bottom faces of items in $\OPTL$ (resp.  $T$). We discard the items in $\OPTonel \setminus T$ that are cut by these planes.}
    \label{fig:Planes cutting}
\end{figure}

Let $T\subseteq \OPTonel$ be the $k$ items of largest profit in $\OPTonel$, for some $k$ to be chosen later. Our goal is to obtain a container packing of all items in $\OPTL$, together with a subset of $\OPTonel$ having a profit of at least $(1-\epsilon)\optonel$. To this end, we shall pack all items in $T$ and ensure that the number of discarded items of $\OPTonel$ is bounded by $ck$, for some constant $c$. By a result of \cite{jansen1999improved}, the profit of the discarded items can be bounded by $\epsilon \optonel$, for appropriately chosen $k$. We now outline the details of our restructuring procedure. See \cref{fig:Planes cutting}.

We draw horizontal planes passing through the top and bottom faces of the items in $T\cup \OPTL$, and discard the items of $\OPTonel \setminus T$  intersecting these planes (see \cref{fig:Planes cutting}). The number of discarded items is bounded by $2|\OPTL|=O(1/\mu^3)$. The remaining items of $\OPTonel\setminus T$ are now completely packed inside at most $k+2|\OPTL|+1$ horizontal strips demarcated by the planes. Each such strip is a cuboidal box that is pierced from top to bottom by at most $1/\hmax^2$ Large \kvnnew{items}. Further, there are only polynomially-many positions of these Large \kvnnew{items} inside the strip owing to \Cref{lem:discretepos}. We focus on one such strip $\SM$.

\begin{restatable}{lemma}{boundingStackcontainer}
\label{lem:boundingStackcontainer}
    The items of $\OPTonel \setminus T$ present inside the strip $\SM$ can be repacked into at most $O(1/\mu^8)$ Stack containers, whose widths and depths lie in a polynomial-sized set.
\end{restatable}
\begin{proof}
Consider the top face of the strip $\SM$. We extend the boundary edges of the at most $1/\mu^2$ Large \kvnnew{items}, till they hit the boundary of another Large \kvnnew{item} or the boundary of the knapsack. Consider all the rectangular regions formed by these lines. We call a region to be \textit{maximal} if it does not contain any Large \kvnnew{item}, and is not completely contained inside a larger region (see \cref{fig:maximal-regions}).
    \begin{figure}
        \centering

        \includesvg[width=0.7\linewidth]{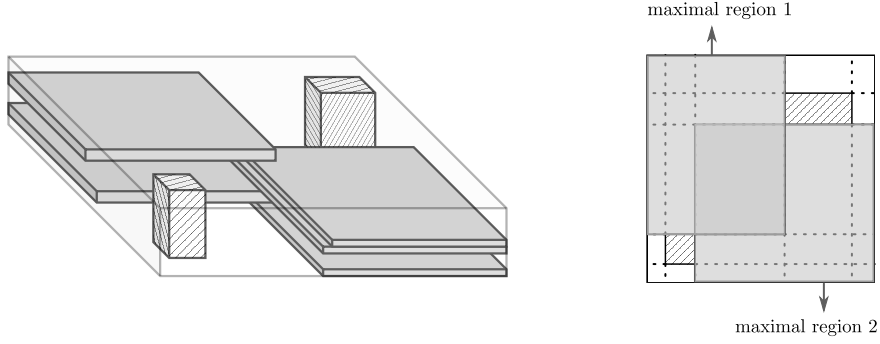}
        \caption{\kvn{Maximal regions. Left: Strip between two consecutive planes. Large items (hatched) span the entire height of the strip. Right: Top view of the strip.}}
        \label{fig:maximal-regions}
    \end{figure}
    
    
    Then note that the top face of each item of $\OPTonel$ that is packed inside $\SM$ must lie completely inside at least one of these maximal regions, and we assign the item to one of these regions arbitrarily. Afterward, we permute the items of $\OPTonel$ so that all items assigned to a particular maximal region are stacked consecutively one above the other. Hence, for each maximal region, we obtain a Stack container containing all items of $\OPTonel$ assigned to it, whose width and depth are the same as those of the region. Since there were at most $1/\mu^2$ Large \kvnnew{items} to begin with, the number of maximal regions and therefore the number of Stack containers obtained is bounded by $O((1/\mu^2)^4) = O(1/\mu^8)$. Also, from \Cref{lem:discretepos}, the possible locations of the Large \kvnnew{items} were polynomially many, and thus there is only a polynomial number of choices for the widths and depths of the Stack containers.
\end{proof}

For each item in \kvnnew{$T\cup \OPTL$}, we create a Stack container containing the single item itself. Together with \Cref{lem:boundingStackcontainer}, we get a constant number of Stack containers inside the knapsack. We now discretize the heights of the Stack containers, so that they lie in a polynomial-sized set.

\begin{restatable}{lemma}{discard}
\label{lem:discard}
    By discarding $O(1/\mu^{11})k$ items from $\OPTonel\setminus T$, the height of each Stack container can be made of the form $j\Bar{h}$, where $j\in [n]$ and $\Bar{h}$ is the height of some item in \kvnnew{$I_{1\ell}\cup L$}.
\end{restatable}
\begin{proof}
If a Stack container contains a single item, we make the height of the container equal to the height of that item. Consider now a Stack container having at least two items of $I_{1\ell}$. Let $\Bar{h}$ be the height of an item with the second largest height. We discard the item with the largest height and use the resulting empty space to round down the height of the container to the nearest integer multiple of $\Bar{h}$. Since the number of items packed inside the container is bounded by $n$, the rounded height of the container cannot exceed $n\Bar{h}$. Finally, since the number of Stack containers inside any strip is $O(1/\mu^8)$ (\Cref{lem:boundingStackcontainer}) and there are at most $k+2|\OPTL|+1 \le O(1/\mu^3)k$ strips, the total number of discarded items is bounded by $O(1/\mu^{11})k$.
\end{proof}

Together with the $O(1/\mu^3)$ items that were discarded previously, \Cref{lem:discard} gives us that the total number of discarded items so far is at most $O(1/\mu^{11})k$. We now use the following lemma from \cite{jansen1999improved}.
\begin{lemma}[\cite{jansen1999improved}]
    Let $p_1\ge p_2\ge \ldots \ge p_n>0$ be a sequence of real numbers and $P=\sum_{i=1}^{n} p_i$. Let $c$ be a nonnegative integer and $\epsilon >0$. If $n = c^{\Omega(1/\epsilon)}$, then there is an integer $k \le c^{O(1/\epsilon)}$ such that $p_{k+1}+\ldots p_{k+ck} \le \epsilon P$.
    
\end{lemma}

The above lemma invoked with $c=O(1/\hmax^{11})$ and with the correct guess of $k$ gives us that the total profit of the discarded items is at most $\epsilon \optonel$, assuming that $|\OPTonel|$ was sufficiently large (at least $n_0 := (1/\mu)^{\Omega(1/\epsilon)}$) to begin with. Otherwise, if $|\OPTonel| \le n_0$, we can create a separate Stack container for each item in $\OPTonel$ having dimensions same as that of the item, which together with the Stack containers for items in $\OPTL$ directly yields a guessable container packing of profit $\opt_L + \optonel$. The following theorem summarizes the arguments of this section.

\begin{restatable}{theorem}{ptasforcombine}
    Consider the packing of $\OPT\cap (I_{1\ell}\cup L)$.
    Then by discarding items of $\OPT\cap I_{1\ell}$ having a total profit of at most $\epsilon\optonel$, the remaining items can be repacked into a guessable container packing, wherein items of $L$ are placed inside at most $1/\hmax^3$ \kvnnew{Stack} containers, and items of $I_{1\ell}$ are packed into at most $(1/\mu)^{O(1/\epsilon)}$ Stack containers.
\end{restatable}



\subsection{Third lower bound on \texorpdfstring{$\optgs$}{optgs}}
\label{sec:thirdlb}


In this subsection, we restructure the packing of items in $\text{OPT}\cap I_1$. Recall that $\optonel = p(\text{OPT}\cap I_{1\ell})$ and $\optones = p(\text{OPT}\cap I_{1s})$, so that $\optonel+\optones = \opt_1$. Also, $v_1 := v(\text{OPT}\cap I_1)$ and $\vones = v(\text{OPT}\cap I_{1s})$.






\begin{restatable}{theorem}{packingI}
\label{thm:packingI1}
    Consider the packing of $\text{OPT}\cap I_1$. It is possible to repack subsets of $\text{OPT}\cap I_{1\ell}$ and $\text{OPT}\cap I_{1s}$ into a guessable Stack and Steinberg container, respectively, whose profit, say $P$, satisfies the following properties.
    \begin{itemize}
        \item $\displaystyle P \ge (1-O(\epsilon))\frac{\opt_{1s}}{\max\{3\vones, 1\}}$.
        \item If $v_1\le 1/3$, then $\displaystyle P  \ge (1-O(\epsilon))\left(\frac{3}{4}\optonel+\optones\right)$.
        \item If $v_1 \le 1/4$, then $P\ge (1-O(\epsilon))\opt_1$.
    \end{itemize}
    
\end{restatable}

In order to prove the above, we first establish the following lemma that will also be used in the next subsection.

\begin{lemma}
\label{lem:stackandstein}
    The following statements hold.
    \begin{enumerate}
        \item If $v_1 \in (2\epsilon, 1/3]$, it is possible to obtain a container packing of a subset $T \subseteq \text{OPT}\cap I_1$ with $p(T) \ge (1-O(\epsilon))\left(\frac{3}{4}\optonel + \optones\right)$, such that the items of $T\cap I_{1\ell}$ and $T\cap I_{1s}$ are packed into a Stack and a Steinberg container, respectively. Both the containers have an $1\times 1$ base and further, their heights are integer multiples of $\epsilon^2$ and they sum to at most $3v_1$.
        \item If $v_1 \le 2\epsilon$, the items of $\text{OPT}\cap I_{1\ell}$ and $\text{OPT}\cap I_{1s}$ can be packed into a Stack and a Steinberg container respectively, both of whose bases have dimensions $1\times 1$, and whose heights are $8\epsilon$ each.
    \end{enumerate}    
\end{lemma}
\begin{proof}
    We first give a proof of the second case, which is much simpler. Since $v_1\le 2\epsilon$, we must have $\honel\le 8\epsilon$, and therefore the items of $\text{OPT}\cap I_{1\ell}$ can be packed into a Stack container with $1\times 1$ base and height $8\epsilon$. Similarly, the items of $\text{OPT}\cap I_{1s}$ can be packed into a Steinberg container with $1\times 1$ base and height $8\epsilon$ using \Cref{lem:steinberg}, since $\mu \le \epsilon\cdot 8\epsilon$ and $\left(\frac{1}{3}-2\epsilon\right)8\epsilon > 2\epsilon \ge \vones$.

    For the first case, we further subdivide the proof into two cases depending on the relative values of $v_1$ and $\vones$.

    \textbf{Case 1 : $v_1-\vones >\epsilon$.} In this case, we first create a Steinberg container $C$ with $1\times 1$ base and height $3(1+7\epsilon^2)\vones + \epsilon^2$. Note that the height of any item $i\in \text{OPT}\cap I_{1s}$ satisfies $h_i \le \mu < \epsilon^4 \le \epsilon^2 h_C$. We show that $\vones \le \left(\frac{1}{3}-2\epsilon^2\right)v(C)$, which implies that the items of $\text{OPT}\cap I_{1s}$ can be fully packed into the container $C$ by \Cref{lem:steinberg}. We have
    \[ \left(\frac{1}{3}-2\epsilon^2\right)v(C) \ge (1-6\epsilon^2)(1+7\epsilon^2)\vones
    > \vones, \]
    for sufficiently small $\epsilon$. We now round up the height of $C$ to the next multiple of $\epsilon^2$. The resulting value of $h_C$ is at most $3\vones+ 9\epsilon^2$. Note that since all items of $\text{OPT}\cap I_{1s}$ are packed, we obtain the full profit of $\optones$ from them. Next, we sort items in $\text{OPT}\cap I_{1\ell}$ in non-increasing order of profit to height ratio, and pick a maximal prefix whose height does not exceed $\epsilon^2 \lfloor3v_1/\epsilon^2\rfloor -h_C$ and pack them into a Stack container of the same height and with an $1\times 1$ base. Note that the height of the Stack container is also an integer multiple of $\epsilon^2$. If we are able to pack all items of $\text{OPT}\cap I_{1\ell}$, we obtain a profit of $\optonel$, and are done. Otherwise, the height of the packed items is at least $3v_1-\epsilon^2 - h_C -\mu\ge 3v_1-3\vones-11\epsilon^2 \ge (1-4\epsilon)(3v_1 - 3\vones)$. The profit obtained from $\text{OPT}\cap I_{1\ell}$ is then at least
    \[ \frac{(1-4\epsilon)(3v_1-3\vones)}{\honel}\optonel \ge \left(\frac{3}{4}-O(\epsilon)\right)\optonel, \]
    where the last inequality follows from the fact that $v_1\ge \frac{1}{4}\honel+\vones$.

    \textbf{Case 2 : $v_1 - \vones \le \epsilon$.} Since $v_1 > 2\epsilon$, it follows that $\vones > \epsilon$ in this case. We group the items of $\text{OPT}\cap I_{1s}$ into maximal groups of volume at most $8\epsilon^2$ each. Since $\vones > \epsilon$, the number of groups is at least $\Omega(1/\epsilon)$ and each group (except possibly one) has a volume of at least $8\epsilon^2-\mu>7\epsilon^2\ge 7\epsilon^2 \vones$. We delete the group having minimum profit. The remaining items have a volume of at most $(1-7\epsilon^2)\vones$ and a profit of at least $(1-O(\epsilon))\optones$. We claim that all these items can be packed into a Steinberg container $C$ having an $1\times 1$ base and height $3\vones-3\epsilon^2$. First note that since $\vones >\epsilon$, $h_C >\epsilon$ and therefore, for any $i\in \text{OPT}\cap I_{1s}$, we have $h_i \le \mu < \epsilon^3 \le \epsilon^2 h_C$. Also
    \[ \left(\frac{1}{3}-2\epsilon^2\right)v(C) \ge (1-6\epsilon^2)(1-\epsilon^2)\vones > (1-7\epsilon^2)\vones\]
    and thus the claim follows using \Cref{lem:steinberg}. We round up the height of the Steinberg container $C$ to the nearest integer multiple of $\epsilon^2$. Next, we create a Stack container of $1\times 1$ base, and having a height such that the sum of the heights of the two containers is exactly $\epsilon^2\lfloor 3v_1/\epsilon^2\rfloor$. This directly ensures that the height of the Stack container is a multiple of $\epsilon^2$. Also note that the height is at least $3v_1-3\vones + \epsilon^2$. We sort the items of $\text{OPT}\cap I_{1\ell}$ in non-increasing order of their profit to height ratio and pick a maximal prefix that has height at most $3v_1-3\vones+\epsilon^2$, which we pack in layers inside the Stack container. If, by doing this, we are able to pack all items of $\text{OPT}\cap I_{1\ell}$, we already obtain a profit of $\optonel$, and are done. Else the total height of the packed items is at least $3v_1-3\vones+\epsilon^2-\mu > 3v_1-3\vones$. Hence, the profit obtained from items of $\text{OPT}\cap I_{1\ell}$ is at least $\frac{3v_1-3\vones}{\honel}\optonel \ge \frac{3}{4}\optonel$, and we are done.
\end{proof}

We are now ready to prove \Cref{thm:packingI1}.

\begin{proof}[Proof of \Cref{thm:packingI1}]
     For the first bound, we sort the items of $\text{OPT}\cap I_{1s}$ in non-increasing order of their profit to volume ratio. We then pick the maximal prefix whose volume does not exceed $\frac{1}{3}-2\epsilon$. Using \Cref{lem:steinberg}, these items can be packed into a Steinberg container of dimensions same as those of the knapsack. We now analyze the profit of this packing. If all items of $\text{OPT}\cap I_{1s}$ are packed, we trivially obtain a profit of $\optones$ from them, and are done. Otherwise, since the volume of each item of $I_{1s}$ is at most $\mu$, the total volume of packed items is at least $\frac{1}{3}- 2\epsilon -\mu > \frac{1}{3}-3\epsilon$. Hence, the profit obtained is at least 
    \[ \frac{\frac{1}{3}-3\epsilon}{\vones}\optones \ge (1-O(\epsilon))\frac{\optones}{3\vones}\]
    and we get the desired guarantee in the first case. 

    The proof of the second and third cases directly follows from \Cref{lem:stackandstein} and \Cref{lem:stackandsteincombine}, respectively.
\end{proof}

\subsection{Fourth lower bound on \texorpdfstring{$\optgs$}{optgs}}
\label{sec:fourthlb}
In the previous sections, we repacked items from exactly one set out of $I_1, I_2, I_3$. In this subsection, we restructure the packing of a subset of items from $I_1$, together with items from $I_2\cup I_3\cup L$. To this end, we introduce a classification of the items in $I_2\cup I_3\cup L$. 
Let $S_2$ (resp. $S_3,S$) be the set of items in $I_2$ (resp. $I_3,L$) with height at most 1/2. Define $\opttwot := p(\text{OPT}\cap S_2)$ and $\opttwoh := p(\text{OPT}\cap (I_2\setminus S_2))$, so that $\opttwot+\opttwoh = \opt_2$. The quantities $\optthreet$ and $\optthreeh$ are defined analogously. Let $\optlt := p(\text{OPT}\cap S)$ and $\optlh := p(\text{OPT}\cap (L\setminus S))$ so that $\optlt+\optlh = \opt_L$. The following theorem repacks items from $I_1\cup S_2\cup S_3\cup S$.

\begin{theorem}
    Consider the packing of $\text{OPT}\cap \left(I_1\cup S_2\cup S_3\cup S\right)$ and assume that $v_1\le 1/6$ and $v_2,v_3>1/4$. Then, there is a container packing of a subset $T$ of these items, having a profit of at least 
    \[ (1-O(\epsilon))\left(\frac{3}{4}\optonel+\optones+\max\left\{\frac{3}{20}(\opttwot+\optthreet),\frac{1}{3}\optlt\right\}\right). \] 
    Further, the packing satisfies the following properties.
    \begin{enumerate}
        \item Items in $T\cap I_{1\ell}$ are packed in a Stack container with $1\times 1$ base and whose height is an integer multiple of $\epsilon^2$.
        \item Items in $T\cap I_{1s}$ are packed in a Steinberg container having an $1\times 1$ base, whose height is also an integer multiple of $\epsilon^2$.
        \item \deb{Only one of the sets $T\cap S_2$, $T\cap S_3$ or $T\cap S$ is non-empty.}    
        \item If $T\cap S_2 \neq \Phi$, then the items of $T\cap S_2$ (analogously $T\cap S_3$) are packed inside a Stack and a Steinberg container both having depth equal to 1. The widths and heights of the containers are again integer multiples of $\epsilon^2$. An analogous result holds if $T\cap S_3 \neq \Phi$.
        \item If $T\cap S \neq \Phi$, then the items of $T\cap S$ are packed as singletons inside Stack containers.
    \end{enumerate}
\end{theorem}
\begin{proof}
    Using \Cref{lem:stackandstein}, we first obtain a container packing of a subset of $\text{OPT}\cap I_1$ into a Stack and a Steinberg container that has a profit of at least $(1-O(\epsilon))(\frac{3}{4}\optonel+\optones)$. It is guaranteed that the sum of the heights of the two containers does not exceed $1/2$ -- when $v_1 \in (2\epsilon, 1/6]$, the sum is at most $3v_1 \le 1/2$, and when $v_1 \le 2\epsilon$, the sum is $16\epsilon \le 1/2$, for sufficiently small $\epsilon$. The empty space above the Stack and Steinberg containers forms a cuboidal box of height, say $H$ (which is a multiple of $\epsilon^2$), in which we shall pack items from one of the sets $\text{OPT}\cap S_2$, $\text{OPT}\cap S_3$ or $\text{OPT}\cap S$. 

    We focus on the packing of $\text{OPT}\cap S_2$. We sort the items in non-increasing order of their profit to volume ratio, and pick a maximal prefix $S'_2$ whose volume does not exceed $H/4$. Note that $H/4 \le 1/4\le v_2$, and therefore $v(S'_2) \ge \frac{H}{4}-\mu > (1-\epsilon)\frac{H}{4}$, since $H\ge 1/2$. Using an analog of \Cref{lem:stackandsteincombine}, a subset of $S'_2$ with profit at least $(1-\epsilon)p(S'_2)$ can be packed into a Stack and a Steinberg container, both having depth 1, height $H$, and widths integral multiples of $\epsilon^2$. The profit thus obtained is at least
    \[ \frac{(1-\epsilon)^2 H/4}{v_2}\opttwot \ge (1-O(\epsilon))\frac{H}{4v_2}\opttwot.\]

    Doing an analogous procedure for the items in $\text{OPT}\cap S_3$, we would similarly obtain a profit of at least $(1-O(\epsilon))\frac{H}{4v_3}\optthreet$. The best out of these two packings will have a profit of at least
    \[(1-O(\epsilon))\frac{H}{4(v_2+v_3)}(\opttwot + \optthreet) \ge (1-O(\epsilon))\frac{H}{4(1-v_1)}(\opttwot + \optthreet).\]

    Now when $v_1 > 2\epsilon$, \Cref{lem:stackandstein} ensures that $H\ge 1-3v_1$, and thus the best of the two above packings has a profit of at least
    \[ (1-O(\epsilon))\frac{1-3v_1}{4(1-v_1)}(\opttwot + \optthreet) \ge \left(\frac{3}{20}-O(\epsilon)\right)(\opttwot + \optthreet).\]
    In the other case when $v_1 \le 2\epsilon$, we have $H \ge 1-16\epsilon$, and the obtained profit is at least 
    \begin{align*}
        (1-O(\epsilon))\frac{1-16\epsilon}{4(1-v_1)}(\opttwot + \optthreet) &\ge \left(\frac{1}{4}-O(\epsilon)\right)(\opttwot + \optthreet)\\ 
        &> \left(\frac{3}{20}-O(\epsilon)\right)(\opttwot + \optthreet).
    \end{align*}

    Finally, we show that it is also possible to pack a subset of $\text{OPT}\cap S$ having profit of at least $\optlt/3$ above the Stack and Steinberg containers for $\text{OPT}\cap I_1$. Taking the best out of three packings for $\text{OPT}\cap S_2, \text{OPT}\cap S_3$ and $\text{OPT}\cap S$ gives the desired profit guarantee of the theorem.  Consider the packing of $\text{OPT}\cap S$. We draw a horizontal plane at a height of $1/2$ from the bottom face of the knapsack. This partitions the items of $\text{OPT}\cap S$ into three groups - those that lie above the plane, those that lie below the plane, and those that are intersected by the plane. Clearly, one of these groups will have a profit of at least $\optlt/3$, and we are done.
\end{proof}

\subsection{Fifth lower bound on \texorpdfstring{$\optgs$}{optgs}}
\label{sec:fifthlb}
In the previous subsection, we repacked items from $S_2, S_3$ or $S$ together with items in $I_1$. 
Now we consider items in $(I_2\setminus S_2)\cup (I_3\setminus S_3)\cup (L\setminus S)$ in OPT. As all these items have heights of more than 1/2, they can not be stacked on top of each other. So, we can ignore their heights, and effectively obtain a 2D Knapsack instance. 
Now we use a container-based algorithm for 2DK from \cite{khan2021approximation} to obtain a container packing of these items. 

\begin{theorem} [\cite{khan2021approximation}]
\label{thm:cont2dk}
Let $S$ denote a set of items (2D rectangles) that can be feasibly packed into a knapsack and let $\eps \in (0,1)$ be any small constant.
Then there exists a subset $S' \subseteq S$ such that $p(S') \ge (1/2-\eps) \cdot p(S)$. The items in $S'$ are packed into the knapsack in containers. Further, the number of containers formed is $O_{\eps}(1)$ and their widths and heights belong to a set whose cardinality is $\text{poly}(|S|)$ and
moreover, this set can be computed in time $\text{poly}(|S|)$.
\end{theorem}

This gives us our final lower bound on $\alg$.

\begin{restatable}{theorem}{twodkthm}
\label{thm:2dknapsack}
    Consider the packing of $\text{OPT}\cap ((I_2\cup I_3\cup L)\setminus (S_2\cup S_3\cup S))$. Then there is a guessable \container packing into $O_{\mu}(1)$ Stack and Area containers each of height 1, having profit at least $\left(\frac{1}{2}-O(\epsilon)\right)(\opttwoh+\optthreeh+\optlh)$.
\end{restatable}

\subsection{Bounding \texorpdfstring{$\optgs$}{optgs}}
\label{subsec:bound}
Equipped with all the lower bounds of the preceding subsections, we establish \Cref{thm:mainthm} by an intricate case analysis depending on the volumes of various sets of items in the optimal packing. 


\mainlemma*

We first state the bounds on $\alg$ that we obtained by our restructuring procedures in the previous subsections.

\structuredpackingalgo*

\begin{corollary}
\label{thm:big&largealgo}
    $\alg \ge (1-O(\epsilon))(\opt_L + \max \{\optonel, \opttwol, \optthreel\})$.
\end{corollary}

\begin{corollary}
\label{cor:packingI1}
    The following statements hold.
    \begin{itemize}
        \item $\alg \ge (1-O(\epsilon))\displaystyle\frac{\optones}{\max\{3\vones,1\}}$.
        \item If $v_1 \le 1/3$, then $\displaystyle\alg \ge (1-O(\epsilon))\left(\frac{3}{4}\optonel + \optones\right)$.
        \item If $v_1 \le 1/4$, then $\alg \ge (1-O(\epsilon))\opt_1$.
    \end{itemize}
    Similar lower bounds also hold for $\alg$ with $(\vones,v_1,\optonel,\optones,\opt_1)$ replaced by \newline \noindent $(\vtwos,v_2,\opttwol,\opttwos,\opt_2)$ or $(\vthrees,v_3,\optthreel,\optthrees,\opt_3)$.
\end{corollary}

\begin{corollary}
\label{cor:combinepacking}
    $\displaystyle\optgs \ge (1-O(\epsilon))\left(\frac{3}{4}\optonel+\optones+\max\left\{\frac{3}{20}(\opttwot+\optthreet),\frac{1}{3}\optlt\right\}\right)$.
\end{corollary}

\begin{corollary}
\label{cor:2dknapsackcontainer}
    $\displaystyle\optgs \ge \left(\frac{1}{2}-O(\epsilon)\right)(\opttwoh+\optthreeh+\optlh)$.
\end{corollary}

We are now ready to prove \Cref{thm:mainthm}. We divide the proof into several cases. 

\subsubsection{At most one out of \texorpdfstring{$v_1,v_2,v_3$}{v1, v2, v3} exceeds \texorpdfstring{$1/3$}{1/3}}
W.l.o.g.~assume that $v_2,v_3 \le 1/3$.

\begin{claim}
\label{cla:9by2}
    $\displaystyle\opt \le \left(\frac{9}{2}+O(\epsilon)\right)\alg$.
\end{claim}
\begin{proof}
    From \Cref{thm:structuredpackingalgo}, we get 
    \begin{equation}
    \label{eqn:2a}
        6\alg \ge (4-O(\epsilon))\opt_1.
    \end{equation}

    From \Cref{thm:big&largealgo}, we get
    \begin{equation}
    \label{eqn:2b}
        4\alg \ge (1-O(\epsilon))(4\opt_L + \opttwol+\optthreel).
    \end{equation}

    Finally, from \Cref{cor:packingI1}, we have the following.
    \begin{eqnarray}
        4\alg &\ge& (1-O(\epsilon))(3\opttwol + 4\opttwos).    \label{eqn:2c}\\
        4\alg &\ge& (1-O(\epsilon))(3\optthreel + 4\optthrees).    \label{eqn:2d}
    \end{eqnarray}

    Adding \eqref{eqn:2a}, \eqref{eqn:2b}, \eqref{eqn:2c} and \eqref{eqn:2d}, we obtain
    \[ 18\alg \ge (1-O(\epsilon))(4\opt_1 + 4\opt_2 + 4\opt_3 + 4\opt_L) = (4-O(\epsilon))\opt,\]
    from which the claim follows.
\end{proof}

\subsubsection{Two out of \texorpdfstring{$v_1,v_2,v_3$}{v1, v2, v3} exceed \texorpdfstring{$1/3$}{1/3}}
W.l.o.g.~we assume $v_1\le1/3$ and $v_2,v_3>1/3$. We further subdivide into two cases depending on the values of $\vtwos$ and $\vthrees$.

\paragraph{At most one out of $\vtwos,\vthrees$ exceeds $1/3$}
W.l.o.g.~we assume $\vtwos \le 1/3$.

\begin{claim}
\label{cla:37by8}
    $\displaystyle\opt \le \left(\frac{37}{8}+O(\epsilon)\right)\alg$.
\end{claim}
\begin{proof}
    From \Cref{thm:structuredpackingalgo}, we obtain
    \begin{equation}
    \label{eqn:3a}
        15\alg \ge (2-O(\epsilon))(\opt_2+4\opt_3).
    \end{equation}

    From \Cref{thm:big&largealgo}, we get
    \begin{equation}
    \label{eqn:3b}
        8\alg \ge (1-O(\epsilon))(8\opt_L + 2\optonel + 6\opttwol).
    \end{equation}

    Finally, \Cref{cor:packingI1} gives the following.
    \begin{eqnarray}
        8\alg &\ge& (1-O(\epsilon))(6\optonel + 8\optones)  \label{eqn:3c}\\
        6\alg &\ge& (6-O(\epsilon))\opttwos  \label{eqn:3d}
    \end{eqnarray}

    Adding \eqref{eqn:3a}, \eqref{eqn:3b}, \eqref{eqn:3c} and \eqref{eqn:3d}, we obtain
    \[ 37\alg \ge (1-O(\epsilon))(8\opt_1 + 8\opt_2 + 8\opt_3 + 8\opt_L) = (8-O(\epsilon))\opt\]
    and the claim follows.
\end{proof}

\paragraph{Both $\vtwos,\vthrees$ exceed $1/3$}
We further subdivide into three cases depending on the value of $v_1$.

\begin{enumerate}
    \item $v_1 > 1/4$: We show the following lemma for this case.

    \begin{lemma}
    \label{lem:optgslb}
        $\displaystyle\optgs \ge \left(\frac{4}{9}-O(\epsilon)\right)(\opttwos+\optthrees)$.
    \end{lemma}
    \begin{proof}
        Since $\vtwos, \vthrees > 1/3$, \Cref{cor:packingI1} implies that $\alg \ge (1-O(\epsilon))\frac{\opttwos}{3\vtwos}$ and $\alg \ge (1-O(\epsilon))\frac{\optthrees}{3\vthrees}$, respectively. Hence,
        \begin{align*}
        \alg \ge& (1-O(\epsilon))\cdot \max\left\{\frac{\opttwos}{3\vtwos},\frac{\optthrees}{3\vthrees}\right\} \ge (1-O(\epsilon))\frac{\opttwos+\optthrees}{3(\vtwos+\vthrees)} \\ \ge& \left(\frac{4}{9}-O(\epsilon)\right)(\opttwos+\optthrees),
        \end{align*}
        where the last inequality follows since $\vtwos+\vthrees \le 1-v_1\le 3/4$.
    \end{proof}

    \begin{claim}
    \label{cla:151by32}
        $\displaystyle\opt \le \left(\frac{151}{32}+O(\epsilon)\right)\alg$.
    \end{claim}
    \begin{proof}
        From \Cref{thm:structuredpackingalgo}, we get
        \begin{equation}
        \label{eqn:4a}
            60\alg \ge (20-O(\epsilon))(\opt_2 + \opt_3).
        \end{equation}

        From \Cref{thm:big&largealgo}, we get
        \begin{equation}
        \label{eqn:4b}
            32\alg \ge (1-O(\epsilon))(32\opt_L + 8\optonel + 12\opttwol + 12\optthreel).
        \end{equation}

        From \Cref{cor:packingI1}, we get
        \begin{equation}
        \label{eqn:4c}
            32\alg \ge (1-O(\epsilon))(24\optonel + 32\optones).
        \end{equation}

        Finally from \Cref{lem:optgslb}, we obtain
        \begin{equation}
        \label{eqn:4d}
            27\alg \ge (12-O(\epsilon))(\opttwos + \optthrees).
        \end{equation}

        Adding \eqref{eqn:4a}, \eqref{eqn:4b}, \eqref{eqn:4c} and \eqref{eqn:4d}, we have that
        \[ 151\alg \ge (1-O(\epsilon))(32\opt_1 + 32\opt_2 + 32\opt_3 + 32\opt_L) = (32-O(\epsilon))\opt,\]
        and are done.
    \end{proof}


    



    \item $v_1 \in (1/6, 1/4]:$ Similar to the previous case, we first show the following lower bound on $\alg$.

    \begin{lemma}
    \label{lem:optgslbcase2}
        $\displaystyle\alg \ge \left(\frac{2}{5}-O(\epsilon)\right)(\opttwos+\optthrees)$.
    \end{lemma}
    \begin{proof}
        Notice that since $v_1>1/6$, we have $\vtwos+\vthrees \le 1-v_1\le 5/6$. Then, similar to the proof of \Cref{lem:optgslb}, we obtain
        \begin{align*}
            \alg \ge& (1-O(\epsilon))\cdot \max\left\{\frac{\opttwos}{3\vtwos},\frac{\optthrees}{3\vthrees}\right\} \ge (1-O(\epsilon))\frac{\opttwos+\optthrees}{3(\vtwos+\vthrees)}\\ \ge& \left(\frac{2}{5}-O(\epsilon)\right)(\opttwos+\optthrees).
        \end{align*}
    \end{proof}

    \begin{claim}
    \label{cla:19by4}
        $\displaystyle\opt \le \left(\frac{19}{4}+O(\epsilon)\right)\alg$.
    \end{claim}
    \begin{proof}
        From \Cref{thm:structuredpackingalgo}, we have
        \begin{equation}
        \label{eqn:5a}
            6\alg \ge (2-O(\epsilon))(\opt_2 + \opt_3).
        \end{equation}

        Next, from \Cref{thm:big&largealgo}, we get
        \begin{equation}
        \label{eqn:5b}
            4\alg \ge (1-O(\epsilon))(4\opt_L + 2\opttwol + 2\optthreel). 
        \end{equation}

        From \Cref{cor:packingI1}, we get
        \begin{equation}
        \label{eqn:5c}
            4\alg \ge (4-O(\epsilon))\opt_1.
        \end{equation}

        Finally, from \Cref{lem:optgslbcase2}, we get
        \begin{equation}
        \label{eqn:5d}
            5\alg \ge (2-O(\epsilon))(\opttwos + \optthrees).
        \end{equation}

        Adding \eqref{eqn:5a}, \eqref{eqn:5b}, \eqref{eqn:5c} and \eqref{eqn:5d}, we obtain
        \[ 19\alg \ge (1-O(\epsilon))(4\opt_1 + 4\opt_2 + 4\opt_3 + 4\opt_L) = (4-O(\epsilon))\opt,\]
        and are done.        
    \end{proof}




    \item $v_1 \le 1/6:$ In this case, we have the following lower bound on $\alg$.


    \begin{claim}
    \label{cla:139by29}
        $\displaystyle\opt \le \left(\frac{139}{29}+O(\epsilon)\right)\alg$.
    \end{claim}
    \begin{proof}
        From \Cref{thm:structuredpackingalgo}, we get
        \begin{equation}
        \label{eqn:6a}
            312\alg \ge (104-O(\epsilon))(\opt_2 + \opt_3).
        \end{equation}

        From \Cref{thm:big&largealgo}, we have
        \begin{equation}
        \label{eqn:6b}
            104\alg \ge (1-O(\epsilon))(75\opt_L + 29\optonel).
        \end{equation}

        Again, from \Cref{cor:combinepacking}, we have 
        \begin{equation}
        \label{eqn:6d}
            116\alg \ge (1-O(\epsilon))(87\optonel + 116\optones + 12(\opttwot + \optthreet)+ 12\optlt).
        \end{equation}

        Finally, from \Cref{cor:2dknapsackcontainer}, we get
        \begin{equation}
        \label{eqn:6e}
            24\alg \ge (1-O(\epsilon))(12\optlh + 12\opttwoh + 12\optthreeh).
        \end{equation}

        Adding \eqref{eqn:6a}, \eqref{eqn:6b}, \eqref{eqn:6d} and \eqref{eqn:6e}, we get
        \[ 556\alg \ge (1-O(\epsilon))(116\opt_1 + 116\opt_2 + 116\opt_3 + 116\opt_L) = (116-O(\epsilon))\opt,\]
        and the claim follows.        
    \end{proof}

    From Claims \ref{cla:9by2}, \ref{cla:37by8}, \ref{cla:151by32}, \ref{cla:19by4}, \ref{cla:139by29}, we obtain an worst case approximation ratio of $\left(\frac{139}{29}+O(\epsilon)\right)$. This bound cannot be improved by the current analysis, since we have the following tight case.

    $\displaystyle\opt_1 = \frac{23}{139}\opt, \opt_2=\opt_3 = \frac{87}{278}\opt$, $\displaystyle\opt_L = \frac{29}{139}\opt$,
    
    $\displaystyle\optonel = 0$, $\displaystyle\optones = \frac{23}{139}\opt$,
    
    $\displaystyle\opttwol = 0$, $\displaystyle\opttwos = \frac{87}{278}\opt$,
    
    $\optthreel = 0$, $\displaystyle\optthrees = \frac{87}{278}\opt$,
    
    $\opttwot = 0$, $\displaystyle\opttwoh = \frac{87}{278}\opt$,
    
    $\displaystyle\optthreet = \frac{40}{139}\opt$, $\displaystyle\optthreeh = \frac{7}{278}\opt$,
    
    $\displaystyle\optlt = \frac{18}{139}\opt$, $\displaystyle\optlh = \frac{11}{139}\opt$

\end{enumerate}

\section{Special cases of \tdk}
In this section, we give improved approximation algorithms for \tdk~for two special cases. The first case is when all input items have the same profit, and therefore the problem reduces to packing the maximum number of items into the knapsack. The second interesting case is when the profit of an item is equal to its volume. The goal in this case is to maximize the total packed volume inside the knapsack.

\subsection{Cardinality case}
\label{sec:nonrot-cardinality}
We show the following result.

\begin{theorem}
    \label{thm:3dkc}
   For any constant $\eps>0$, there is a polynomial-time $(17/4 +\eps)$-approximation algorithm for \tdk~when all items have the same profit. 
\end{theorem}

Since each item in $L$ has a volume of at least $1/\mu^3$, the number of such items in OPT is bounded by $1/\mu^3$. Hence if OPT has at least $\frac{1}{\epsilon\mu^3}$ items, we can discard all items of $L$ by losing a profit of at most $\epsilon \opt$.

\begin{lemma}
\label{lem:largedelete}
    If $|\OPT|> \frac{1}{\epsilon\mu^3}$, the items of $\OPT\cap L$ have a profit of at most $\epsilon \opt$. 
\end{lemma}

As before, we let $\alg$ denote the maximum profit of a guessable container packing. We show the following bound on $\alg$ which, together with \Cref{thm:ptasboxpacking}, completes the proof of \Cref{thm:3dkc}.

\begin{lemma}
    $\displaystyle\opt \le \left(\frac{17}{4}+\epsilon\right)\alg$.
\end{lemma}
\begin{proof}
    If $|\OPT|\le \frac{1}{\epsilon\mu^3}$, we create \kvnnew{a Stack container for each item in $\OPT$}. Otherwise, using \Cref{lem:largedelete}, we have $\opt_L \le \epsilon\opt$, and thus $\opt_1+\opt_2+\opt_3 \ge (1-\epsilon)\opt$. From \Cref{thm:structuredpackingalgo}, we have 
    \begin{equation}
    \label{eqn:cardA}
        12\alg \ge (4-O(\epsilon))(\opt_2 + \opt_3).
    \end{equation}
     Again from \Cref{thm:big&largealgo}, we get 
     \begin{equation}
     \label{eqn:cardB}
         \alg \ge (1-O(\epsilon))\optonel.
     \end{equation}

     Now, one of $v_1, v_2$ or $v_3$ must be at most $1/3$; wlog assume that $v_1 \le 1/3$. Then from \Cref{cor:packingI1}, we have that
     \begin{equation}
     \label{eqn:cardC}
         4\alg \ge (1-O(\epsilon))(3\optonel + 4\optones).
     \end{equation}
    Adding \eqref{eqn:cardA}, \eqref{eqn:cardB} and \eqref{eqn:cardC}, we obtain 
    \[ 17\alg \ge (1-O(\epsilon))(4\opt_1+4\opt_2+4\opt_3)\ge (4-O(\epsilon))\opt,\]
    which completes the proof of the lemma.
\end{proof}

\mt{We further note that this approximation ratio cannot be improved by the current analysis. A tight instance is given by $\opt_1=\frac{5}{17}\opt,$ $\opt_2=\opt_3=\frac{6}{17}\opt$ with $\opt_{1\ell}=\opt_{2\ell}=\opt_{3\ell}=\frac{4}{17}\opt,$ $\opt_{1s}=\frac{1}{17}\opt$ and $\opt_{2s}=\opt_{3s}=\frac{2}{17}\opt.$}


\subsection{Uniform profit-density case}
\label{sec:nonrot-prof-eq-vol}
We show the following result.

\begin{theorem}
    \label{thm:3dkp}
   For any constant $\eps>0$, there is a polynomial-time $(4 +\eps)$-approximation algorithm for \tdk~when the profit of an item is equal to its volume. 
\end{theorem}

Similar to the previous subsection, we show the following bound on $\alg$, which, combined with \Cref{thm:ptasboxpacking}, proves the above theorem.

\begin{lemma}
\label{lem:3dkvolume}
    $\opt \le (4+O(\epsilon))\alg$.
\end{lemma}

We divide the proof of the lemma into two cases depending on the values of $v_1, v_2$ and $v_3$.

\subsubsection{None of \texorpdfstring{$v_1,v_2,v_3$}{v1, v2, v3} exceed \texorpdfstring{$1/4$}{1/4}}
In this case, \Cref{cor:packingI1} gives us $\alg \ge (1-O(\epsilon))\opt_i$, for $i\in [3]$. Also from \Cref{thm:big&largealgo}, we have $\alg \ge (1-O(\epsilon))\opt_L$. Hence,
\[ 4\alg \ge (1-O(\epsilon))(\opt_1 + \opt_2 + \opt_3 + \opt_L) = (1-O(\epsilon))\opt\]
and we are done.

\subsubsection{At least one of \texorpdfstring{$v_1,v_2,v_3$}{v1, v2, v3} exceeds \texorpdfstring{$1/4$}{1/4}}
W.l.o.g assume that $v_1 > 1/4$. We consider a maximal subset of $\OPT\cap I_1$ whose volume does not exceed $1/4$. Since the volume of each item in $\OPT\cap I_1$ is bounded by $\mu$, the volume of this subset is at least $\frac{1}{4}-\mu$. These items can be packed into a Stack and a Steinberg container using \Cref{lem:stackandsteincombine} by losing only an $\epsilon$-fraction of volume. The total volume packed is therefore at least $(1-\epsilon)\left(\frac{1}{4}-\mu\right) \ge \frac{1}{4}-O(\epsilon)$, from which \Cref{lem:3dkvolume} directly follows using the fact that $\opt\le 1$.  

\mt{Once again, the approximation ratio of $(4+O(\eps))$ cannot be improved by the current analysis. A corresponding instance is given by $\opt_1=\opt_2=\opt_3=\opt_L=\frac{1}{4}\opt$.}



\section{3D Knapsack with Rotations}
\label{sec:tdkr}
In this section, we present improved approximation algorithms for 3DK when the items are allowed to be rotated by $90$ degrees about the axes.
The previous best approximation ratio was $5+\eps$ by \cite{3d-knapsack}. We improve this to $30/7+\eps\approx4.286+\eps$.
We also give improved approximation ratios for the cardinality case ($24/7+\eps$; see \cref{sec:tdkrc}) as well as the special case of packing maximum volume ($3+\eps$; see \cref{sec:tdkrp}).
We follow the same theme as in the setup of non-rotations.
\subsection{\texorpdfstring{\LContainers}{L-Containers}}
\label{sec:Lcontainer}
For the case of rotations, we will define a new type of container called an \emph{\LContainer} in the same way we defined other types of containers in \cref{sec:cont-class}. 

\textbf{\LContainer.} An \LContainer{} $C$ of dimensions $w_C\times d_C\times h_C$ must satisfy the condition that $w_C\ge h_C$.
The capacity of $C$, $\mathtt{cap}(C)$, is given by $w_Ch_C-w_C^2/4$ (see \cref{pack-sheets} below).
Each item $i$ packed inside $C$ must have some orientation with width at least $w_C/2$, depth at least $d_C/2$, and height at most $\eps h_C$. 
If these conditions are satisfied, then we set $f_C(i)\coloneqq w_ih_i$ and $\infty$ otherwise.

\begin{remark}
    In the above definition of \LContainer, we note that the capacity is negative if $h_C<w_C/4$. We will however ensure that the \LContainers{} that we use do not have this property.
\end{remark}

Let $C$ be an \LContainer{} and let $T$ be a set of items such that $\sum_{i\in T}f_C(i)\le \mathtt{cap}(C)$.
The packing algorithm $\mathcal A_C$ to pack set $T$ in container $C$ is based on the following lemma.
\begin{lemma}
\label{pack-sheets}
Consider a rectangular region $R$ of length (horizontal dimension) $\ell$ and breadth (vertical dimension) $b$, such that $\ell\ge b$.
Let $S$ be a set of rectangles such that each rectangle has a length at least $\ell/2$ and breadth at most $\delta b$, where $\delta<1$ is a small constant. Suppose that the total area of rectangles in $S$ is at most $\ell b-\ell^2/4-3\delta b^2$.
Then, if rotations are allowed, we can pack the entire set $S$ in the specified rectangular region $R$.
\end{lemma}
\begin{proof}
We do away with the case when $b\le \ell/2$ because \mt{it is} easier. In this case, if we just stack up the rectangles on top of each other,
the total breadth is at most $b$. This is because the total area of the rectangles is at most $\ell b-\ell^2/4$ and
since the length of each rectangle is at least $\ell/2$; so the breadth of the stack is at most $2b-\ell/2\le b$ (since $\ell \ge 2b$).
Hence, the stack fits in the rectangular region $R$.

From now on, we assume that $b>\ell/2$. To pack $S$, we first sort the rectangles in non-increasing order of lengths. 
Then we start at the bottom left corner of $R$ and keep stacking them
one on top of each other to the maximum possible extent. These are shown as light grey rectangles in \cref{fig:stack-up}.
\begin{figure}[ht]
\begin{center}
\includesvg[width=0.4\textwidth,]{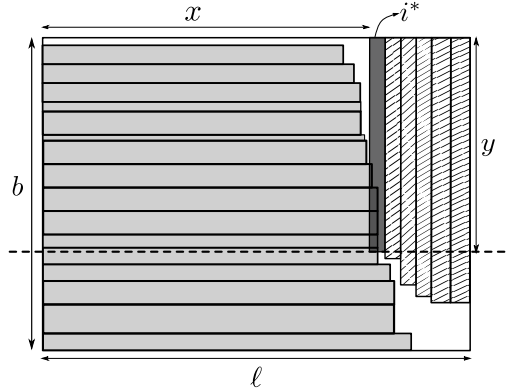}
\end{center}
\caption{The light grey rectangles are packed first. Then, the remaining rectangles (shaded) are rotated and packed from the right. The dark grey rectangle is the first one we could not pack (assuming contradiction).}
\label{fig:stack-up}
\end{figure}

If we exhaust the entire set $S$ by this procedure, then we are done. Otherwise, note that the first rectangle in the unpacked set must have a length of at most $b$. (Otherwise,
the total area of the stacked up rectangles is strictly more than $b(b-\delta b)=(1-\delta)b^2\ge \ell b-\ell^2/4-3\delta b^2$ (since $(b-\ell/2)^2\ge 0$), which is a contradiction.)
So, we rotate these remaining rectangles and align them vertically in the remaining space, starting from the top right corner of $R$, to the maximum extent possible.
These rectangles are denoted as hatched rectangles in \cref{fig:stack-up}.

We claim that we exhaust the
entire set $S$ in this manner. Suppose not. Let $i^*$, the dark grey rectangle in \cref{fig:stack-up}, denote the
first rectangle that we could not pack in this manner.
Let $x$ denote the horizontal distance from the left side of $R$ to the left side of $i^*$, if $i^*$ were packed. Also, let $y$
be the length of $i^*$ before rotating.
As shown in \cref{fig:stack-up}, we also draw a horizontal dashed line passing through the bottom edge of $i^*$, if it were packed.

The light grey rectangle intersecting the dashed line and all the light grey rectangles below it have a length of at least $x$.
So, the total area of these rectangles is at least $x(b-y)$. Every light grey rectangle above the dashed line has
length at least $y$ (since they have larger length compared to $i^*$)
and their total breadth is at least $y-2\delta b$. Hence, this amounts to an area of at least $y(y-2\delta b)$.
Finally, the total area of shaded rectangles is at least $y(\ell-x-\delta b)$. Therefore, the total packed area of all the rectangles
is at least
\begin{align}
x(b-y)+y(y-2\delta b)+y(\ell-x-\delta b)=bx+\ell y+y^2-2xy-3y\delta b.\label{eq:multi-fn}
\end{align}
\begin{claim}
\label{claim:ineq-x-y}
For $x\in(\ell/2,\ell]$ and $y\in(\ell/2,b]$, we have $bx+\ell y+y^2-2xy> b\ell-\frac{\ell^2}{4}$.
\end{claim}
\begin{proof}
Let $f(s,t)=bs+\ell t+t^2-2st$. Then
\begin{align*}
\frac{\partial f(s,t)}{\partial t}=\ell+2t-2s\quad\text{ and }\quad \frac{\partial^2 f(s,t)}{\partial t^2}=2.
\end{align*}
For $s\in(\ell/2,\ell]$ and $t\in(\ell/2,b]$, we have that $\partial f(s,t)/\partial t>2\ell-2b\ge  0$.
Hence, for a fixed $s$ and varying $t$, the function is strictly increasing since the partial derivative
of $f$ with respect to $t$ is always positive. So, the minimum of the function $f(s,t)$ in range
$[\ell/2,\ell]\times[\ell/2,b]$ occurs when $t=\ell/2$. But when $t=\ell/2$, we have $f(s,t)=3\ell^2/4-(\ell-b) s\ge 3\ell^2/4-(\ell-b)\ell=b\ell-\ell^2/4$.
\end{proof}
Thus, using the above claim and \cref{eq:multi-fn}, the total area packed is strictly more than $b\ell-\ell^2/4-3\delta b^2$ but this is a contradiction since the total area of
rectangles in $S$ is at most $b\ell-\ell^2/4-3\delta b^2$.
\end{proof}
Owing to the above lemma, the packing algorithm $\mathcal A_C$ to pack a set of items $T$ satisfying $\sum_{i\in T}f_C(i)\le \mathtt{cap}(C)$ can be devised. See \cref{alg:lcontainer} for a pseudocode.
\begin{algorithm}[h]
    \caption{$\mathcal A_C$ for an \LContainer{} given an input set $T$ satisfying $f_C(T)\le \mathtt{cap}(C)$}
    \begin{algorithmic}[1]
        \State Sort the items of $T$ in non-increasing order of profit/front area.
        \State Pick the largest prefix $T'$ whose total front area does not exceed $w_C h_C-h_C^2/4-3\eps h_C^2$.
        \State Sort the items in $T'$ in non-increasing order of widths.
        \State Start stacking them up one over the other to the maximum extent possible. (See \cref{fig:l-container-2}.)
        \State Rotate the remaining items (if any) so that the width and height are interchanged.
        \State Pack these remaining items starting from the right most face and top right corner of $C$. (See \cref{fig:l-container-2}.)
    \end{algorithmic}
    \label{alg:lcontainer}
\end{algorithm}

\begin{figure}
    \centering
    \includesvg[width=0.5\linewidth]{img/L-container.svg}
    \caption{An \LContainer}
    \label{fig:l-container-2}
\end{figure}

We will now show that $\mathcal A_C$ packs at least $(1-O(\eps))p(T)$.
Let us assume that $T'\subset T$, as otherwise, we are done. Since $T'$ is the largest prefix of $T$ with front area at most
$w_C h_C-h_C^2/4-3\eps h_C^2$, and since each item in $T$ has front area at most $\eps w_Ch_C$,
we obtain that the front area of $T'$ is at least $w_C h_C-h_C^2/4-3\eps h_C^2-\eps w_Ch_C$. Hence, we obtain that
\begin{align*}
    p(T')&\ge \left(\frac{\text{front area of }T'}{\text{front area of } T}\right)p(T)\\
        &\ge \left(\frac{w_C h_C-h_C^2/4-3\eps h_C^2-\eps w_Ch_C}{w_C h_C-h_C^2/4}\right)p(T)\\
        &\ge\left(1-\eps\left(\frac{3h_C^2+h_Cw_C}{w_Ch_C-h_C^2/4}\right)\right)p(T)\\
        &= \left(1-\eps\left(\frac{3h_C+w_C}{w_C-h_C/4}\right)\right)p(T').
\end{align*}
Since $w_C\ge h_C$, it can be verified that $(3h_C+w_C)/(w_C-h_C/4)\le 8$. Hence, we obtain that $p(T')\ge (1-8\eps)p(T)$.
\subsection{A packing lemma for rotations (Proof of \texorpdfstring{\cref{lem:3DRVolPack}}{Lemma 2.6})}
\label{sec:3dr-vol-pack}
Based on \cref{pack-sheets}, we show a general packing lemma in case of rotations that shows us how to pack a constant fraction
of volume of items, that are small in at least one dimension.

\tDRVolPack*
\begin{proof}
    Since rotations are allowed, we first orient each item $i\in T$ such that $w_i\le w,d_i\le w,h_i\le \eps^2 w$.
We first partition the set $T$ into two sets as follows.
\begin{itemize}
	\item $(T_{\ell})$ Items of width at least $w/2$, depth at least $w/2$ (recall that they have height at most $\eps^2 w$).
	\item {$(T_s)$ All other items (these items have height at most $\eps^2 w$
	and either width at most $w/2$ or depth at most $w/2$).}
\end{itemize}
We can pack the set $T_s$ in layers, as given by \cref{lem:steinberg-layers-height}.
We will also increase the height of the resultant packing by $2\eps w$
and then round it down to the nearest integral multiple of $\eps^2 w$.
This step is to make sure that height of the packing is discretized and also ensure that
the packing forms a Steinberg container.
Let this height be $h_s$. By \cref{lem:steinberg-layers-height}, we have that
\begin{align}
    h_s\le \left(\frac{3}{w^2}\right)v(T_s)+4\eps^2 w+2\eps w\le \left(\frac{3}{w^2}\right)v(T_s)+6\eps w.\label{eq:hs-upper-bound}
\end{align}
Depending on the value of $v(T_s)$, we consider two cases to pack $T_{\ell}$.

\textbf{Case 1: $\displaystyle v(T_s)\ge \left(\frac16-3\eps\right)v(B)$.}

In this case, we just stack up the items of $T_\ell$ above the packing of $T_s$ and show that the entire packing fits inside the box $B$.
Let $h_\ell$ denote the total height of items in $T_\ell$.
Since each item in $T_\ell$ has a base area strictly more than $w^2/4$,
\begin{align*}
    h_{\ell}&\le \frac{4}{w^2}v(T_\ell)\\
    &\le \frac{4}{w^2}\left(\left(\frac{7}{24}-5\eps\right)v(B)-v(T_s)\right).
\end{align*}
This implies that
\begin{align*}
    h_{\ell}+h_s&\le \frac{4}{w^2}\left(\left(\frac{7}{24}-5\eps\right)v(B)-v(T_s)\right)+\left(\frac{3}{w^2}\right)v(T_s)+6\eps w\\
            &\le \left(\frac76-20\eps\right)w+6\eps w-\frac{v(T_s)}{w^2}\\
            &\le \left(\frac76-20\eps\right)w+6\eps w-\left(\frac{1}{6}-3\eps\right)w\\
            &\le \left(1-11\eps\right)w.
\end{align*}
The packing of $T_{\ell}$ will be a Stack container.
We also round the height of this container up to the closest multiple of $\eps^2 w$. This will make the height of the entire packing at most
$\left(1-10\eps\right)w$.
Hence, the entire set $T$ can be packed into a Steinberg container and a Stack container.

\textbf{Case 2: $\displaystyle v(T_s)< \left(\frac16-3\eps\right)v(B)$.}

In this case, we use \cref{pack-sheets} to pack $T_\ell$ in the space above $T_s$.
As before, we will reserve a space of height $\eps w$ so that the height of the packing of $T_{\ell}$
can be rounded up to the closest multiple of $\eps^2 w$. This leaves a space of height 
$h'_{\ell}\coloneqq w-h_s-\eps w$ above the packing of $T_s$. We will show that $T_{\ell}$ can be packed here.
We know $v(T_\ell)\le \left(\frac{7}{24}-5\eps\right)v(B)-v(T_s)$.
Let us define the \emph{front area} of the set $T_\ell$ as $\sum_{i\in T_\ell}w_ih_i$.
Since the depth of each item in $T_\ell$ is at least $w/2$, we obtain that the total front area of the set $T_\ell$
is at most 
\begin{align*}
    A\coloneqq \frac{1}{w}\left(\frac{7}{12}-10\eps\right)v(B)-\frac{2}{w}v(T_s)=w^2\left(\frac{7}{12}-10\eps-\frac{2}{w^3}v(T_s)\right).
\end{align*}
We neglect the depth dimension of all the items in $T_\ell$ and treat them as a set of rectangles $R_\ell$. Similarly,
we neglect the depth dimension of the space of height $h'_{\ell}$ above the packing of $T_s$ and treat it as a rectangular region $\mathcal R$
of dimensions $w\times h'_{\ell}$.
We would now like to use \cref{pack-sheets} to pack $R_{\ell}$ in $\mathcal R$ as an \LContainer.
Let us look at the conditions of the lemma that need to be satisfied.
First, we have that $w\ge h'_{\ell}$. Then, each rectangle in $R_\ell$ has length at least $w/2$, which is ensured by the definition of $T_{\ell}$.
However, we only have that each rectangle in $R_{\ell}$ has breadth at most $\eps^2 w$, but we want it \mt{to }be at most $\eps' h'_{\ell}$ for some small $\eps'$.
To obtain this condition, we note that since $v(T_s)\le (\frac16-3\eps)v(B)$, by \cref{eq:hs-upper-bound}, it must be the case that
\begin{align}
    h'_{\ell}&=w-h_s-\eps w\nonumber\\    
        &\ge w-\left(\frac{3}{w^2}\right)v(T_s)-7\eps w\label{eq:hl-lower-bound}\\
        &\ge w-\left(\frac12-9\eps\right)w-7\eps w\nonumber\\
        &\ge \frac{w}{2}.\nonumber
\end{align}
Hence, we obtain that every rectangle in $R_{\ell}$ has breadth at most $(2\eps^2)h'_{\ell}$.
Finally, we need to show that the total area of $R_{\ell}$ is at most $wh'_{\ell}-w^2/4-3(2\eps^2)(h'_{\ell})^2$.
Indeed, we have
\begin{align*}
    h'_{\ell}w-\frac14w^2-6\eps^2(h_{\ell}')^2 &= (w-6\eps^2h_{\ell}')h_{\ell}'-\frac{w^2}{4}\\
    &\ge (w-6\eps^2w)(w-h_s-\eps w)-\frac{w^2}{4}\\
    &\ge (1-6\eps^2)(1-\eps)w^2-\frac{w^2}{4}-w h_s(1-6\eps^2)\\
    &\ge (1-7\eps)w^2-\frac{w^2}{4}-wh_s\\
    &\ge \left(\frac34-7\eps\right)w^2-w\left(\frac{3}{w^2}v(T_s)+6\eps w\right)\\
    &= \left(\frac34-7\eps\right)w^2-w^2\left(\frac{3}{w^3}v(T_s)+6\eps\right)\\
    &\ge \left(\frac34-13\eps-\frac{3}{w^3}v(T_s)\right)w^2.
\end{align*}
Since we have $v(T_s)\le (1/6-3\eps)w^3$, it can be verified that 
\begin{align*}
\left(\frac34-13\eps-\frac{3}{w^3}v(T_s)\right)\ge \left(\frac{7}{12}-10\eps-\frac{2}{w^3}v(T_s)\right).
\end{align*}
This implies that $wh'_{\ell}-w^2/4-3(2\eps^2)(h'_{\ell})^2\ge A$.
This shows that $T_\ell$ can be packed in an \LContainer{} of unit width and unit depth in the region above $T_s$. We use the reserved space of dimensions $1\times1\times\eps w$
to make sure that the height of this \LContainer{} is an integral multiple of $\eps^2 w$. Hence, we conclude that in this case, the set $T$ can be packed into a Steinberg container
and an \LContainer{}, whose heights are integral multiples of $\eps^2 w$.

As for the runtime, it is equal to the sum of times to pack $T_s$ and $T_\ell$. From \cref{lem:steinberg-layers-height}, we see that the time to pack $T_s$ is at most
$O(n\log^2 n/\log \log n)$. Packing $T_\ell$ takes $O(n\log n)$ time as we need to sort them in non-increasing order of heights and stack some items up and then arrange the remaining items side by side.
\end{proof}

\subsection{A PTAS for computing the optimal guessable container packing}
\begin{lemma}
\label{lem:gapwithrotation}
    There is an algorithm that returns a packing of profit at least~$(1-\epsilon)\alg$, when rotations of items are allowed.
\end{lemma}
\begin{proof}
    We guess the types and positions of the $O_{\epsilon}(1)$ containers in the most profitable container packing, and then define an instance of GAP as in the proof of \Cref{thm:ptasboxpacking}. The capacity of each container is set similarly as in the proof of \Cref{thm:ptasboxpacking}, and the profit of an item for any container is set as the item's actual profit. The items' sizes for the containers are also defined similarly, except that we need to account for the rotations of the items. Note that there are {\em six}  possible orientations for each item. We call an orientation to be \textit{feasible} for a container if it respects the container constraints as stated in \cref{sec:cont-class}. At a high level, while setting an item's size for a container, we look at all feasible orientations of the item, and choose the ``best'' one. This is described in more detail below.

    For a Stack container that stacks items along the height, an item's size is set as the minimum height among all feasible orientations of the item. Similarly, for an Area container that packs items using \NFDH along the bottom face of the container, the size of an item is set as the minimum base area of a feasible orientation, if one exists, and $\infty$ otherwise. For a Volume container, we set the size equal to the item's volume if there is a feasible orientation (where the dimensions of the item do not exceed an $\epsilon$-fraction of the container's dimensions) of the item, and $\infty$ otherwise. Similarly for Steinberg containers, an item's size is set equal to its volume, if there exist feasible orientations of the item inside the containers, and $\infty$ otherwise.
    For an \LContainer{}, the size of an item $i$ is the minimum front area (width$\times$height) among all feasible orientations of the item. From here on, we can construct the instance of GAP with the relevant functions $f_C$ for each possible container and solve it as we did in \cref{thm:ptasboxpacking}.
\end{proof}
\subsection{A \texorpdfstring{$(30/7+\eps)$}{(30/7 + ε)}-approximation algorithm for the general case}
\label{sec:rot-general-case}
We obtain the following result for \tdk~with rotations. 
\begin{theorem}
    \label{thm:3dkrot}
   For any constant $\eps>0$, there is a polynomial-time $(30/7+\eps)$-approximation algorithm for \tdkr. 
\end{theorem}

We will follow the same approach as the one we used in the setup when rotations were not allowed:
We obtain several lower bounds on $\optgs$ (the maximum profit of a guessable container packing) and then combine them.
Since \cref{lem:gapwithrotation} shows us how to pack a profit of at least $(1-\eps)\optgs$, it suffices to show a lower bound on $\optgs$ in terms of $\opt$.
One main difference is that we develop a lower bound based on \LContainers. 

As usual, let $I$ denote the input set, and let $\mu$ be a constant much smaller than $\eps$ (setting $\mu=\eps^{2^{1/\eps^2}}$ suffices).
The set $L$ contains items that have all three dimensions more than $\mu$.
Let $\OPT$ denote an optimal packing as well as the set of items in this optimal packing. 
Let $\OPT_L\coloneqq \OPT\cap L$.
Let $\OPT_1$ denote the set of items of $\OPT\setminus \OPT_L$ which are oriented in $\OPT$
such that their height is at most $\mu$. Let $\OPT_2$ denote the set of items of $\OPT\setminus (\OPT_L\cup\OPT_1)$ which are oriented in $\OPT$
such that their width is at most $\mu$ in $\OPT$. Finally, let $\OPT_3\coloneqq\OPT\setminus(\OPT_L\cup\OPT_1\cup\OPT_2)$.
Let $\opt_L=p(\OPT_L)$ and $\opt_i\coloneqq p(\OPT_i)$ for each $i\in\{1,2,3\}$.

Next, we discuss the lower bound on $\optgs$ based on \LContainers{}.
\subsubsection{Lower bound on \texorpdfstring{$\optgs$}{optgs} based on \texorpdfstring{\LContainers}{L-Containers}}
\label{sec:7by24guarantee}

\begin{lemma}
\label{lem:rotationalgo}
    We have that $\optgs\ge \left(\frac{7}{24}-O(\epsilon)\right)(\opt_1+\opt_2+\opt_3)$, when rotations of items are allowed.
\end{lemma}
\begin{proof}
    Recall that any item in $\OPT_1\cup \OPT_2\cup \OPT_3=I\setminus L$ has at least one of the dimensions not exceeding $\mu$.
    However, since rotations are allowed, we can assume that the items of $I\setminus L$ have height at most $\mu$.
    We sort the items of $I\setminus L$ in non-increasing order of their profit densities (the profit/volume ratio), and
    pick the largest prefix $P$ whose total volume does not exceed $7/24-5\sqrt{\mu}$. Since the volume of each item
    in $I\setminus L$ is at most $\mu$, we obtain that $v(P)\ge \min\{7/24-6\sqrt{\mu},v(I\setminus L)\}$.
    Since $P$ denotes the subset of $I\setminus L$ of highest profit density whose volume is at least $v(P)$, it implies that
    \begin{align*}
        p(P)&\ge \min\left\{\opt_1+\opt_2+\opt_3,\frac{v(P)}{v(\OPT_1\cup\OPT_2\cup\OPT_3)}(\opt_1+\opt_2+\opt_3)\right\}\\
            &\ge (7/24-6\sqrt{\mu})(\opt_1+\opt_2+\opt_3).
    \end{align*} 
    The last inequality follows since $v(\OPT_1\cup\OPT_2\cup\OPT_3)\le 1$ and $v(P) \ge \min\{v(I\setminus L),7/24 -6\sqrt{\mu}\}$.
    To conclude, \cref{lem:3DRVolPack} ensures that the set $P$ can be packed into a knapsack in two containers (one Steinberg container and one Stack container or \LContainer),
    both with unit width and unit depth, and heights as integral multiples of $\mu$.
\end{proof}

\subsubsection{Other lower bounds on \texorpdfstring{$\optgs$}{optgs}}
\label{sec:appxratio}
We can obtain more lower bounds on $\optgs$ in the case of rotations by performing similar restructurings as the ones in the case of non-rotations.
We divide $\OPT_1$ into two sets $\OPT_{1\ell}$ and $\OPT_{1s}$, where $\OPT_{1\ell}$ contains items in $\OPT_1$ which have both width and depth more than $1/2$, and $\OPT_{1s}=\OPT_1\setminus \OPT_{1\ell}$.
Define $\opt_{1\ell}\coloneqq p(\OPT_{1\ell})$ and $\opt_{1s}\coloneqq p(\OPT_{1s})$, so that $\opt_{1\ell} + \opt_{1s} = \opt_1$. Similarly, we define $\OPT_{2\ell}$ as the set of items in $\OPT_2$ with both width and height more than $1/2$, and $\OPT_{2s}$ as $\OPT_2\setminus \OPT_{2\ell}$.
Finally, we define $\OPT_{3\ell}$ as the set of items in $\OPT_3$ with both depth and height more than $1/2$, and $\OPT_{3s}$ as $\OPT_3\setminus \OPT_{3\ell}$.
The profits $\opt_{2\ell}, \opt_{2s}$ and $\opt_{3\ell}, \opt_{3s}$ are also defined appropriately.

We perform two different restructurings of the optimal packing into guessable container packings. 

\begin{claim}
    \label{claim:rot-number-theory}
    We have $\optgs\ge (1-O(\eps))(\opt_L+\max\{\opt_{1\ell},\opt_{2\ell},\opt_{3\ell}\})$.
\end{claim}
\begin{proof}
    We focus on the packing restricted to $\OPT_{1\ell}\cup\OPT_L$ in the optimal packing $\OPT$.
    We can perform the same restructuring procedure outlined in \Cref{sec:secondlb}. Thus, by \cref{thm:big&largealgo}, we obtain the claim.
\end{proof}

\begin{claim}
    \label{claim:rot-stein-gs}
    We have $\displaystyle \optgs\ge \left(\frac13-O(\eps)\right)(\opt_{1s}+\opt_{2s}+\opt_{3s})$.
\end{claim}
\begin{proof}
    Let $\OPT_s\coloneqq \OPT_{1s}\cup \OPT_{2s}\cup\OPT_{3s}$. In other words, $\OPT_s$ is the set of items in the optimal packing that have one dimension at most
$\mu$, and at least one of the other two dimensions not exceeding $1/2$. Thus $p(\OPT_s) = \opt_{1s}+\opt_{2s}+\opt_{3s}$. We sort items of $\OPT_s$ in non-increasing order
of profit density (profit/volume ratio) and select the maximum prefix whose volume is at most $\frac{1}{3}-2\mu$. Using \cref{lem:steinberg}, these selected items can be packed into a 
Steinberg container of height $1$. If by doing this, we are able to pack all items of $T$, we trivially obtain a profit of 
$\opt_{1s}+\opt_{2s}+\opt_{3s}$. Otherwise, the total volume of packed items is at least $\frac{1}{3}-3\mu$, and the profit obtained is at least 
\begin{align*}
    \frac{\frac{1}{3}-3\mu}{v(\OPT_s)}(\opt_{1s}+\opt_{2s}+\opt_{3s}) \ge \left(\frac{1}{3}-O(\epsilon)\right)(\opt_{1s}+\opt_{2s}+\opt_{3s}),
\end{align*}
where the last inequality follows since $\mu$ is at most $\eps$.
\end{proof}
We are ready to state the final lower bound on $\optgs$.
\begin{lemma}
\label{lem:containerrotation}
    We have that $\alg \ge (1-O(\epsilon))({\opt}/{6}+{\opt_L}/{3})$.
\end{lemma}
\begin{proof}
    We combine \cref{claim:rot-number-theory} and \cref{claim:rot-stein-gs}. By \cref{claim:rot-number-theory}, we have
    \begin{align*}
        \optgs&\ge (1-O(\eps))(\opt_L+\max\{\opt_{1\ell},\opt_{2\ell},\opt_{3\ell}\})\\
                &\ge (1-O(\eps))\left(\opt_L+\frac{\opt_{1\ell}+\opt_{2\ell}+\opt_{3\ell}}{3}\right).
    \end{align*}
    By \cref{claim:rot-stein-gs}, we have
    \begin{align*}
        \optgs\ge \left(\frac13-O(\eps)\right)(\opt_{1s}+\opt_{2s}+\opt_{3s}).
    \end{align*}
    Adding both the above equations, we obtain
    \begin{align*}
        \optgs&\ge(1-O(\eps))\left(\frac{\opt_L}{2}+\frac{\opt_{1\ell}+\opt_{1s}+\opt_{2\ell}+\opt_{2s}+\opt_{3\ell}+\opt_{3s}}{6}\right)\\
                &=(1-O(\eps))\left(\frac{\opt_L}{2}+\frac{\opt_1+\opt_2+\opt_3}{6}\right)\\
                &=(1-O(\eps))\left(\frac{\opt_L}{2}+\frac{\opt-\opt_L}{6}\right)\\
                &=(1-O(\eps))\left(\frac{\opt_L}{3}+\frac{\opt}{6}\right).
    \end{align*}    
    That ends the proof.
\end{proof}
\subsubsection{Final approximation ratio}

The following lemma, combined with \Cref{lem:gapwithrotation}, establishes \Cref{thm:3dkrot}.

\begin{lemma}
    There exists a guessable container packing of profit at least $\left(\frac{7}{30}-O(\eps)\right)\opt$, where the set of containers belongs to one of the following types.
    \begin{itemize}
        \item All \kvnnew{Stack} containers.
        \item One Stack container and one Steinberg container.
        \item One \LContainer{} and one Steinberg container.
    \end{itemize}
\end{lemma}
\begin{proof}
    By \cref{lem:containerrotation}, if $\opt_L/3+\opt/6\ge 7\opt/30$, then we are done.
    So, assume now that $\opt_L/3+\opt/6< 7\opt/30$. This implies that $\opt_L< \opt/5$.
    Recall that \cref{lem:rotationalgo} gives us the lower bound of $\optgs\ge (\frac{7}{24}-O(\epsilon))(\opt_1+\opt_2+\opt_3)$.
    This in turn implies that $\optgs\ge (7/24-O(\epsilon))(\opt-\opt_L)\ge (7/30-O(\eps))\opt$. 
\end{proof}
\mt{This approximation ratio is also tight by our current analysis. One such instance is given by $\opt_1=\frac{11}{15}\opt, \opt_2=\opt_3=\frac{1}{30}\opt, \opt_L=\frac{1}{5}\opt$. The respective partitions of $\opt_i$ are given by $\opt_{1\ell}=\opt_{2\ell}=\opt_{3\ell}=\frac{1}{30}\opt$ and thus $\opt_{1s}=\frac{7}{10}\opt$ and $\opt_{2s}=\opt_{3s}=0$.}

\subsection{Cardinality case}
\label{sec:tdkrc}
Based on \cref{lem:rotationalgo}, we can design a $(24/7+\eps)$-approximation algorithm in the cardinality case when rotations are allowed.

\begin{theorem}
    \label{thm:r3dkce}
   For any constant $\eps>0$, there is a polynomial-time $(24/7 +\eps)$-approximation algorithm for \tdkr~when all items have the same profit.  
\end{theorem}

In order to prove \Cref{thm:r3dkce}, we show the following lower bound on $\alg$, which together with \Cref{lem:gapwithrotation} completes the proof.

\begin{lemma}
    There exists a guessable container packing of profit at least $\left(\frac{7}{24}-O(\eps)\right)\opt$ when all items have the same profit, where the set of containers belongs to one of the following types.
    \begin{itemize}
        \item All \kvnnew{Stack} containers.
        \item One Stack container and one Steinberg container.
        \item One \LContainer{} and one Steinberg container.
    \end{itemize}
\end{lemma}
\begin{proof}
    We first consider the case when the number of items in $\OPT$ is at most $1/\mu^4$.
    Then, each item can be packed in a unique \kvnnew{Stack} container.

    Next, we consider the case when $\opt$ is at least $1/\mu^4$. Note that since $\opt_L\le 1/\mu^3$, we must have that $\opt_1+\opt_2+\opt_3\ge (1-\mu)\opt$.
    Thus, using \cref{lem:rotationalgo}, we obtain a guessable container packing of profit at least $(7/24-O(\eps))(\opt_1+\opt_2+\opt_3)$, which is at least $(7/24-O(\eps))\opt$.
    \cref{lem:rotationalgo} also ensures that this is a container packing of two containers (one Steinberg and one Stack/\LContainer{}).
\end{proof}
\mt{We note that this approximation ratio is tight once again. An instance where we  achieve the approximation ratio of $24/7+O(\eps)$ is given by $\opt_1=\opt,\opt_2=\opt_3=\opt_L=0$, with $\opt_{1\ell}=\frac{7}{24}\opt$ and $\opt_{1s}=\frac{17}{24}\opt$.}
\subsection{Uniform profit-density case}
\label{sec:tdkrp}
In this section, we show the following result.

\begin{restatable}{theorem}{profitequalsvolrot}
    \label{thm:r3dkrp}
   For any constant $\eps>0$, there is a polynomial-time $(3 +\eps)$-approximation algorithm for \tdkr~when the profit of an item is equal to its volume.  
\end{restatable}

As before, we let $\alg$ denote the maximum profit of a guessable container packing. We show the following lemma, which coupled with \Cref{lem:gapwithrotation} gives us the above theorem.

\begin{lemma}
    $\opt \le (3+O(\epsilon))\alg$.
\end{lemma}
\begin{proof}
    Trivially, we have $\alg \ge \opt_L$, since we can create a \kvnnew{Stack} container for every item in $L$. For brevity of notation, let $\text{OPT}_s = \text{OPT}\cap (I\setminus L)$. W.l.o.g.~assume that the items of $\text{OPT}_s$ are oriented such that $w_i \ge d_i \ge h_i$ holds for each $i\in \text{OPT}_s$. We now classify the items in $\text{OPT}_s$ as follows: let $I' = \{i \in \text{OPT}_s \mid w_i > 1/2, d_i > 1/2, h_i \le \mu\}$ and let $I'' = \text{OPT}_s \setminus I'$. Also, we define $\opt' = p(I')$ and $\opt'' = p(I'')$. We can obtain (guessable) container packings of (subsets of) $I'$ and $I''$ as described below.

    Ignoring the depths of the items in $I'$, we obtain a set of rectangles with one dimension at most $\mu$ and the other exceeding $1/2$. By \cref{pack-sheets}, these rectangles can be packed either exhaustively or to an extent such that the total (front) area packed is at least $\frac{3}{4}-3\mu$. Since the depth of each item is more than $1/2$, the total volume (and therefore profit) packed is at least $\min\left\{\frac{3}{8}-\frac{3\mu}{2},\opt'\right\}\ge\min\left\{\frac{1}{3}-\frac{3\mu}{2},{\opt'}\right\}$. Further, this is either a Stack or an \LContainer~whose dimensions are same as those of the knapsack.
\kvnr{I am changing this from $\frac38$ to $\frac13$.}
    Next, note that since the items were oriented such that $w_i \ge d_i \ge h_i$, it follows that each item in $I''$ has depth at most $1/2$. Hence by \Cref{lem:steinberg}, we can pack a subset of $I''$ with volume (profit) at least $\min\left\{\frac{1}{3}-2\mu, \opt''\right\}$ into a Steinberg container with dimensions identical to that of the knapsack.

    From the above discussion, we get that
    \[ \alg \ge \min\left\{\frac13-2\mu,\frac{{\opt_L}+{\opt'}+{\opt''}}{3}\right\} \ge \left(\frac{1}{3}-O(\epsilon)\right)\opt,\]
    since $\opt\le 1$, and thus the lemma follows.
\end{proof}

Finally, this approximation ratio is also tight by our current analysis. An instance to show this is given by $\opt'=\opt''=\opt_L=\frac{1}{3}\opt.$

\section{Hardness for Guessable Container Packings}
\label{sec:hardnesscontainer}
In this section, we show that, when rotations are not allowed, there exist input instances for which any guessable container packing can contain at most $1/3$ factor of the optimal profit.
This construction is in the same spirit as \cite{2dks-Lpacking}, where they showed for 2D packing one cannot obtain $(2-\eps)$-approximation using rectangular regions that are either packed using $\NFDH$ or stacks.
However, we use more types of containers and show we can not still beat factor 3 by $O(1)$ number of containers. 

For ease of presentation, we assume that the knapsack has side length $N=2^{m+1}$, where $m$ is an integer. 
We define an input instance consisting of $3m$ items, which can be partitioned into three sets $\{G_i\}_{i\in[m]},\{H_i\}_{i\in[m]},\{I_i\}_{i\in[m]}$ as follows.
\begin{itemize}
    \item For each $i\in[m]$, the item $G_i$ has both width and depth $N+1-2^{i-1}$ and height $2^{i-1}$.
    \item For each $i\in[m]$, the item $H_i$ has width $2^{i-1}$ and depth $N+1-2^{i-1}$ and height $N+1-2^{i}$.
    \item For each $i\in[m]$, the item $I_i$ has width and height $N+1-2^{i}$ and depth $2^{i-1}$.
\end{itemize}
Each item has a profit exactly equal to $1$.
\begin{remark}
    \label{remark:hardness}
    Each item $G_i$ has both width and depth more than $N/2$. Each item $H_i$ has both depth and height more than $N/2$. Each item $I_i$ has both width and height more than $N/2$.
\end{remark}
The next lemma shows that the above-defined input instance can be entirely packed in a knapsack, implying that the optimal profit is $3m$.
\begin{lemma}
    \label{lem:opt-hardness}
    The above-defined input instance has an optimal profit of $3m$.
\end{lemma}
\begin{proof}
    Consider $G_1,H_1,I_1$. All three items can be configured as shown in \cref{fig:3D-L}. This leaves a space of $(N-1)\times (N-1)\times (N-1)$
    in the remaining knapsack.
    We pack $G_2,H_2,I_2$ in the same configuration, leaving a space of dimensions $(N-1-2)\times (N-1-2)\times (N-1-2)$.
    We follow the same procedure for the remaining items as well. After packing $G_i,H_i,I_i$, a space of $(N+1-2^i)\times (N+1-2^i)\times (N+1-2^i)$ is left in the knapsack.
    Since, $2^{m+1}=N$, all the items in $\{G_i\}_{i\in[m]}\cup\{H_i\}_{i\in[m]}\cup\{I_i\}_{i\in[m]}$
    can be packed in the knapsack. Such optimal packing is shown in \cref{fig:opt-hardness}.
    \begin{figure}
        \centering
        \begin{minipage}{0.45\textwidth}
            \centering
            \includesvg[width=0.9\textwidth]{img/3d-L.svg} 
            \caption{Packing $G_1,H_1,I_1$.}
            \label{fig:3D-L}
        \end{minipage}\hfill
        \begin{minipage}{0.45\textwidth}
            \centering
            \includesvg[width=0.9\textwidth]{img/hardness-opt.svg} 
            \caption{Optimal packing of the entire instance.}
            \label{fig:opt-hardness}
        \end{minipage}
    \end{figure}
\end{proof}
\begin{theorem}
\label{thm:hardness}
    For any integer $m$ and constant $\eps>0$, for any container packing of the above input instance with profit at least $(1/3+\eps)(3m)$, the number of containers must be at least $\Omega(\eps\log N)$.
\end{theorem}
To show the above theorem, intuitively, we need to show that as we try to pack more profit, an average container can not hold many items.
To do this, we make use of the fact that items $G_i, H_i, I_i$ are skewed relative to each other. More formally, we show the following lemma.
\begin{lemma}
\label{lem:hardness}
    Consider a container packing $\mathcal C$ of the above input instance. Let $s_1$ denote the number of indices $i$ such that the pair $(G_i,H_i)$ is packed.
    Let $s_2$ denote the number of indices $i$ such that the pair $(H_i,I_i)$ is packed.
    Let $s_3$ denote the number of indices $i$ such that the pair $(G_i,I_i)$ is packed.
    Then, the number of containers is at least $(s_1+s_2+s_3)/3$.
\end{lemma}
We prove this lemma through a series of claims. From now on, let us consider an arbitrary container packing $\mathcal C$ of a subset of the above input instance.
\begin{claim}
    Every container in $\mathcal C$ \kvnnew{must} be a Stack container.
\end{claim}
\begin{proof}
    Due to \cref{remark:hardness}, there cannot be an Area or a Volume container or a Steinberg container in $\mathcal C$. 
\end{proof}
    
    
    
Now, we show three claims that will give us \cref{lem:hardness}. These claims and their proofs essentially follow the same theme, except that they correspond to different dimensions.
\begin{claim}
\label{lem:ghg}
    Consider any $i\in[m]$. Suppose both $G_i$ and $H_i$ are packed in $\mathcal C$, possibly in different containers. Then, no $G_j$ with $j>i$ can be packed in the container where $G_i$ is packed.
\end{claim}
\begin{proof}
    We first consider the case when $G_i,H_i$ are packed in the same container.
    We already know that this must be a Stack container. We further claim that $G_i$ and $H_i$ are stacked along the height; hence, any other items in the container must also be stacked along the height.
    Indeed, from \cref{remark:hardness}, we know that the depths of both $G_i,H_i$ are greater than $N/2$; so, they can not be stacked along the depth. The sum of their widths is given by $(N+1-2^{i-1})+(2^{i-1})=N+1$; so, they can not be stacked along the width either. Hence, we conclude that $G_i,H_i$ must be packed along their heights. However, the sum of heights of $G_i,H_i,G_j$ is given by $(2^{i-1})+(N+1-2^{i})+(2^{j-1})>N$.

    Now, consider the case when $G_i$ and $H_i$ are packed in different containers.
    By \cref{remark:hardness}, the sum of depths of both $G_i,H_i$ is greater than $N$. Similarly, the sum of widths of both containers must be at least $(N+1-2^{i-1})+2^{i-1}>N$.
    So, there must exist a line along the height that passes through both the containers containing $G_i,H_i$.
    Now, assume for the sake of contradiction that $G_j$ lies in the same container as $G_i$. \cref{remark:hardness} tells us that $G_i,G_j$ must be stacked along height.
    Hence, we conclude that the sum of heights of both the containers containing $G_i,H_i$ must be at least the sum of heights of $G_i,H_i,G_j$, which we showed to be greater than $N$ in the previous paragraph,
    and this is a contradiction.
\end{proof}
\begin{claim}
\label{lem:hih}
    Consider any $i\in[m]$. Suppose both $H_i$ and $I_i$ packed in $\mathcal C$. Then, no $H_j$ with $j>i$ can be packed in the container where $H_i$ is packed.
\end{claim}
\begin{proof}
    We first consider the case when $H_i,I_i$ are packed in the same container.
    We know that this is a Stack container. We further show that $H_i$ and $I_i$ are stacked along the width.
    Indeed, from \cref{remark:hardness}, we know that the heights of both $H_i,I_i$ are greater than $N/2$; so, they can not be stacked along the heights. The sum of their depths is given by $(N+1-2^{i-1})+(2^{i-1})=N+1$; so, they can not be stacked along their depths. Hence, $H_i,I_i$ must be packed along their widths. However, the sum of widths of $H_i,I_i,H_j$ is given by $(2^{i-1})+(N+1-2^{i})+(2^{j-1})>N$.

    Now, consider the case when $H_i$ and $I_i$ are packed in different containers.
    By \cref{remark:hardness}, the sum of heights of both $H_i,I_i$ is greater than $N$. Similarly, the sum of depths of both containers must be at least $(N+1-2^{i-1})+2^{i-1}>N$.
    So, there must exist a line along the width that passes through both the containers containing $H_i,I_i$.
    Now, assume for the sake of contradiction that $H_j$ lies in the same container as $H_i$. \cref{remark:hardness} tells us that $H_i,H_j$ must be stacked along width.
    Hence, we conclude that the sum of widths of both the containers containing $H_i,I_i$ must be at least the sum of heights of $H_i,I_i,H_j$, which we showed to be greater than $N$,
    which is a contradiction.
\end{proof}
\begin{claim}
\label{lem:gig}
    Consider any $i\in[m]$. Suppose both $G_i$ and $I_i$ are packed in $\mathcal C$. Then, no $G_j$ with $j>i$ can be packed in the same container containing $G_i$.
\end{claim}
\begin{proof}
    We first consider the case when $G_i,I_i$ are packed in the same container.
    We already know that this must be a Stack container. We further claim that $G_i$ and $I_i$ are stacked along the height; hence, any other items in the container must also be stacked along the height.
    Indeed, from \cref{remark:hardness}, we know that the widths of both $G_i,I_i$ are greater than $N/2$; so, they can not be stacked along the width. The sum of their depths is given by $(N+1-2^{i-1})+(2^{i-1})=N+1$; so, they can not be stacked along the depth either. Hence, we conclude that $G_i,I_i$ must be packed along their heights. However, the sum of heights of $G_i,I_i,G_j$ is given by $(2^{i-1})+(N+1-2^{i})+(2^{j-1})>N$.

    Now, consider the case when $G_i$ and $I_i$ are packed in different containers.
    By \cref{remark:hardness}, the sum of widths of both $G_i,I_i$ is greater than $N$. Similarly, the sum of depths of both containers must be at least $(N+1-2^{i-1})+2^{i-1}>N$.
    So, there must exist a line along the height that passes through both the containers containing $G_i,I_i$.
    Now, assume for the sake of contradiction that $G_j$ lies in the same container as $G_i$. \cref{remark:hardness} tells us that $G_i,G_j$ must be stacked along height.
    Hence, we conclude that the sum of heights of both the containers containing $G_i,I_i$ must be at least the sum of heights of $G_i,I_i,G_j$, which we showed to be greater than $N$,
    and this a contradiction.
\end{proof}
Now, let us define three sets
\begin{align*}
    &\mathcal P_{GH}=\{(G_i,H_i):i\in[m]\text{ and }G_i,H_i\text{ are packed in $\mathcal C$}\}\\
    &\mathcal P_{HI}=\{(H_i,I_i):i\in[m]\text{ and }H_i,I_i\text{ are packed in $\mathcal C$}\}\\
    &\mathcal P_{GI}=\{(G_i,I_i):i\in[m]\text{ and }G_i,I_i\text{ are packed in $\mathcal C$}\}
\end{align*}
In light of the above three claims, we can define a mapping $\pi:\mathcal P_{GH}\cup \mathcal P_{HI}\cup \mathcal P_{GI}\to \mathcal C$ as follows.
\begin{align*}
    \pi(G_i,H_i) &= \text{the container in which }G_i\text{ is packed,}\\
    \pi(H_i,I_i) &= \text{the container in which }H_i\text{ is packed,}\\
    \pi(G_i,I_i) &= \text{the container in which }G_i\text{ is packed.}
\end{align*}
\begin{claim}
\label{claim:mapping}
    The mapping $\pi$ is such that for any container $C$ in $\mathcal C$, at most three tuples in $\mathcal P_{GH}\cup \mathcal P_{HI}\cup \mathcal P_{GI}$ are mapped to $C$.
\end{claim}
\begin{proof}
    For the sake of contradiction, assume that there is a container $C$ in $\mathcal C$ such that there are four tuples in $\mathcal P_{GH}\cup \mathcal P_{HI}\cup \mathcal P_{GI}$ that are mapped to $C$.
    Then, one of the following cases must hold.
    \begin{itemize}
        \item There exist $i,j\in[m]$ such that $i\le j$ and $(G_i,H_i),(G_j,H_j)\in \mathcal P_{GH}$ and $\pi(G_i,H_i)=\pi(G_j,H_j)$.
        \item There exist $i,j\in[m]$ such that $i\le j$ and $(H_i,I_i),(H_j,I_j)\in \mathcal P_{GH}$ and $\pi(H_i,I_i)=\pi(H_j,I_j)$.
        \item There exist $i,j\in[m]$ such that $i\le j$ and $(G_i,I_i),(G_j,I_j)\in \mathcal P_{GH}$ and $\pi(G_i,I_i)=\pi(G_j,I_j)$.
    \end{itemize}
    However, each of the above cases must lead to a contradiction. 
    In the first case, we have $\pi(G_i,H_i)=\pi(G_j,H_j)$, it must be the case that $G_i,G_j$ are packed in the same container.
    This, however, leads to a contradiction to \cref{lem:ghg} because $G_i$ and $H_i$ are packed in the packing $\mathcal C$. The remaining cases follow analogously.
\end{proof}
\begin{proof}[Proof of \cref{lem:hardness}]
    Follows from \cref{claim:mapping}.
\end{proof}
\begin{proof}[Proof of \cref{thm:hardness}]
    Define an input instance containing $3m$ items as at the beginning of the section. 
    Let $s_1$ denote the number of indices $i$ such that the pair $(G_i,H_i)$ is packed.
    Let $s_2$ denote the number of indices $i$ such that the pair $(H_i,I_i)$ is packed.
    Let $s_3$ denote the number of indices $i$ such that the pair $(G_i,I_i)$ is packed.
    Then by \cref{lem:hardness}, the number of containers is at least $(s_1+s_2+s_3)/3$. Further, notice that
    the total number of items that can be packed is at most $m+2(s_1+s_2+s_3)$. Therefore, any container packing with $c$ containers can hold at most $m+6c$ items. For this quantity to be at least
    $(1/3+\eps)(3m)$, we require $c$ to be at least $m\eps/2$ which is $\Omega(\eps\log N)$.
\end{proof}
    

\section{Acknowledgments}
We are grateful to Lars Pr\"{a}del for helping us with some details of the 3D Strip Packing result.
We are also thankful to Marvin Lira for initial discussions.
Debajyoti and Sreenivas are grateful to the Google PhD Fellowship Program. Arindam Khan's research is supported in part by Google India Research Award, SERB Core Research Grant
(CRG/2022/001176) on “Optimization under Intractability and Uncertainty”, Ittiam Systems CSR grant, and the Walmart
Center for Tech Excellence at IISc (CSR Grant WMGT-23-0001). 
Klaus Jansen's and Malte Tutas' research has been supported by the Deutsche Forschungsgemeinschaft(DFG) - Project DFG JA 612 / 25-2 ``Fine-grained complexity and algorithms for scheduling and packing''.
\bibliography{ref}
\appendix

\section{Approximation Algorithms}
\label{sec:approx}

In this section, we discuss definitions related to approximation algorithms and approximation schemes.

For a maximization problem $\Pi$, an algorithm $A$ is an $\alpha$-approximation algorithm ($\alpha>1$), if for all input instances $I$ of $\Pi$, we have  $A(I)\ge \frac{1}{\alpha} \text{OPT}(I)$. 
A maximization problem $\Pi$ admits a polynomial-time approximation scheme (PTAS) if, for every constant $\eps>0$, there exists a $(1+\eps)$-approximation algorithm with running time $n^{f(1/\eps)}$, where $f$ is a function that depends only on $1/\eps$.
If the runtime of the PTAS is ${f(1/\eps)} \cdot n^{O(1)},$ where the exponent of $n$ is independent of $1/\eps,$ then we call it to be an efficient PTAS (or EPTAS). 
Furthermore, if the running time of the PTAS is $(1/\eps\cdot n)^{O(1)}$ it is referred to as a fully polynomial-time approximation scheme (FPTAS). 

When admitting some additional additive error in the approximation factor, we refer to algorithms as asymptotic approximation algorithms. For a maximization problem $\Pi$, an algorithm $A$ is an $\alpha$-asymptotic approximation algorithm ($\alpha>1$), if for all input instances $I$ of $\Pi$, we have  $A(I)\ge \frac{1}{\alpha} \text{OPT}(I)- o(\text{OPT}(I))$. 
Similar to non-asymptotic approximation algorithms, there exist asymptotic approximation schemes as well. A maximization problem $\Pi$ admits an asymptotic polynomial time approximation scheme (APTAS) if for every constant $\eps>0$, there exists a $(1+\eps)$-approximation algorithm with running time $n^{f(1/\eps)}$ where $f$ is a function that depends only on $1/\eps$. Furthermore, there are equivalent definitions of asymptotic efficient PTAS (AEPTAS) and asymptotic FPTAS (AFPTAS).

For the sake of completeness, we define the notion of an approximation algorithm for a minimization problem as well. 
For a minimization problem $\Pi$, an algorithm $A$ is an $\alpha$-approximation algorithm ($\alpha>1$), if for all input instances $I$ of $\Pi$, we have  $A(I) \le \alpha \text{OPT}(I)$. All further definitions follow analogously.


\section{Generalized Assignment Problem}
\label{subsec:GAP}
The Maximum Generalized Assignment Problem (GAP) is a combinatorial problem capturing the objectives of different assignment problems. In GAP, we are given a set of $k$ knapsacks with (possibly different) capacity constraints and a set of $n$ items that have sizes and profits that are possibly different depending on the knapsack they are placed in. The goal is to pack a subset of items into the bins such that the profit is maximized. Formally, when item $i$ is packed into knapsack $j$, it requires a size of $s_{i,j}\in \mathbb{Q}$ and has a profit of $p_{i,j}\in \mathbb{Q}.$ Knapsack $j$ has capacity $C_j \in \mathbb{Q}$.

In its most general case, GAP is known to be APX-hard, i.e., there does not exist a PTAS for it unless $P=NP$~\cite{ChekuriK05}. The best know approximation algorithms have an approximation ratio of $(1-1/e+\eps)$~\cite{FeigeV06,FleischerGMS11}. A special case of GAP is known as the Multiple Knapsack Problem (MKP). GAP becomes MKP when all profits and sizes are the same across all bins, i.e., $p_{i,j}=p_{i,k}$ and $s_{i,j}=s_{i,k}$ for all bins $j,k.$  The Multiple Knapsack Problem does admit a PTAS~\cite{ChekuriK05}. 

In this work, we require a PTAS for GAP with a constant number of knapsacks. This PTAS has been formally defined in~\cite{2dks-Lpacking} by extending previously known techniques.

\akr{added GAP theorem here.}
\begin{lemma}[\cite{2dks-Lpacking}]
    \label{lem:gap}
    For any fixed constant $\eps>0$, there exists an algorithm for Generalized Assignment Problem with $k$ knapsacks that runs in time $O((\frac{1+\eps}{\eps})^k n^{k/\eps^2+k+1})$ and returns a solution with profit at least $(1-3\eps)\opt$, where $\opt$ is the profit of the optimal solution. 
    
\end{lemma}


\section{Next-Fit-Decreasing-Height}
\label{subsec:nfdh}
Next-Fit-Decreasing-Height (\NFDH) is one of the oldest defined heuristics for packing rectangles in a bigger rectangular box, 
developed by Coffman, Garey, Johnson, and Tarjan \cite{coffman1980performance} in 1980.
It follows a simple procedure: We first sort the set of rectangles $I$ in non-increasing order of height. Then, we place the rectangles next to each other horizontally starting at the bottom left of the box.  We repeat this placing of rectangles until the next item would exceed the border of the strip. At this point, we place this item above the first item we placed and generate a new level. Then, we continue to place rectangles on this level. We repeat this procedure until all rectangles are placed or until we can not open a new level.

\NFDH can be adapted to $3$ dimensions as follows. Let $B$ denote a (three-dimensional) box that we would like to pack a given set of (three-dimensional) items into.
We first sort the items in the order of non-increasing heights.
Then, we pick the largest possible prefix of this order. We ignore the height of items in this prefix,
view them as rectangles,
and then pack them using (two-dimensional) \NFDH algorithm on the base of the box $B$.
This forms the first layer of our packing. Next, we shift the base of
$B$ along the height so that it now lies on the tallest item in the first layer. We continue packing the remaining items using (two-dimensional) \NFDH until we are not able to pack any other items or 
until the item set is exhausted. We refer to this algorithm as \tNFDH. See \cref{fig:3d-nfdh} for a picture.

\begin{figure}
    \centering
    \includesvg[width=0.7\linewidth]{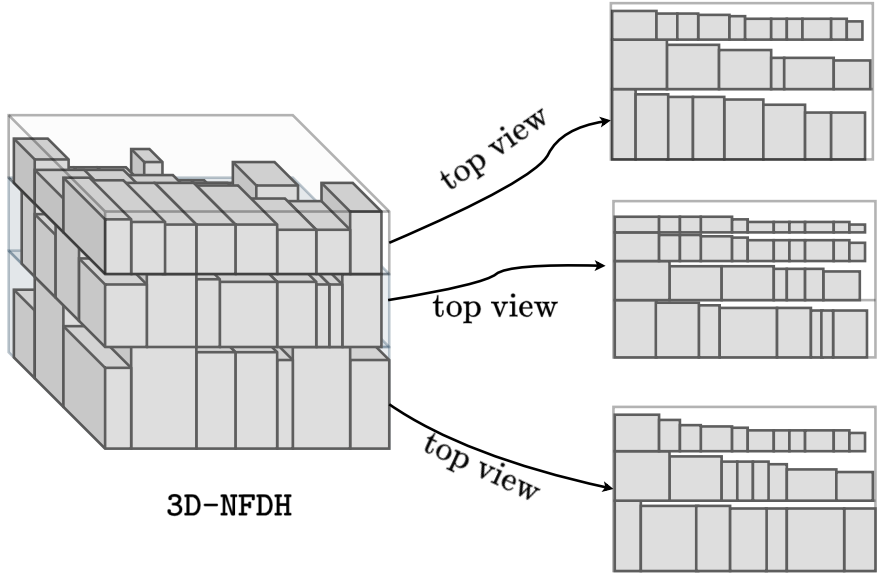}
    \caption{\tNFDH with three layers. Note how the heights of the layers keep decreasing. Further, it can be observed that each layer is obtained by (two-dimensional) \NFDH.}
    \label{fig:3d-nfdh}
\end{figure}

We are now ready to prove \cref{lem:3d-hfdh} from \cref{sec:subroutines}.
\threednfdh*
\begin{proof}
    Assume for the sake of contradiction that the set $T$ could not be packed in $B$ using \tNFDH. We first sort the items in $T$ in decreasing order of heights.
    Let $T'\subset T$ be the largest prefix that can be packed in $B$ using \tNFDH.
    Let $k$ be the number of layers we use to pack the set $T'$.
    (As a side note, $k\ge \ceil{1/\eps}$.)
    Further, let $h_j$ (resp. $h'_j$) denote the minimum (resp. maximum) height of an item in the $j^{\textrm{th}}$ layer.
    Let $\tilde h$ the height of the first item from the set, $T\setminus T'$, that could not be packed.
    Then, we have that
    \begin{align*}
        h'_1+h'_2+\dots+h'_k+\tilde{h}>h_B.
    \end{align*}
    Further, by \cref{lem:2d-hfdh}, we know that in each layer, the total base area covered is at least $(1-2\eps)w_Bd_B$.
    Since the height of each item in the $j^{\text{th}}$ layer is at least $h_j$, we have that the total volume packed is at least
    \begin{align*}
        (1-2\eps)w_Bd_B(h_1+h_2+\dots+h_{k-1}+h_k).
    \end{align*}
    Since we ordered the items in non-increasing order of height, we obtain that
    \begin{align*}
        v(T')&\ge(1-2\eps)w_Bd_B(h'_2+h'_3+\dots+h'_k+\tilde{h})\\
        &>(1-2\eps)w_Bd_B(h_B-h'_1)\\
        &\ge (1-2\eps)(1-\eps)w_Bd_Bh_B\\
        &\ge (1-3\eps)w_Bd_Bh_B\\
        &\ge v(T),
    \end{align*}
    which is a contradiction. \kvn{Hence, the entire set $T$ can be packed using \tNFDH{}. As for the runtime, observe that after each layer is packed,
    the items that will be packed in the next layer can be obtained using a binary search, which takes $O(\log n)$ time. Since each NFDH step takes $O(n_j\log n)$ time,
    where $n_j$ denotes the number of items in the $j^{\text{th}}$ layer, we obtain the specified runtime guarantee of $O(n\log^2 n)$ as well.}
\end{proof}

\end{document}